\documentclass[10pt,a4paper]{article}
\usepackage[latin1]{inputenc}
\usepackage{setspace}
\usepackage{amsmath}
\usepackage{amsfonts}
\usepackage{amssymb}
\usepackage{amsthm}
\usepackage{bm}
\usepackage{graphicx}
\usepackage{enumitem}
\graphicspath{ {./figures/} }

\usepackage[margin=1in]{geometry}
\usepackage[title]{appendix}
\usepackage{authblk}
\usepackage{xcolor}
\usepackage{times}
\usepackage{bm}
\usepackage{natbib}
\usepackage{graphicx}
\usepackage{float}
\usepackage{xcolor}
\usepackage{amsfonts}
\usepackage{bbm}
\usepackage{booktabs} 
\usepackage{caption}
\usepackage{subcaption}
\usepackage{enumitem}
\usepackage{mathrsfs} 
\usepackage{accents}
\usepackage{xr}
\usepackage{hyperref}

\usepackage[mathscr]{euscript}
\usepackage{bbm}

\theoremstyle{plain}
\newtheorem{theorem}{Theorem}
\newtheorem{lemma}{Lemma}
\newtheorem{corollary}{Corollary}

\theoremstyle{remark}

\newtheorem{condition}{Condition}

\usepackage{caption}
\usepackage[noend]{algpseudocode}

\newcommand{\indep}{\perp \!\!\! \perp}

\newcommand{\uQ}{\underaccent{\bar}{Q}}
\newcommand{\uq}{\underaccent{\bar}{q}}
\newcommand{\argmin}[1]{\underset{#1}{\operatorname{argmin}}\;} 

\usepackage{tikz}
\usetikzlibrary{positioning}

\tikzset{%
    every neuron/.style={
        circle,
        draw,
        minimum size=1cm
    },
    neuron missing/.style={
        draw=none, 
        scale=4,
        text height=0.333cm,
        execute at begin node=\color{black}$\vdots$
    },
}

\allowdisplaybreaks

\usepackage{natbib}
\title{Efficient Estimation under Data Fusion}
\author[*]{Sijia Li}
\author[$\dagger$]{Alex Luedtke}
\affil[*]{Department of Biostatistics, University of Washington}
\affil[$\dagger$]{Department of Statistics, University of Washington}

\date{}

\begin{document}
\maketitle

\begin{abstract}
We aim to make inferences about a smooth, finite-dimensional parameter by fusing data from multiple sources together. Previous works have studied the estimation of a variety of parameters in similar data fusion settings, including in the estimation of the average treatment effect and average reward under a policy, with the majority of them merging one historical data source with covariates, actions, and rewards and one data source of the same covariates. In this work, we consider the general case where one or more data sources align with each part of the distribution of the target population, for example, the conditional distribution of the reward given actions and covariates. We describe potential gains in efficiency that can arise from fusing these data sources together in a single analysis, which we characterize by a reduction in the semiparametric efficiency bound. We also provide a general means to construct estimators that achieve these bounds. In numerical experiments, we illustrate marked improvements in efficiency from using our proposed estimators rather than their natural alternatives. Finally, we illustrate the magnitude of efficiency gains that can be realized in vaccine immunogenicity studies by fusing data from two HIV vaccine trials.  
\end{abstract}

\section{Introduction}
\label{sec:introduction}
The rapid expansion of available data has facilitated the use of data fusion, which allows researchers to combine information from many data sources, each collected on a potentially distinct population at a different time, in order to obtain valid summaries of a target population of interest. In practice, data fusion often renders more relevant information or is less expensive than traditional analyses that only leverage a single data source. For example, technology companies integrate numerous unlabeled data with a small amount of labeled data to make accurate predictions, in a process known as semi-supervised learning 
\citep{chakrabortty2016robust}. In education, policy-makers leverage multiple datasets generated by different current policies to evaluate a new rule of interest in a fast, inexpensive, and effective way \citep{kallus2020optimal}. In genomics, integrating expression data, gene sequencing data, and network data gives a heterogeneous description of the gene and a distinct view of the underlying machinery of the cell \citep{lanckriet2004statistical}. In clinical trials, experimental data can be fused with observational data to evaluate a treatment regime on a different target population than the study population \citep[e.g.,][]{wedam2020fda}.
 
There are many recent works introducing statistical methods for particular data fusion problems. Many of them focus on bridging causal conclusions via data fusion as illustrated in the aforementioned clinical trial example. This is true, for example, in works on transportability \citep{pearl2011transportability,hernan2011compound,bareinboim2014transportability,stuart2015assessing,rudolph2017robust,dahabreh2019extending,dahabreh2019efficient,dong2020integrative}, re-targeting under covariate shifts \citep{kallus2020optimal}, using surrogate index to infer long-term outcomes \citep{athey2019surrogate}, and correcting external validity bias \citep{stuart2011use,mo2020learning}, in that all these research areas focus on bridging causal effects from a source population to a different target population. While these works considered merging two datasets only, \cite{dahabreh2019efficient} and \cite{lu2021you} considered bridging data from multiple trials to a target population and others have studied combining experimental data with multiple observational data sources in the presence of unmeasured confounding \citep{evans2018doubly,sun2018semiparametric}. Moreover, \cite{bareinboim2016causal} studied the identifiability results for a general causal parameter when multiple heterogeneous data sources are available. Data fusion is also used in non-causal problems. For example, semi-supervised learning \citep{chapelle2009semi,chakrabortty2016robust} represents another important application of data fusion. 

Due to the considerable number of open problems in this area, it is of interest to describe a general framework and approach that allows researchers to tackle data fusion problems in generality without limiting themselves to specific parameters, numbers of datasets, or data structures. In this paper, we will consider a general case where different data sources align with different parts of the distribution of the target population and derive efficient estimators based on all available data. We derive a key object needed to both quantify the best achievable level of statistical efficiency when data from multiple sources are fused together and to construct estimators that achieve these gains. Our results generalize previous works that study estimation of specific parameters in data fusion problems under nonparametric models by allowing for both consideration of general parameters and arbitrary semiparametric or nonparametric models. In addition to an example provided in Section~\ref{sec:examples}, in  Appendix \ref{additional examples} we illustrate the wide applicability of the proposed method by presenting five other examples. There we also provide a discussion of implementation, possible extensions, connections to missing data problems, and proofs.

\section{Notations and Problem Setup}\label{sec:notations and setup}

We begin by defining some notation. For a natural number $m$, we write $[m]$ to denote $\{1,\ldots,m\}$. For a distribution $\nu$, we let $E_{\nu}$ denote the expectation operator under $\nu$. Throughout we use $Z=(Z_1,\ldots,Z_d)$ to denote a random variable and, for $j\in [d]$, we let $\bar{Z}_j=(Z_1,\ldots,Z_j)$, where we use the convention that $\bar{Z}_0=\emptyset$. We use capital letters, such as $\bar{Z}_j$ and $S$, to denote random variables and the corresponding lowercase letters, such as $\bar{z}_j$ and $s$, to denote their realizations. In an abuse of notation, we condition on lowercase letter in expectations to indicate conditioning on the corresponding random variable taking a specific value: for example, $E_\nu(Z_2|z_1)=E_\nu(Z_2|Z_1=z_1)$. For any distribution $Q$ of $Z$ and $j\in [d]$, we will let $Q_j(\,\cdot\mid \bar{z}_{j-1})$ denote the conditional distribution of $Z_j\mid \bar{Z}_{j-1}=\bar{z}_{j-1}$. Similarly, for any distribution $P$ of $(Z,S)$, we will let $P_j(\,\cdot\mid \bar{z}_{j-1},s)$ denote the conditional distribution of $Z_j\mid \bar{Z}_{j-1}=\bar{z}_{j-1},S=s$. Here and throughout we suppose sufficient regularity conditions so that all such conditional distributions are well defined and that all discussed distributions of $Z_j\mid \bar{Z}_{j-1}=\bar{z}_{j-1}$ and $Z_j\mid \bar{Z}_{j-1}=\bar{z}_{j-1},S=s$ are defined on some common measurable space. We use $\rightarrow$ to specify the domain and codomain of an arbitrary function and use $\mapsto$ to denote the input and output of a function --- for example, if $f$ is the standard normal density function, then it would be accurate to write both $f : \mathbb{R}\rightarrow\mathbb{R}$ and $f : x\mapsto (2\pi)^{-1/2}e^{-x^2/2}$.

Suppose we have a collection of $k$ data sources and want to estimate an $\mathbb{R}^b$-valued summary $\psi(Q^0)$ of a target distribution $Q^0$ that is known to belong to a collection $\mathcal{Q}$ of distributions
of a random variable $Z=(Z_1, \ldots, Z_d)$, where $Z$ takes values in $\mathcal{Z}=\prod_{j=1}^d \mathcal{Z}_j$. The summary $\psi$ may only depend on a subset of the conditional distributions of $Z_j\mid \bar{Z}_{j-1}$. To handle such cases, we let $\mathcal{I}\subset [d]$ denote a set of irrelevant indices $j$ such that $\psi$ is not a function of the distribution of $Z_j\mid \bar{Z}_{j-1}$ --- more concretely, $\psi(Q)=\psi(Q')$ for all $Q,Q'\in\mathcal{Q}$ such that $Q_j=Q_j'$ for all $j\in [d]\backslash \mathcal{I}$. We do not require that $\mathcal{I}$ be the largest possible set of irrelevant indices --- this means that, for any parameter $\psi$, we can take $\mathcal{I}=\emptyset$, while, for certain parameters $\psi$, it will be possible to take $\mathcal{I}$ to be a nonempty set. To ensure that it makes sense to compare the distributions of $Z_j\mid \bar{Z}_{j-1}$ under different distributions $Q$ and $Q'$ in $\mathcal{Q}$, we assume here and throughout that all pairs of distributions in $\mathcal{Q}$ are mutually absolutely continuous. We let $\mathcal{J}=[d]\backslash \mathcal{I}$ denote the set of indices that may be relevant to the evaluation of $\psi$, termed the set of relevant indices.

Rather than observe draws directly from $Q^0$, we see $n$ independent copies of $X = (Z, S)$ drawn from some common distribution $P^0$, where $Z$ takes values in $\mathcal{Z}$ and $S$ is a categorical random variable denoting the data source has support $[k]$. 
The distribution $P^0$ is known to align with $Q^0$ in the sense described below, which makes it possible to relate the conditional distributions $P^0_j(\,\cdot\mid \bar{z}_{j-1},s)$ and $Q^0_j(\,\cdot\mid \bar{z}_{j-1})$. 
\begin{condition}
\label{cond:identifiability}
\textit{(Sufficient alignment)}\; For each relevant index $j\in \mathcal{J}$, there exists a known set $\mathcal{S}_j\subseteq [k]$ such that, for all $s\in \mathcal{S}_j$, both of the following hold:
\begin{enumerate}[label=\alph*.,ref=\alph*]
    \item\label{suff_overlap} \textit{(Sufficient overlap)}\; the marginal distribution of $\bar{Z}_{j-1}$ under sampling from $Q^0$ is absolutely continuous with respect to the conditional distribution of $\bar{Z}_{j-1}\mid S=s$ under sampling from $P^0$; and
    \item\label{suff_alignment} \textit{(Common conditional distributions)}\;  $P^0_j(\,\cdot\mid \bar{z}_{j-1},s)=Q^0_j(\,\cdot\mid \bar{z}_{j-1})$ $Q^0$-almost everywhere.
\end{enumerate}
\end{condition}
We will provide an example in Section~\ref{sec:examples} and another five examples in the  Appendix ~\ref{additional examples} where the above condition is plausible. Three of those examples represent generalizations of existing results, and, in all of those cases, a version of the above alignment condition was previously assumed. A detailed comparison between Condition~\ref{cond:identifiability} and identification conditions given in existing works is provided in  Appendix~\ref{additional examples}. We refer to $\mathcal{S}_j$, $j\in[d]$, as fusion sets and suppose they are known and prespecified in advance. As $Q^0$ is unknown beyond its membership to $\mathcal{Q}$, the above implies that $P^0$ is known to belong to the collection $\mathcal{P}$ of distributions $P$ with support on $\mathcal{Z}\times [k]$ for which there exists a $Q\in\mathcal{Q}$ such that, for all $j\in \mathcal{J}$ and $s\in \mathcal{S}_j$, the following analogues of Condition~\ref{cond:identifiability} hold: (a) the marginal distribution of $\bar{Z}_{j-1}$ under sampling from $Q$ is absolutely continuous with respect to the conditional distribution of $\bar{Z}_{j-1}\mid S=s$ under sampling from $P$, and (b) $P_j(\,\cdot\mid \bar{z}_{j-1},s)=Q_j(\,\cdot\mid \bar{z}_{j-1})$ $Q$-almost everywhere. Hereafter we refer to $\mathcal{P}$ and $\mathcal{Q}$ as models.

Condition~\ref{cond:identifiability}\ref{suff_overlap} ensures that, for $s\in\mathcal{S}_j$, null sets under the distribution of $\bar{Z}_{j-1}|S=s$ implied by $P^0$ are also null sets under the marginal distribution of $\bar{Z}_{j-1}$ implied by $Q^0$, which ensures that the conditional distribution $P^0_j(\,\cdot\mid \bar{z}_{j-1},s)$ appearing in Condition~\ref{cond:identifiability}\ref{suff_alignment} is uniquely defined up to $Q^0$-null sets.  It is worth noting that we have not assumed that the conditional distribution of $\bar{Z}_{j-1}\mid S=s$ under sampling from $P^0$ is absolutely continuous with respect to the marginal distribution of $\bar{Z}_{j-1}$ under sampling from $Q^0$, which allows $\bar{Z}_{j-1}$ to take values not seen in the target distribution when sampled from aligning data sources under $P^0$. Condition~1\ref{suff_alignment} implies exchangeability over data sources, namely that $Z_j \indep S \mid (\bar{Z}_{j-1}, S\in \mathcal{S}_j)$ for $j\in [d]$, which imposes a nontrivial conditional independence condition on the data-generating distribution when $\mathcal{S}_j$ is not a singleton for at least one $j$. This exchangeability condition is testable \citep{luedtke2019omnibus,westling2021nonparametric}. Nevertheless, we advocate choosing the fusion sets based on outside knowledge rather than via hypothesis testing to avoid challenges associated with post-selection inference.

Previous data fusion works have shown that variants of Condition~\ref{cond:identifiability} enable the identification of $\psi(Q^0)$ as a functional $\phi$ of the observed data distribution $P^0$ in particular problems \citep[e.g.,][]{rudolph2017robust,dahabreh2019efficient} and in general causal inference problems \citep[e.g.,][]{pearl2011transportability,bareinboim2016causal}. We will heavily rely on such an identifiability result to construct estimators of $\psi(Q^0)$, and so we present it explicitly here. To this end, we define a mapping $\theta : \mathcal{P}\rightarrow\mathcal{Q}$. In particular, for any $P\in\mathcal{P}$, we let $\theta(P)$ denote an arbitrarily selected distribution from the set $\mathcal{Q}(P)$ of distributions $Q\in\mathcal{Q}$ that are such that, for each $j\in \mathcal{J}$, $Z_j\mid \bar{Z}_{j-1}$ under sampling from $Q$ has the same distribution as $Z_j\mid \bar{Z}_{j-1},S\in\mathcal{S}_j$ under sampling from $P$. Because $P\in\mathcal{P}$, there must be at least one distribution in $\mathcal{Q}(P)$. Moreover, the value in $\mathcal{Q}(P)$ selected when defining $\theta(P)$ is irrelevant for our purposes since, as is evident below, our identifiability result only concerns the value of $\psi\circ \theta(P^0)$, and this value does not depend on the conditional distributions of $Z_j\mid \bar{Z}_{j-1}$ under $\theta(P^0)$ for irrelevant indices $j$.
\begin{theorem}\label{thm:identifiability}
Let $\phi=\psi\circ \theta$. Under Condition~\ref{cond:identifiability}, $\psi(Q^0)=\phi(P^0)$. 
\end{theorem}
Importantly, $\theta(P^0)$ can be evaluated without knowing the value of the true target distribution $Q^0$. Consequently, the above result shows that it is possible to learn the summary $\psi(Q^0)$ of the target distribution based only on the distribution of the observed data distribution $P^0$. This motivates estimating $\phi(P^0)$, and therefore $\psi(Q^0)$, based on a random sample drawn from $P^0$. Before presenting such estimation strategies, we will exhibit an example that fits within this data fusion framework.

\section{Example: longitudinal treatment effect}
\label{sec:examples}
While some clinical trials focus on evaluating a fixed treatment at a single time point, others involve longitudinal treatments that can vary over time. Let  $X = (U_1,A_1, \ldots, U_{T-1}, A_{T-1}, Y)$, where indices denote time, $A_t$ denotes the binary treatment at time $t$, and $U_t$ denotes the time-varying variable of interest at time $t$. Under this setup, we have $Z_1 = U_1$, $Z_2=A_1$, $Z_3=U_2$, \ldots, $Z_{2T-1}=Y$. We suppose that the final outcome of interest, $Y$, is real-valued. For ease of notation, we let $\Bar{H}_t=(U_1,A_1,\ldots, U_t)$ for each $t \in [T-1]$ denote the history up to time $t$. We consider three models $\mathcal{Q}$ for the unknown target distribution. The first is nonparametric, and consists of all distributions with some common support where treatment assignment satisfies the strong positivity condition that, conditionally on the past, each treatment is assigned with probability bounded away from zero. The second is semiparametric, and supposes that there is some unknown function $g : \prod_{j=1}^{2T-2}\mathcal{Z}_j\rightarrow\mathbb{R}$ such that the conditional distribution $Y \mid \bar{H}_{T-1}=h_{T-1},A_{T-1} =a_{T-1}$ is symmetric about $g(\bar{h}_{T-1},a_{T-1})$. When considering this semiparametric model, we suppose that, for each $Q\in\mathcal{Q}$, the conditional distribution $Q_{2T-1}$ has a corresponding conditional density $q_{2T-1}$ and that $q_{2T-1}(\,\cdot\mid \bar{H}_{T-1},A_{T-1})$ is almost surely differentiable. 
The third model we consider is also semiparametric and imposes that, under sampling from each $Q\in\mathcal{Q}$, $(\bar{H}_{T-1},A_{T-1})$ has support in $\mathbb{R}^p$ and there exists some vector of coefficients $\beta\in \mathbb{R}^p$ and error distribution $\tau_\alpha$ belonging to a regular parametric family $\{\tau_{\tilde{\alpha}} : \tilde{\alpha}\in\mathbb{R}^c\}$ of conditional distributions of a real-valued error $\epsilon$ given $(\bar{H}_{T-1},A_{T-1})$ such that $Y = \beta^\top \kappa(\bar{H}_{T-1},A_{T-1}) + \epsilon$, where $\kappa: \mathbb{R}^p \rightarrow \mathbb{R}^c$ is a known transformation of the history and treatment through time $T-1$ and $E_{\tau_\alpha}[\epsilon\mid \bar{H}_{T-1},A_{T-1}]=0$ almost surely.

\cite{hernan2020causal} discuss a range of causal parameters under such a longitudinal setting, including the mean outcome under a dynamic treatment regime, the parameters indexing a marginal structural mean model, and the average treatment effect of always being on treatment versus never being on treatment. For simplicity, here we focus on the final of these parameters. Under causal assumptions (Chapter 19.4 of \citealp{hernan2020causal}), this causal effect is identified with $\psi(Q)\equiv  E_{Q_1}\{L_1^1(\Bar{H_1})\} - E_{Q_1}\{L_1^0(\Bar{H_1})\}$, where, for $a\in\{0,1\}$, we define $L_T^a(\bar{h}_T)=y$ and, recursively from $t=T-1,\ldots,1$, define $L_t^a(\Bar{h}_t)=E_{Q_{2t+1}}\{L_{t+1}(\Bar{H}_{t+1})\mid \Bar{h}_t, A_t =a\}$. 
Because $\psi(Q)$ can be written as a function of $Q_1,Q_3,\ldots,Q_{2T-1}$, we see that we can take $\mathcal{I}=\{2,4,\ldots,2T-2\}$ in this example. This is consistent with the well known fact that the conditional average treatment effect does not depend on the treatment assignment probabilities, namely the distribution $Q_{2t}$ of $A_t|\bar{H}_t$ for $t\in [T-1]$.

We consider the scenario where we obtain data from $k$ sources. Some data sources contain observations from all $T$ time points --- for example, measurements of monthly CD4 count in HIV treatment trials or observational settings. Others may only contain such  measurements up to a time point $t<T$ such that $(U_1,A_1, \ldots, U_t, A_{t})$ is observed and $U_s$ and $A_s$ are missing for all $s>t$. We indicate missingness by writing that $U_s=A_s=\star$ in such cases. 
Such partial observations may still have valuable information, for example, about how longitudinal CD4 count responds to treatment shortly after the initiation of antiretroviral therapy. Under Condition~\ref{cond:identifiability}, 
$\psi(Q^0)$ can be identified as $\phi(P^0)  = E_{P^0}\{\Tilde{L}_1^1(\Bar{H_1})|S \in \mathcal{S}_1\}-E_{P^0}\{\Tilde{L}_1^0(\Bar{H_1})|S \in \mathcal{S}_1\}$, where, for $a\in\{0,1\}$, we define $\tilde{L}_T^a(\bar{h}_T)=y$ and, recursively from $t=T-1,\ldots,1$, define $\tilde{L}_t^a(\Bar{h}_t)=E_{P^0}\{\tilde{L}_{t+1}(\Bar{H}_{t+1})\mid \Bar{h}_t, A_t =a,S\in\mathcal{S}_{2t+1}\}$. When $T=1$, this estimand simplifies to the well-studied average treatment effect.

\section{Methods}
\label{sec:methods}
\subsection{Review of semiparametric theory}
\label{sec:efficiency theory}

We review some important aspects of semiparametric theory.  A more comprehensive review is given in  Appendix~\ref{sec:efficiency theory_appendix}. An estimator $\hat{\phi}$ of $\phi(P)$ is called asymptotically linear with influence function $D_P$ if it can be written as $\hat{\phi} - \phi(P) = n^{-1}\sum_{i=1}^n  D_P(X_i) + o_p(n^{-1/2})$, where $E_P\{D_P(X_i)\}=0$ and $\sigma_P^2\equiv E_P\{D_P(X_i)^2\}<\infty$. 
One reason such estimators are attractive is that they are consistent and asymptotically normal, in the sense that $\sqrt{n} \{ \hat{\phi} - \phi(P)\} \xrightarrow{d} N(0,\sigma_P^2)$ under sampling $n$ independent draws from $P$. This facilitates the construction of confidence intervals and hypothesis tests. If $\hat{\phi}$ is regular and asymptotically linear at $P$ \citep{bickel1993efficient}, then $\phi$ is pathwise differentiable and the influence function $D_P$ is a gradient of $\phi$, in the sense that, for all submodels $\{P^{(\epsilon)} : \epsilon\in [0,\delta)\}\in\mathscr{P}(P,\mathcal{P})$, $\frac{\partial}{\partial \epsilon} \phi(P_{\epsilon}) \mid_{\epsilon=0} = E_P\{D_P(X)h(X)\}$. The tangent set $\mathcal{T}(P,\mathcal{P})$ of $\mathcal{P}$ at $P$ is defined as the set of all scores of submodels in $\mathscr{P}(P,\mathcal{P})$. 
The canonical gradient $D_P^*$ corresponds to the $L_0^2(P)$-projection of any gradient $D_P$ onto the closure of the linear span of scores in $\mathcal{T}(P,\mathcal{P})$, where $L_0^2(P)$ is the Hilbert space of $P$-mean-zero functions, finite variance functions. 

One way to construct a regular asymptotically linear estimator with influence function $D_{P^0}$ is through one-step estimation \citep{bickel1982adaptive}. 
Given an estimate $\widehat{P}$ of $P^0$, 
the one-step estimator is given by $\hat{\phi} \equiv \phi(\widehat{P}) + \sum_{i=1}^n D_{\widehat{P}}(X_i)/n$. 
This estimator will be asymptotically linear with influence function $D_{P^0}$ if the remainder term $R(\widehat{P},P^0)\equiv \phi(\widehat{P}) - \phi(P^0) + E_{P^0}\{D_{\widehat{P}}(X)\}$ is $o_p(n^{-1/2})$ and the empirical mean of $D_{\widehat{P}}(X)-D_{P^0}(X)$ is within $o_p(n^{-1/2})$ of the mean of this term when $X\sim P^0$. The latter of these requirements will hold under an appropriate empirical process condition \citep{van1996weak}. Alternative approaches for constructing asymptotically linear estimators can be found in \cite{van2006targeted}, \cite{van2003unified}, and \cite{tsiatis2006semiparametric}. 

All of the results in this section extend naturally to the case where $\phi$ is $\mathbb{R}^b$-valued. In such cases, gradients (respectively, the canonical gradient) of $\phi$ are $\mathbb{R}^b$-valued functions whose $b$-th entry corresponds to a gradient (respectively, the canonical gradient) of the $b$-th coordinate projection of $\phi$. Estimators can similarly be constructed coordinatewise. Due to this straightforward extension from univariate to $b$-variate settings, the theoretical results in the next subsection focus on the special case where $b=1$.

\subsection{Derivation of canonical gradient of a general target parameter}
\label{sec:efficientIFs}
In this section, we provide approaches for obtaining the canonical gradient of $\phi$ in the model $\mathcal{P}$ implied by $\mathcal{Q}$ and the data fusion conditions.  
We will focus on settings where distributions in $Q$ can be separately defined via their conditional distributions, so that it is possible to modify a conditional distribution $Q_j$ under a distribution $Q\in\mathcal{Q}$, not modify any of the other conditional distributions $Q_{j'}$, $j'\not=j$, and still have it be the case that the resulting distribution belongs to $\mathcal{Q}$. This condition is formalized in the following.
\begin{condition}
\label{cond:var_indep}
    \textit{(Variation independence)}\; There exist sets $\mathcal{Q}_j$ of conditional distributions of $Z_j\mid \bar{Z}_{j-1}$, $j\in [d]$, such that $\mathcal{Q}$ is equal to the set of all distributions $Q$ such that, for all $j\in [d]$, the conditional distribution $Q_j$ belongs to $\mathcal{Q}_j$.
\end{condition}
The above condition is satisfied by the model $\mathcal{Q}$ described in each of our examples. It is also satisfied in many other interesting semiparametric examples, such as those where $E_{Q^0}(Z_1)$ or $E_{Q^0}(Z_2\mid Z_1)$ is known, but is not satisfied in some others, such as in cases where $E_{Q^0}(Z_2)$ is known --- in this case, knowing the marginal distribution $Q_1$ restricts the values that the conditional distribution $Q_2$ can take.

The upcoming results will provide forms of gradients of $\phi$ at $P^0$ in terms of gradients of $\psi$ at a generic distribution $\uQ^0\in \mathcal{Q}(P^0)$, where we recall that $\mathcal{Q}(P^0)$ is the set of distributions in $\mathcal{Q}$ whose relevant conditional distributions align with $P^0$. Since $Q^0\in\mathcal{Q}(P^0)$, all of these results are valid when $\uQ^0=Q^0$. However, since the distribution $Q^0$ is not generally identifiable from $P^0$ --- indeed, there may be no alignment for conditional distributions irrelevant to $\psi$ --- the particular value of $Q^0$ may be unknowable even given infinite data. In contrast, the set $\mathcal{Q}(P^0)$ would be knowable in such a setting. We, therefore, allow for the specification of an arbitrary distribution from the identifiable set $\mathcal{Q}(P^0)$ --- for example, it is always possible to take $\uQ^0=\theta(P^0)$, whose value depends only on $P^0$.

We require a strong overlap condition on the chosen $\uQ^0\in\mathcal{Q}(P^0)$ and the relevant conditional distributions of $P^0$. 
If $\uQ^0=Q^0$, then the upcoming condition strengthens the overlap condition (Condition~\ref{cond:identifiability}\ref{suff_overlap}) that was used to establish identifiability. To state this condition, for each $j\in \mathcal{J}$, we let $\lambda_{j-1}$ denote the Radon-Nikodym derivative of the marginal distribution of $\bar{Z}_{j-1}$ under sampling from $\uQ^0$ relative to the conditional distribution of $\bar{Z}_{j-1}\mid S\in\mathcal{S}_j$ under sampling from $P^0$.
\begin{condition}\label{cond:posPart}
\textit{(Strong overlap)}\; 
For each $j\in\mathcal{J}$, there exists $c_{j-1}>0$ such that $\uQ^0\{c_{j-1}^{-1}\le \lambda_{j-1}(\bar{Z}_{j-1})\le c_{j-1}\}=1$.
\end{condition}
Since $\bar{Z}_0=\emptyset$ almost surely under both of these distributions, $\lambda_{j-1}$ is the constant function that returns $1$ when $j=1$. Hence, the above condition only imposes a nontrivial requirement on $j\in\mathcal{J}\backslash \{1\}$. 

In the upcoming results, we suppose that the tangent set $\mathcal{T}(\uQ^0,\mathcal{Q})$ of $\mathcal{Q}$ at $\uQ^0$ is a closed linear subspace of $L_0^2(\uQ^0)$. We can therefore refer to $\mathcal{T}(\uQ^0,\mathcal{Q})$ as the tangent space without causing any confusion. The following lemma shows that, under the above conditions and the earlier stated data fusion condition, the pathwise differentiability of $\psi$ is equivalent to the pathwise differentiability of $\phi$. 
\begin{lemma}\label{lem:pd}
Suppose that Conditions~\ref{cond:identifiability}, \ref{cond:var_indep}, and \ref{cond:posPart} hold. Under these conditions, $\psi$ is pathwise differentiable at $\uQ^0$ relative to $\mathcal{Q}$ if and only if $\phi$ is pathwise differentiable at $P^0$ relative to $\mathcal{P}$.
\end{lemma}
The next result provides a means to derive gradients of $\phi$. 
\begin{theorem}\label{thm:grad}
Suppose that Conditions~\ref{cond:identifiability}, \ref{cond:var_indep}, and \ref{cond:posPart} hold and that $\psi$ is pathwise differentiable at $\uQ^0$ relative to $\mathcal{Q}$ with gradient $D_{\uQ^0}$. Under these conditions, the following function is a gradient of $\phi$ at $P^0$ relative to $\mathcal{P}$:
\begin{align}
    D_{P^0}(z,s)&\equiv \sum_{j\in\mathcal{J}} \frac{\mathbbm{1}(s\in\mathcal{S}_j)}{P(S\in\mathcal{S}_j)} \lambda_{j-1}(\bar{z}_{j-1})D_{\uQ^0,j}(\bar{z}_j), \label{eq:phiGrad}
\end{align}
where $D_{\uQ^0,j}(\bar{z}_j)\equiv  E_{\uQ^0}\{D_{\uQ^0}(Z)\mid \bar{Z}_j=\bar{z}_j\}-E_{\uQ^0}\{D_{\uQ^0}(Z)\mid \bar{Z}_{j-1}=\bar{z}_{j-1}\}$.
\end{theorem}
Given any gradient of $\phi$, the canonical gradient can be derived by projecting that gradient onto the tangent space of $\mathcal{P}$ at $P^0$. The form of this projection is provided in Lemma~\ref{lem:Tproj} in Appendix~\ref{app:projTanSpace}. Applying this projection to a gradient of the form in \eqref{eq:phiGrad} provides a form for the canonical gradient. In what follows we use $\Pi_{\uQ^0}\{\,\cdot\mid \mathcal{A}\}$ to denote the $L_0^2(\uQ^0)$-projection operator onto a subspace $\mathcal{A}$ of $L_0^2(\uQ^0)$.
\begin{corollary}\label{cor:canGrad}
Suppose the conditions of Theorem~\ref{thm:grad} hold and also that $\Pi_{\uQ^0}\{\lambda_{j-1}D_{\uQ^0,j}\mid \mathcal{T}(\uQ^0,\mathcal{Q})\}(\bar{z}_j) = \lambda_{j-1}(\bar{z}_{j-1})\Pi_{\uQ^0}\{D_{\uQ^0,j}\mid \mathcal{T}(\uQ^0,\mathcal{Q})\}(\bar{z}_j)$ $P^0$-a.s. for all $j\in\mathcal{J}$. Then, the canonical gradient of $\phi$ relative to $\mathcal{P}$ is
\begin{align}
    D_{P^0}^*(z,s)&\equiv  \sum_{j\in\mathcal{J}} \mathbbm{1}(\bar{z}_{j-1}\in\bar{\mathcal{Z}}_{j-1}^\dagger) \frac{\mathbbm{1}(s\in\mathcal{S}_j)}{P(S\in\mathcal{S}_j)} \Pi_{\uQ^0}\{r_j\mid \mathcal{T}(\uQ^0,\mathcal{Q})\}(\bar{z}_j), \label{eq:phiCanGrad}
\end{align}
where $\bar{\mathcal{Z}}_{j-1}^\dagger$ denotes the support of $\bar{Z}_{j-1}$ under sampling from $\uQ^0$ and $r_j\in L_0^2(\uQ_j^0)$ is such that $r_j(\bar{z}_j)=\lambda_{j-1}(\bar{z}_{j-1})D_{\uQ^0,j}(\bar{z}_j)$.
\end{corollary}
Corollary~\ref{cor:canGrad} generalizes previous works in data fusion, which mainly focus on particular parameters in nonparametric models, to arbitrary parameters in both nonparametric and semiparametric models for which $\mathcal{T}(Q^0,\mathcal{Q}_j)$ is sufficiently flexible. The condition on $\Pi_{\uQ^0}\{\lambda_{j-1}D_{\uQ^0,j}\mid \mathcal{T}(\uQ^0,\mathcal{Q})\}(\bar{z}_j)$ in Corollary~\ref{cor:canGrad} holds under arbitrary semiparametric $\mathcal{Q}$ when the fusion sets are nested, i.e. $\mathcal{S}_j \subseteq \mathcal{S}_m$ for any $m < j$, $m,j \in \mathcal{J}$.  Moreover, for each $j$ and any fusion sets, this condition holds when $\mathcal{Q}_j$ is nonparametric, has a symmetric conditional density as in Section~3, or satisfies the conditional moment restriction from Supplementary Appendix B. Furthermore, this condition holds provided the conditional distribution $Q_j$ of $Z_j\mid \bar{Z}_{j-1}$ satisfies a variation independence condition across possible values of $\bar{Z}_{j-1}$. Stated with $\bar{Z}_{j-1}$ discrete to avoid measurability considerations, this condition takes the form $\mathcal{Q}_j=\{Q_j : Q_j(\,\cdot\mid \bar{z}_{j-1})\in \mathcal{Q}_{j|\bar{z}_{j-1}}\,\forall \bar{z}_{j-1}\in\bar{\mathcal{Z}}_{j-1}\}$, where each $\mathcal{Q}_{j|\bar{z}_{j-1}}$ is a (possibly restricted) model for the distribution of $\bar{Z}_j\mid \bar{Z}_{j-1}=\bar{z}_{j-1}$.

We propose to construct either a one-step estimator using the canonical gradients derived using the procedure above. Under regularity conditions (outlined in Section~\ref{sec:efficiency theory}), the resulting estimator will be efficient among all regular and asymptotically linear estimators.

Because the canonical gradient is unique, the right-hand side of \eqref{eq:phiCanGrad} will be the same regardless of the chosen value of $\uQ^0\in\mathcal{Q}(P^0)$ that satisfies Condition~\ref{cond:posPart}. However, the calculations required to simplify that expression may differ. Indeed, that expression depends on $\uQ^0$ through the definitions of $r_j$ and through the projection operator $\Pi_{\uQ^0}\{\,\cdot\mid \mathcal{T}(\uQ^0,\mathcal{Q})\}$. 

Computing the projection in \eqref{eq:phiCanGrad} may be challenging in some semiparametric models, though there is substantial existing work providing the form of this projection in a variety of interesting examples \citep{pfanzagl1990estimation,bickel1993efficient,van2003unified,tsiatis2006semiparametric}. 
In contrast, computing this projection is necessarily trivial when $\mathcal{Q}$ is locally nonparametric, since in this case the tangent space of $\mathcal{Q}$ at $\uQ^0$ is equal to $L_0^2(\uQ^0)$ and the projection operator is the identity operator. Hence, in this special case, \eqref{eq:phiGrad} and \eqref{eq:phiCanGrad} are equal, and applying Theorem~\ref{thm:grad} to the one gradient $D_{\uQ^0}$ for $\psi$ that can possibly exist relative to a locally nonparametric model for $\mathcal{Q}$ necessarily yields the canonical gradient relative to $\mathcal{P}$. 
In semiparametric models, where there is more than one possible initial candidate gradient $D_{\uQ^0}$ to plug into \eqref{eq:phiGrad}, it is natural to wonder whether there is any such candidate for which \eqref{eq:phiGrad} and \eqref{eq:phiCanGrad} coincide. In general, this will fail to hold unless there is a gradient $D_{\uQ^0}$ of $\psi$ for which $z\mapsto \lambda_{j-1}(\bar{z}_{j-1})D_{\uQ^0,j}(\bar{z}_j)$ belongs to $\mathcal{T}(\uQ^0,\mathcal{Q})$ for all $j\in\mathcal{J}$. 
This does not hold in general. One example of a case where this typically fails to hold occurs in a model where it is known that $Z_j\perp \bar{Z}_{j-2}\mid Z_{j-1}$ for some $j\in\mathcal{J}$. 

In certain cases, there is a close connection between data fusion problems and missing data problems. In particular, if the fusion sets are nested such that $\mathcal{S}_d \subseteq \mathcal{S}_{d-1} \subseteq \ldots \subseteq \mathcal{S}_1$ with $\mathcal{S}_d$ nonempty, then the fusion framework can be framed as a monotone missing data problem. In this problem, the full data structure $Z^{F}$ is drawn from $Q^0$. The observed variable $Z$ is such that $Z_j=Z_j^F$ for all $j$ such that $S\in\mathcal{S}_j$ and $Z_j$ is coded as missing ($Z_j=\star$) otherwise. Under Condition~\ref{cond:identifiability}, this results in a missing data problem that satisfies coarsening at random \citep{heitjan1991ignorability}. Existing works provide general approaches to derive canonical gradients in these models and, when the fusion sets are nested, our results emerge as a special case of them. For example, our Theorem~\ref{thm:grad} can be derived as a special case of Theorem~1.1 of \cite{van2003unified} in this case. Moreover, the canonical gradient presented in Corollary~\ref{cor:canGrad} can be derived using a general strategy for deriving efficient estimators in monotone missing data problems \citep{robins1994estimation,tsiatis2006semiparametric}. To ease notation, when illustrating this we focus on the case where $\mathcal{J}=[d]$ and the support of $\bar{Z}_{j-1}\mid S\in\mathcal{S}_j$ under $P^0$ is equal to $\bar{\mathcal{Z}}_{j-1}^\dagger$. Letting $C=\max\{j : S\in\mathcal{S}_j\}$, with the maximum over the empty set being equal to $0$, it can be verified that the canonical gradient $D_{P^0}^*$ from Corollary~\ref{cor:canGrad} satisfies
\begin{align*}
    D_{P^0}^*(z,s)&= \frac{\mathbbm{1}(c=d)}{P^0(C\ge d)}\sum_{j=1}^d \Pi_{\uQ^0}\{r_j\mid \mathcal{T}(\uQ^0,\mathcal{Q})\}(\bar{z}_j) \\
    &\quad+ \sum_{\ell=1}^{d-1} \frac{\mathbbm{1}(c= \ell)-\mathbbm{1}(c\ge \ell)P^0(C=\ell\mid C\ge \ell)}{P^0(C\ge \ell+1)} \sum_{j=1}^{\ell} \Pi_{\uQ^0}\{r_j\mid \mathcal{T}(\uQ^0,\mathcal{Q})\}(\bar{z}_j).
\end{align*}
Noting that $f_{Q^0}(z):= \sum_{j=1}^d \Pi_{\uQ^0}\{r_j\mid \mathcal{T}(\uQ^0,\mathcal{Q})\}(\bar{z}_j)$ is a gradient of $\psi$ at $\uQ^0$ relative to $\mathcal{Q}$, the above identity shows that $D_{P^0}^*$ is the influence function of an augmented inverse probability weighted complete-case estimator derived from the full-data influence function $f_{Q^0}$ \citep[see Theorem~10.1 of][]{tsiatis2006semiparametric}. Consequently, the form of this canonical gradient could be identified or numerically approximated using the general strategies outlined in Chapter~11 of \cite{tsiatis2006semiparametric}. As is noted in \cite{tsiatis2006semiparametric}, the strategies outlined in that chapter are quite complex, often involving some form of numerical approximation, and so ``the actual implementation needs to be considered on a case-by-case basis.'' Corollary~\ref{cor:canGrad} demonstrates that, in the special case of data fusion with nested fusion sets, a simple form for the canonical gradient is available. Moreover, this gradient can be computed explicitly whenever the projections used to define it have a known form, which will typically be the case when a closed-form efficient estimator for $\psi(Q^0)$ is known in the non-data fusion setting where a random sample from $Q^0$ is directly observed. 

In settings where the fusion sets are not nested, it is not clear that it is even possible to write our data fusion problem as a missing data problem that satisfies coarsening at random. Hence, Corollary~\ref{cor:canGrad} presents what is, to our knowledge, the first general strategy for deriving the canonical gradient in data fusion problems. We conclude by noting that non-nested cases are far from the exception. Indeed, many previously studied data fusion problems involve non-nested fusion sets \citep[e.g.,][]{rudolph2017robust,westreich2017transportability,dahabreh2019extending,kallus2020optimal}.

\subsection{Canonical Gradient in Our Example}
\label{sec:estimation}

We now derive the canonical gradient of the longitudinal treatment effect described in each of the three semiparametric models described in Section~\ref{sec:examples}. 
An initial gradient $D_{Q^0}$ to plug into Theorem~\ref{thm:grad} can be found in Theorem~1 of \cite{van2012targeted} \citep[see also][]{bang2005doubly}. 
Following notations introduced in Section~\ref{sec:examples} and the results from Corollary~\ref{cor:canGrad}, we can use this initial gradient to show that the canonical gradient of $\phi$ under a locally nonparametric model is $D_{P^0}(x) = D^1_{P^0}(x)-D^0_{P^0}(x)$, where, for $a' \in \{0,1\}$ and letting $\bar{\mathcal{U}}^{\dagger}_{t}$ denote the support of $\bar{U}_{t}$ under sampling from $Q^0$, $D^{a'}_{P^0}=\sum_{t=1}^{T}D^{a'}_{P_{2t-1}^0}$ with
\begin{align}
   D^{a'}_{P_{2t-1}^0}(x) & \equiv \mathbbm{1}(\bar{u}_{t-1} \in \bar{\mathcal{U}}^{\dagger}_{t-1})\frac{\mathbbm{1}(s \in \mathcal{S}_{2t-1})}{\textnormal{pr}(S \in \mathcal{S}_{2t-1}) } \left\{\prod_{m=1}^{t-1} \frac{\mathbbm{1}(a_m=a')}{\textnormal{pr}(A_m=a'\mid \bar{u}_{m},\bar{A}_{m-1}=a', S\in \mathcal{S}_{2t-1})}\right\} \nonumber\\
    & \quad \cdot\left\{\prod_{m=1}^{t-1} \frac{ dP^0(u_m \mid  \bar{u}_{m-1},\bar{A}_{m-1} = a',  S\in \mathcal{S}_{2m-1}) }{dP^0(u_m \mid  \bar{u}_{m-1},\bar{A}_{m-1} = a', S\in \mathcal{S}_{2t-1})}  \right\}\{\Tilde{L}^{a'}_{t}(\Bar{h}_{t},s) - \Tilde{L}^{a'}_{t-1}(\Bar{h}_{t-1},s)\}, \label{eq:IF_long_nonpar}
\end{align}
where we abuse notation and write $\bar{A}_{m-1} = a'$ to mean that $A_j = a'$ for all $j \in [m-1]$. Compared to the canonical gradient $D_{Q^0}$ under no data fusion, the above uses all available data sources that are valid, as reflected by the indicator terms, and then corrects for the variable shifts between the target population and the observed as reflected by the first two terms inside the curly brackets.

For the symmetric location semiparametric model described in Section~\ref{sec:examples}, we provide the canonical gradient of $\phi$ in Appendix~\ref{example:longitudinal treatment effect appendix}. For the linear semiparametric model described in Section~\ref{sec:examples}, when the fusion sets are nested, the canonical gradient of $\phi$ takes the form of $D^{\dagger}_{P^0} = D^{\dagger 1}_{P^0}-D^{\dagger 0}_{P^0}$, where
\begin{align}
    D^{\dagger a'}_{P^0}(x) & = D^{a'}_{P^0}(x) - D^{a'}_{P_{2T-1}^0}(x) +\mathbbm{1}(\bar{u}_{T-1} \in \bar{\mathcal{U}}^{\dagger}_{T-1})\frac{\mathbbm{1}(s \in \mathcal{S}_{2T-1})}{P^0(S \in \mathcal{S}_{2T-1})} \nonumber \\
    &\quad \cdot  \left[E_{\uQ^0}\left\{\ell(X)\ell(X)^\top\right\}^{-1} E_{\uQ^0}\left\{\ell(X) \lambda_{2T-2}(\bar{H}_{T-1},A_{T-1}) D^{a'}_{\uQ^0_{2T-1}}(X)\right\}\right]^\top \ell(x),\label{eq:IF_long_tau}
\end{align}
where $\uQ^0$ is a generic element of $\mathcal{Q}(P^0)$, $\ell = (\ell_{\beta}, \ell_{\alpha})$ with $\ell_{\beta}(z) \equiv \nabla_{\beta} \log \tilde{\tau}_\alpha\{y - \beta^{\top}\kappa(\bar{h}_{T-1}, a_{T-1})\mid \bar{h}_{T-1}, a_{T-1}\}$ and $\ell_{\alpha}(z) \equiv \nabla_{\alpha} \log \tilde{\tau}_\alpha\{y - \beta^{\top}\kappa(\bar{h}_{T-1}, a_{T-1})\mid \bar{h}_{T-1}, a_{T-1}\}$, where $\tilde{\tau}_\alpha$ denotes the conditional density function of the error distribution $\tau_{\alpha}.$ The key distinction between \eqref{eq:IF_long_tau} and \eqref{eq:IF_long_nonpar} lies in the subspace $\mathcal{T}(\uQ^0_{2T-1},\mathcal{Q}_{2T-1})$, where the projection used in \eqref{eq:IF_long_tau} can be derived via a score function argument for finite-dimensional parameters. A detailed roadmap is at  Appendix~\ref{example:longitudinal treatment effect appendix}.

\section{Experiment and data illustration}
\subsection{Experiment}
\label{sec:simulations}
We simulated $k=9$ data sources with $T=4$ and fixed data sizes as specified in Table~\ref{tab:Long_setup}. The variable $\bar{U}_{3}$ of the target population follows a multivariate normal distribution as specified in Table~\ref{tab:Long_setup_distributions}, while treatments $\bar{A}_3$ are independent Bernoulli($0.5$). The outcome variable $Y$ is such that $Y = \beta^{\top}\kappa(\bar{H}_{3}, A_{3}) + \epsilon $ and the heteroskedastic error $\epsilon$ satisfies $\epsilon\mid \bar{H}_{3}=\bar{h}_3, A_{3}=a_3\sim \tau_{\alpha}(\,\cdot\mid \bar{h}_3,a_3)$, where, for $\tilde{\alpha}>0$, $\tau_{\tilde{\alpha}}(\,\cdot\mid \bar{h}_3,a_3)$ denotes the distribution of $\tilde{\alpha} u_3$ times a random variable following a student's t-distribution with 3 degrees of freedom. The indexing parameter $\alpha$ equals 0.1 and the values of $\beta$ and the form of $\kappa$ are specified in Appendix~\ref{app:sim_long}. The underlying true distribution belongs to the two models mentioned in Section~\ref{sec:examples}, where, for the second semiparametric model, the error distribution is known to belong to $\{\tau_{\tilde{\alpha}} : \tilde{\alpha}>0\}$. We evaluated the treatment effect of always being on treatment versus never being on treatment. Under this setup, data source 9 aligns perfectly with the target population distribution and it is possible to provide valid inferences for $\psi(Q^0)$ using this data source alone.  We compared three one-step estimators that were constructed via the canonical gradients under these models respectively, and under three scenarios: (1) no data fusion with $\mathcal{S}_7 = \mathcal{S}_5 =\mathcal{S}_3 = \mathcal{S}_1 = \{9\}$, (2) partial data fusion with $\mathcal{S}_7 = \{6,9\}$, $\mathcal{S}_5 = \{5,9\}$, $\mathcal{S}_3 = \{3,9\}$ and $\mathcal{S}_1 = \{1,3,9\}$, and (3) complete data fusion with $\mathcal{S}_7 = \{6,8,9\}$, $\mathcal{S}_5 = \{5,6,8,9\}$, $\mathcal{S}_3 = \{3,5,7,9\}$ and $\mathcal{S}_1 = \{1,3,9\}$. 

The outcome regressions and the propensity scores were estimated via SuperLearner \citep{van2007super} with a library containing a generalized linear model with interaction terms and general additive model under their default settings in the \texttt{SuperLearner} R package \citep{polley2010super}. Each density in the ratios that appear in the second line of Equation~\ref{eq:IF_long_nonpar} was estimated via kernel density estimation using a normal scale bandwidth. Details on the estimation of the conditional density of the regression error in the semiparametric model where this density is known to be symmetric are given in  Appendix~\ref{app:sim}. 
For each simulation study presented in this work, 1000 Monte Carlo replications were conducted.

Using data fusion yields around $10 \%$ and $20\%$ efficiency gains for partial and complete fusion respectively in the nonparametric case. Compared to the nonparametric estimator that was constructed using only data source 9, the semiparametric estimators gained approximately $40$\%, $50$\%, and $60$\% efficiency under no, partial and complete data fusion. Coverage was near nominal for all estimators (93\%-98\%), and the widths of intervals decreased along the same lines as the mean squared error did, with more data fusion and more restrictive statistical models each leading to tighter intervals (Table~\ref{tab:Long_numbers}).

\begin{table}[htb]
\centering
\caption{Bias, variance and coverage of the estimated longitudinal treatment effect. }
\begin{tabular}{lrrrrrrrrr}
\toprule
   & \multicolumn{3}{c}{No Data Fusion} 
& \multicolumn{3}{c}{Partial Data Fusion}&
\multicolumn{3}{c}{Complete Data Fusion}  \\
\cmidrule(l){2-4} \cmidrule(l){5-7}\cmidrule(l){8-10}
 &Bias& Var& Cov\%& Bias& Var& Cov \% &  Bias& Var &Cov \% \\
\midrule
Nonparametric & -0.003 & 0.170 & 97 & -0.004 & 0.154  & 98 & 0.003 & 0.137 & 98 \\
Symmetric & -0.004 & 0.106  & 93& 0.064 & 0.087  & 94& 0.049 & 0.076  & 94\\
\bottomrule
\end{tabular}
\label{tab:Long_numbers}
\end{table}

\subsection{Data illustration: evaluating the immunogenicity of an HIV vaccine}
\label{sec:real data analysis}
The STEP study and the Phambili study were two phase IIb trials that evaluated the safety and efficacy of the same HIV vaccine regimen in different populations \citep{buchbinder2008efficacy,gray2011safety}. The STEP study was conducted at 34 sites in multiple continents, whereas the Phambili study tested the same vaccine at 5 sites in South Africa.  Although the two studies enrolled in distinct geographic areas, there was overlap in the overall demographic characteristics (Table~\ref{tab:HVTN}). Both studies suggested that the vaccine did not prevent HIV-1 infection, although most vaccinees developed an HIV-specific immune response to Clade B as measured by interferon-$\gamma$ ELISpot. 
We studied the associations between immune responses and baseline covariates in Kullback-Leibler projections onto univariate working logistic regression models (details in  Appendix \ref{sec: real_data_appendix}). Studying the associations of these immune responses with baseline covariates is important because, for future vaccines, these responses may be a correlate of protection that can be used to bridge vaccine efficacy estimates to new populations \citep{plotkin2012nomenclature}. We separately treated each of the study populations from the two studies as the target population and compared the estimation results generated by using one single dataset with using both datasets.

The Phambili ELISpot data consist of measurements of Gag, Nef, and Pol immune responses for 93 vaccinees, and the STEP immunogenicity data consist of measurements for 722 vacinees and 257 placebo participants. 
The two trials used different sampling schemes. In Phambili, the immunogenicity assessment was conducted on the first 93 vaccinees who were HIV-1 antibody negative at the week 12 visit and had received the second injection \citep{gray2011safety}. In STEP, a two-phase sampling scheme was adopted to oversample HIV cases \citep{huang2014immune}. To account for this two-phase sampling scheme, we weighted the STEP data by the inverse probability of being sampled given infection status and treatment group. We aimed to evaluate the three HIV-specific immune responses for the vaccine group, namely Gag, Nef, and Pol to clade-B. We used the same criteria as \cite{huang2014immune} for defining a positive immune response. Regardless of the target population, we assumed that the conditional distribution of immune response for the vaccine group between STEP and Phambili given baseline covariates are the same (Condition~\ref{cond:identifiability}\ref{suff_alignment}). These baseline covariates consisted of baseline adenovirus serotype-5 positivity along with age, body mass index, sex, and circumcision status. 
We combined sex and circumcision status into a single 3-level categorical variable to differentiate uncircumcised men from circumcised men and women. 
Data from the HVTN 204 phase II trial support the plausibility of Condition~\ref{cond:identifiability}\ref{suff_alignment}. In particular, they suggest that HIV-specific immune response profiles do not differ by geographic region, whereas baseline adenovirus serotype-5 neutralizing antibodies are strongly associated with HIV-specific immune responses \citep{churchyard2011phase}. 
To examine if Condition~\ref{cond:identifiability}\ref{suff_alignment} is reasonable, we followed the approach proposed by \cite{luedtke2019omnibus} and found no strong evidence against it (Gag: p=0.31; Nef: p=0.46; Pol: p=0.98). In addition, we observed reasonable overlap between the distributions of covariates (Table~\ref{tab:HVTN} and Figure~\ref{fig:density}).

\begin{table}[tb]
\centering
\caption{Estimated coefficient in univariate working logistic model between baseline covariates and Gag using the STEP and Phambili study data. Estimation results are presented as estimates  (standard errors). }
\begin{tabular}{lrrrr}
\toprule
   & \multicolumn{2}{c}{Augmenting STEP} 
& \multicolumn{2}{c}{Augmenting Phambili} \\
\cmidrule(l){2-3} \cmidrule(l){4-5}
 &  STEP (N=979) & Both (N=1072)  &  Phambili (N=93)  & Both (N=1072) \\
\midrule
\hspace{1em}Age & 0.04 (0.01) & 0.04 (0.01) & -0.04 (0.06) & -0.05 (0.04)\\
\hspace{1em}Circumcised male& 1.82 (0.15) & 1.82 (0.14) & 1.01 (0.58) & 0.93 (0.40)\\
\hspace{1em}Uncircumcised male& -0.70 (0.21) & -0.68 (0.20) & 0.24 (0.75) & 0.07 (0.57)\\
\hspace{1em}Female & 1.04 (0.53) & 1.29 (0.40) & 0.51 (0.69) & 0.38 (0.47)\\
\hspace{1em}BMI & 0.01 (0.02) & 0.01 (0.02) & -0.04 (0.05) & -0.05 (0.03)\\
\hspace{1em}Ad-5 positivity & -1.47 (0.25) & -1.47 (0.24) & -0.87 (0.80) & -0.60 (0.62)\\
\bottomrule
\end{tabular}
\label{tab:beta_gag}
\end{table}

We studied coefficients in univariate logistic working models for baseline covariates age, sex, body mass index and baseline adenovirus serotype-5 positivity for each of the three HIV-specific immune markers. We used data from both STEP and Phambili to estimate the conditional expectation of immune response given the set of covariates using SuperLearner \citep{van2007super,polley2010super} with a library containing a random forest, generalized additive model, and elastic net. Results for immune marker Gag are presented in Table~\ref{tab:beta_gag} below and results for Nef and Pol can be found at  Appendix~\ref{sec: real_data_appendix}. All estimators that make use of data fusion gave estimates that were very close to the estimators that only used data from one trial. In contrast, the corresponding standard errors were reduced by more than 30\% for each immune response when data from Phambili were augmented with the STEP study data. The proposed methods also brought efficiency gains when data from STEP were augmented, though these gains were more modest due to the smaller sample size of the Phambili study.

\section*{Acknowledgements}
\label{sec:acknowledgement}
The authors are grateful to Peter Gilbert for helpful discussions and to the participants, investigators, and sponsors of the STEP and Phambili trials. This work was supported by the NIH through award numbers DP2-LM013340 and 5UM1AI068635-09. We thank Ellen Graham, Marco Carone and Andrea Rotnitzky for helpful discussions.

\section*{Appendix}

\begin{appendices}
The appendices are organized as follows. Appendix~\ref{app:tanspace} provides a characterization of the tangent space of $\mathcal{P}$. Appendix~\ref{app:projTanSpace} provides a means to project onto this tangent space. Appendix~\ref{sec:equivDiff} establishes the equivalence of the pathwise differentiability of $\psi$ and $\phi$, and proves the results from Section~\ref{sec:efficientIFs} from the main text. Appendix~\ref{sec: projections_examples} lists out some interesting examples on how to perform projections for different semiparametric models.  Appendix~\ref{example:longitudinal treatment effect appendix} provides further details on the longitudinal treatment effect example from the main text, as well as five additional examples in Appendix~\ref{additional examples}. Appendix~\ref{sec:efficiency theory_appendix} gives a review of nonparametric and semiparametric efficiency theory, followed by a discussion on the construction of estimators in Appendix~\ref{app:contruct_estimators}. Appendix~\ref{app:sim} provides additional information about and results from our simulation study. Appendix~\ref{sec:miscellanea} provides more discussion on potential extensions and implementation of the proposed methods. 

\section{Deriving the canonical gradient}\label{app:deriveGrad}

\subsection{Characterizing the tangent space of $\mathcal{P}$}\label{app:tanspace}
Throughout this appendix we let $\uQ^0$ denote a generic element of $\mathcal{Q}(P^0)$ and $\bar{\mathcal{Z}}_{j-1}^\dagger$ denote the support of $\bar{Z}_{j-1}$ under sampling from $\uQ^0$. Because the distributions in $\mathcal{Q}$ are mutually absolutely continuous, $\bar{\mathcal{Z}}_{j-1}^\dagger$ is the same regardless of the particular value of $\uQ^0\in\mathcal{Q}(P^0)$. 
Let $\mathcal{T}(P^0,\mathcal{P})$ denote the tangent set of model $\mathcal{P}$ at $P^0$.
Because we have assumed that $\mathcal{T}(\uQ^0,\mathcal{Q})$ is a closed linear subspace of $L_0^2(\uQ^0)$, it can be verified that $\mathcal{T}(P^0,\mathcal{P})$, which is the tangent set of a model that is nonparametric model up to the restriction imposed by the data fusion alignment condition, is itself a closed linear subspace of $L_0^2(P^0)$. Therefore, we also refer to $\mathcal{T}(P^0,\mathcal{P})$ as the tangent space of $\mathcal{P}$ at $P^0$. 
Let $L_0^2(\uQ_j^0)$ denote the subspace of $L_0^2(\uQ^0)$ consisting of all functions $f$ for which there exists a function $g : \prod_{i=1}^j \mathcal{Z}_i\rightarrow\mathbb{R}$ that is such that $f(z)=g(\bar{z}_j)$ for all $z=(z_1,\ldots,z_d)\in\mathcal{Z}$ and $E_{\uQ^0}[g(\bar{Z}_j)\mid \bar{Z}_{j-1}]=0$ with $\uQ^0$-probability one. In an abuse of notation, when $f\in L_0^2(\uQ_j^0)$ we let $f(\bar{z}_j)$ denote the unique value that $f(z')$ takes for all $z'$ that are such that $\bar{z}_j'=\bar{z}_j$, so that $f(\bar{z}_j)=f(z)$. Similarly let $L_0^2(P_j^0)$ denote the subspace of $L_0^2(P^0)$ consisting of all functions $f$ for which there exists a function $g : (\prod_{i=1}^j \mathcal{Z}_i)\times \mathcal{S}\rightarrow\mathbb{R}$ that is such that $f(z,s)=g(\bar{z}_j,s)$ for all $z=(z_1,\ldots,z_d)\in\mathcal{Z}$ and $s\in [k]$ and $E_{P^0}[g(\bar{Z}_j,S)\mid \bar{Z}_{j-1},S]=0$ with $P^0$-probability one, and define $L_0^2(P_0^0)$ to be the subspace of $L_0^2(P^0)$ consisting of all functions $f$ for which there exists $g : \mathcal{S}\rightarrow\mathbb{R}$ that is such that $f(z,s)=g(s)$ for all $z\in\mathcal{Z}$ and $s\in [k]$ and that is such that $E_{P^0}[g(S)]=0$. In a similar abuse of notation as that noted earlier, we let $f(\bar{z}_j,s)= f(z,s)$ when $f\in L_0^2(P_j^0)$ and $f(s)= f(z,s)$ when $f\in L_0^2(P_0^0)$.

For each $j\in [d]$, let $\mathcal{T}(\uQ^0,\mathcal{Q}_j)$ be the subspace of $L_0^2(\uQ_j^0)$ that consists of all $f_j\in L_0^2(\uQ_j^0)$ that arise as scores of univariate submodels $\{Q^{(\epsilon)} : \epsilon\in [0,\delta)\}$ for which  $Q_i^{(\epsilon)}=\uQ_i^0$ for all $\epsilon\in [0,\delta)$ and $i\not=j$ and for which $Q^{(\epsilon)}=\uQ^0$ when $\epsilon=0$, where here and throughout score functions are defined in a quadratic mean differentiability sense \citep[Section~7.2 of][]{van2000asymptotic}. Similarly, for $j\in \{0\}\cup [d]$, let $\mathcal{T}(P^0,\mathcal{P}_j)$ denote the subspace of $L_0^2(P_j^0)$ that consists of all $f_j\in L_0^2(P_j^0)$ that arise as scores of univariate submodels $\{P^{(\epsilon)} : \epsilon\in [0,\delta)\}$ for which  $P_i^{(\epsilon)}=P_i^0$ for all $\epsilon\in [0,\delta)$ and $i\not=j$ and for which $P^{(\epsilon)}=P^0$ when $\epsilon=0$. By Condition~\ref{cond:var_indep} from the main text and the definition of $\mathcal{P}$, $\mathcal{P}$ is variation independent in the sense that there exist sets $\mathcal{P}_j$ of conditional distributions of $Z_j\mid \bar{Z}_{j-1},S$, $j\in [d]$, such that $\mathcal{P}$ is equal to the set of all distributions $P$ such that, for all $j\in [d]$, the conditional distribution $P_j$ belongs to $\mathcal{P}_j$ and the marginal distribution of $S$ belongs to the collection of categorical distributions on $[k]$. Hence, by Lemma~1.6 of \cite{van2003unified} and the fact that the tangent set of $\mathcal{P}$ at $P^0$ is a closed linear space, the tangent space $\mathcal{T}(P^0,\mathcal{P})$ of $\mathcal{P}$ at $P^0$ takes the form $\bigoplus_{j=0}^d \mathcal{T}(P^0,\mathcal{P}_j)\equiv \{\sum_{j=0}^d f_j : f_j\in \mathcal{T}(P^0,\mathcal{P}_j)\}$, and the $L_0^2(P)$ projection of a function onto $\mathcal{T}(P^0,\mathcal{P}_j)$ is equal to the sum of the projections onto $\mathcal{T}(P^0,\mathcal{P}_j)$, $j=0,1,\ldots,d$.

Since the marginal distribution of $S$ is unrestricted, $\mathcal{T}(P_0^0,\mathcal{P}_0)=L_0^2(P_0^0)$. Moreover, if $j\in\mathcal{I}$, then the conditional distribution of $Z_j\mid \bar{Z}_{j-1},S$ is also unrestricted, and so $\mathcal{T}(P^0,\mathcal{P}_j)=L_0^2(P_j^0)$. The following result characterizes the other tangent spaces that appear in the direct sum defining $\mathcal{T}(P^0,\mathcal{P})$.
\begin{lemma}\label{lem:TPj}
If Conditions~\ref{cond:identifiability}, \ref{cond:var_indep}, and \ref{cond:posPart} from the main text hold and $j\in\mathcal{J}$, then
\begin{align}
    \mathcal{T}(P^0,\mathcal{P}_j)\equiv \Big\{&(z,s)\mapsto h_j(\bar{z}_j,s) + \mathbbm{1}_{\mathcal{S}_j}(s) \mathbbm{1}_{\bar{\mathcal{Z}}_{j-1}^\dagger}(\bar{z}_{j-1}) [f_j(\bar{z}_j)-h_j(\bar{z}_j,s)] \nonumber \\
    &: f_j\in\mathcal{T}(\uQ^0,\mathcal{Q}_j),h_j\in  L_0^2(P_j^0)\Big\}. \label{eq:TPj}
\end{align}
\end{lemma}
\begin{proof}[Proof of Lemma~\ref{lem:TPj}]
Fix $j\in\mathcal{J}$ and let $\mathcal{A}_j$ denote the right-hand side of \eqref{eq:TPj}. We first show that $\mathcal{A}_j\subseteq \mathcal{T}(P^0,\mathcal{P}_j)$, and then we show $\mathcal{T}(P^0,\mathcal{P}_j)\subseteq\mathcal{A}_j$.\\[1em]
\textbf{Part 1 of proof: $\mathcal{A}_j\subseteq \mathcal{T}(P^0,\mathcal{P}_j)$.} Fix $f_j\in\mathcal{T}(\uQ^0,\mathcal{Q}_j)$ and $h_j\in  L_0^2(P_j^0)$. As $f_j\in \mathcal{T}(\uQ_j^0,\mathcal{Q}_j)$, there exists a univariate submodel $\{Q^{(\epsilon)} : \epsilon\in [0,\delta)\}$ with score $f_j$ at $\epsilon=0$ for which  $Q_i^{(\epsilon)}=\uQ_i^0$ for all $\epsilon\in [0,\delta)$ and $i\not=j$ and for which $Q^{(\epsilon)}=\uQ^0$ when $\epsilon=0$. For each $\epsilon\in [0,\delta)$, we let $P^{(\epsilon)}$ be the distribution whose Radon-Nikodym derivative satisfies
\begin{align*}
    &\frac{dP^{(\epsilon)}}{dP^0}(z,s) \\
    &=  \left(\frac{dQ_j^{(\epsilon)}}{d\uQ_j^0}(z_j\mid \bar{z}_{j-1})\right)^{\mathbbm{1}_{\mathcal{S}_j}(s)\mathbbm{1}_{\bar{\mathcal{Z}}_{j-1}^\dagger}(\bar{z}_{j-1})} [C_j^{(\epsilon)}(\bar{z}_{j-1},s)\kappa(\epsilon h_j(\bar{z}_j,s))]^{\mathbbm{1}\left(\{s\not\in \mathcal{S}_j\}\cup \{\bar{z}_{j-1}\not\in\bar{\mathcal{Z}}_{j-1}^\dagger\}\right)},
\end{align*}
where $\kappa(x)=2[1+\exp(-2x)]^{-1}$, $\frac{dQ_j^{(\epsilon)}}{d\uQ_j^0}(\cdot\mid \bar{z}_{j-1})$ denotes the Radon-Nikodym derivative of $Q_j^{(\epsilon)}(\cdot\mid \bar{z}_{j-1})$ relative to $\uQ_j^0(\cdot\mid \bar{z}_{j-1})$ and $c_j^{(\epsilon)}(\bar{z}_{j-1},s)\equiv 1/\int \kappa(\epsilon h_j(\bar{z}_j,s)) P^0(dz_j\mid \bar{z}_{j-1},s)$. It can be readily shown that $P^{(\epsilon)}$ belongs to $\mathcal{P}$ since, for all $j'\in \mathcal{J}$ and $s\in \mathcal{S}_{j'}$, (a) the marginal distribution of $\bar{Z}_{j'-1}$ under sampling from $Q^{(\epsilon)}$ is absolutely continuous with respect to the conditional distribution of $\bar{Z}_{j'-1}\mid S=s$ under sampling from $P^{(\epsilon)}$, and (b) $P_{j'}^{(\epsilon)}(\cdot\mid \bar{z}_{j'-1},s)=Q_{j'}^{(\epsilon)}(\cdot\mid \bar{z}_{j'-1})$ $Q^{(\epsilon)}$-almost everywhere. Indeed, (b) can be seen to hold by inspecting the definition of $P^{(\epsilon)}$, and (a) can be seen to hold for all $j'\not=j$ since Condition~\ref{cond:identifiability} holds and for $j'=j$ by the following observations: the marginal distribution of $\bar{Z}_{j-1}$ under $Q^{(\epsilon)}$ is absolutely continuous with respect to the analogous marginal distribution under $\uQ^0$ since distributions in $\mathcal{Q}$ are mutually absolutely continuous; the marginal distribution of $\bar{Z}_{j-1}$ under $\uQ^0$ is absolutely continuous with respect to the distribution of $\bar{Z}_{j-1}\mid S=s$ under $P^0$ by Condition~\ref{cond:identifiability}; and the distribution of $\bar{Z}_{j-1}\mid S=s$ under $P^0$ can be seen to be absolutely continuous with respect to the analogous distribution under $P^{(\epsilon)}$ by inspecting the definition of $P^{(\epsilon)}$. We will now also show that $\{P^{(\epsilon)} : \epsilon\in [0,\delta)\}$ is quadratic mean differentiable at $\epsilon=0$ with score $(z,s)\mapsto h_j(\bar{z}_j,s) + \mathbbm{1}_{\mathcal{S}_j}(s) \mathbbm{1}_{\bar{\mathcal{Z}}_{j-1}^\dagger}(\bar{z}_{j-1}) [f_j(\bar{z}_j)-h_j(\bar{z}_j,s)]$. In particular, we will show that $r(\epsilon)=o(\epsilon^2)$, where
\begin{align*}
    &r(\epsilon) \\
    &\equiv \int \left(\frac{dP^{(\epsilon)}}{dP^0}(z,s)^{1/2}- 1 - \frac{1}{2}\epsilon \left\{h_j(\bar{z}_j,s) + \mathbbm{1}_{\mathcal{S}_j}(s) \mathbbm{1}_{\bar{\mathcal{Z}}_{j-1}^\dagger}(\bar{z}_{j-1}) [f_j(\bar{z}_j)-h_j(\bar{z}_j,s)]\right\}\right)^2 dP^0(z,s).
\end{align*}
Here we have chosen to use $P^0$ as the dominating measure (this choice simplifies our calculations, but has no bearing on the quadratic mean differentiability property since this property is invariant to the choice of dominating measure). To show the above, we start by noting that
\begin{align}
    r(\epsilon)&= \int \mathbbm{1}_{\mathcal{S}_j}(s)\mathbbm{1}_{\bar{\mathcal{Z}}_{j-1}^\dagger}(\bar{z}_{j-1})\left(\frac{dQ_j^{(\epsilon)}}{d\uQ_j^0}(z_j\mid \bar{z}_{j-1})^{1/2}- 1 - \frac{1}{2}\epsilon f_j(\bar{z}_j)\right)^2 dP^0(z,s) \label{eq:rEps} \\
    &\quad + \int \mathbbm{1}\left(\{s\not\in \mathcal{S}_j\}\cup \{\bar{z}_{j-1}\not\in\bar{\mathcal{Z}}_{j-1}^\dagger\}\right) \left([C_j^{(\epsilon)}(\bar{z}_{j-1},s)\kappa(\epsilon h_j(\bar{z}_j,s))]^{1/2}- 1 - \frac{1}{2}\epsilon h_j(\bar{z}_j,s)\right)^2 \nonumber \\
    &\quad\hspace{2em}\cdot dP^0(z,s). \nonumber
\end{align}
Straightforward calculations show that the second term on the right is $o(\epsilon^2)$ (cf. Example~25.16 and Lemma~7.6 in \cite{van2000asymptotic}). We now argue that the first term is $o(\epsilon^2)$. To do this, we let $\bar{P}_{j-1}^0(\cdot\mid \mathcal{S}_j)$ denote conditional distribution of $\bar{Z}_{j-1}$ given that $S\in\mathcal{S}_j$ under $P^0$ and $P_j^0(\cdot\mid \bar{z}_{j-1},\mathcal{S}_j)$ the conditional distribution of $Z_j$ given that $\bar{Z}_{j-1}=\bar{z}_{j-1}$ and $S\in \mathcal{S}_j$. We note that
\begin{align*}
    &\int \mathbbm{1}_{\mathcal{S}_j}(s)\mathbbm{1}_{\bar{\mathcal{Z}}_{j-1}^\dagger}(\bar{z}_{j-1}) \left(\frac{dQ_j^{(\epsilon)}}{d\uQ_j^0}(z_j\mid \bar{z}_{j-1})^{1/2}- 1 - \frac{1}{2}\epsilon f_j(\bar{z}_j)\right)^2 dP^0(z,s) \\
    &= P^0(S\in\mathcal{S}_j)\int \mathbbm{1}_{\bar{\mathcal{Z}}_{j-1}^\dagger}(\bar{z}_{j-1}) \left(\frac{dQ_j^{(\epsilon)}}{d\uQ_j^0}(z_j\mid \bar{z}_{j-1})^{1/2}- 1 - \frac{1}{2}\epsilon f_j(\bar{z}_j)\right)^2 P^0(dz\mid \mathcal{S}_j) \\
    &= P^0(S\in\mathcal{S}_j)\int_{\bar{\mathcal{Z}}_{j-1}^\dagger} \int_{\mathcal{Z}_j} \left(\frac{dQ_j^{(\epsilon)}}{d\uQ_j^0}(z_j\mid \bar{z}_{j-1})^{1/2}- 1 - \frac{1}{2}\epsilon f_j(\bar{z}_j)\right)^2  \\
    &\hspace{11em}\cdot P_j^0(dz_j\mid \bar{z}_{j-1},\mathcal{S}_j)\bar{P}_{j-1}^0(d\bar{z}_{j-1}\mid \mathcal{S}_j) \\
    &= P^0(S\in\mathcal{S}_j)\int_{\bar{\mathcal{Z}}_{j-1}^\dagger} \int_{\mathcal{Z}_j} \left(\frac{dQ_j^{(\epsilon)}}{d\uQ_j^0}(z_j\mid \bar{z}_{j-1})^{1/2}- 1 - \frac{1}{2}\epsilon f_j(\bar{z}_j)\right)^2  \uQ_j^0(dz_j\mid \bar{z}_{j-1})\bar{P}_{j-1}^0(d\bar{z}_{j-1}\mid \mathcal{S}_j). \\
    \intertext{Letting $\bar{\uQ}_{j-1}^0$ denote the marginal distribution of $\bar{Z}_{j-1}$ under sampling from $\uQ^0$ and letting $c_{j-1}<\infty$ denote the constant guaranteed to hold by Condition~\ref{cond:posPart}, the above display continues as}
    &= P^0(S\in\mathcal{S}_j)\int_{\bar{\mathcal{Z}}_{j-1}^\dagger} \int_{\mathcal{Z}_j} \left(\frac{dQ_j^{(\epsilon)}}{d\uQ_j^0}(z_j\mid \bar{z}_{j-1})^{1/2}- 1 - \frac{1}{2}\epsilon f_j(\bar{z}_j)\right)^2  \\
    &\hspace{11em}\cdot \uQ_j^0(dz_j\mid \bar{z}_{j-1})\lambda_{j-1}(\bar{z}_{j-1})^{-1}\bar{\uQ}_{j-1}^0(d\bar{z}_{j-1}) \\
    &\le c_{j-1} P^0(S\in\mathcal{S}_j)\int_{\bar{\mathcal{Z}}_{j-1}^\dagger} \int_{\mathcal{Z}_j} \left(\frac{dQ_j^{(\epsilon)}}{d\uQ_j^0}(z_j\mid \bar{z}_{j-1})^{1/2}- 1 - \frac{1}{2}\epsilon f_j(\bar{z}_j)\right)^2  \uQ_j^0(dz_j\mid \bar{z}_{j-1}) \bar{\uQ}_{j-1}^0(d\bar{z}_{j-1}) \\
    &= c_{j-1} P^0(S\in\mathcal{S}_j)\int \left(\frac{dQ_j^{(\epsilon)}}{d\uQ_j^0}(z_j\mid \bar{z}_{j-1})^{1/2}- 1 - \frac{1}{2}\epsilon f_j(\bar{z}_j)\right)^2  d\uQ^0(z),
\end{align*}
where we used that, on $\bar{\mathcal{Z}}_{j-1}^\dagger$, $1/\lambda_{j-1}(\cdot)$ is the Radon-Nikodym derivative of the absolutely continuous part of the conditional distribution of $\bar{Z}_{j-1}\mid S\in\mathcal{S}_j$ under sampling from $P^0$ relative to the marginal distribution of $\bar{Z}_{j-1}$ under sampling from $\uQ^0$, where this absolutely continuous part defined via Lebesgue's decomposition theorem. 
The right-hand side above is $o(\epsilon^2)$ since $\{Q^{(\epsilon)} : \epsilon\in [0,\delta)\}$ is quadratic mean differentiable. Returning to \eqref{eq:rEps}, this shows that $r(\epsilon)=o(\epsilon^2)$. Hence, $\{P^{(\epsilon)} : \epsilon\in [0,\delta)\}$ is a submodel of $\mathcal{P}$ that is such that $P^{(\epsilon)}=P^0$ when $\epsilon=0$ and that is quadratic mean differentiable at $\epsilon=0$ with score $(z,s)\mapsto h_j(\bar{z}_j,s) + \mathbbm{1}_{\mathcal{S}_j}(s) \mathbbm{1}_{\bar{\mathcal{Z}}_{j-1}^\dagger}(\bar{z}_{j-1}) [f_j(\bar{z}_j)-h_j(\bar{z}_j,s)]$. As $f_j\in\mathcal{T}(\uQ^0,\mathcal{Q}_j)$ and $h_j\in  L_0^2(P_j^0)$ were arbitrary, $\mathcal{A}_j\subseteq \mathcal{T}(P^0,\mathcal{P}_j)$.\\[1em]
\textbf{Part 2 of proof: $\mathcal{T}(P^0,\mathcal{P}_j)\subseteq \mathcal{A}_j$.}  Fix $g_j\in \mathcal{T}(P^0,\mathcal{P}_j)$ and let $\{P^{(\epsilon)} : \epsilon\in [0,\delta)\}$ be a submodel of $\mathcal{P}$ that is such that $P^{(\epsilon)}=P^0$ when $\epsilon=0$ and that has score $g_j$ at $\epsilon=0$. By the variation independence of $P_i^0$ and $P_j^0$, $i\not=j$, we can suppose without loss of generality that $P_i^{(\epsilon)}=P_i^0$ for all $i\not=j$ and also that the marginal distribution of $S$ under $P^{(\epsilon)}$ is equal to the marginal distribution of $S$ under $P^0$. We will show that there exist $f_j\in\mathcal{T}(\uQ^0,\mathcal{Q}_j)$ and $h_j\in  L_0^2(P_j^0)$ such that $g_j(z,s)=h_j(\bar{z}_j,s) + \mathbbm{1}_{\mathcal{S}_j}(s) \mathbbm{1}_{\bar{\mathcal{Z}}_{j-1}^\dagger}(\bar{z}_{j-1}) [f_j(\bar{z}_j)-h_j(\bar{z}_j,s)]$ $P^0$-almost everywhere. Since $\mathcal{T}(P^0,\mathcal{P}_j)$ is necessarily a subset of the maximal tangent space $L_0^2(P_j^0)$, we can let $h_j=g_j$. It remains to show that there exists an $f_j\in\mathcal{T}(\uQ^0,\mathcal{Q}_j)$ that is such that $g_j(z,s)=f_j(z)$ for $P^0$-almost all $z\in \bar{\mathcal{Z}}_{j-1}^\dagger$ and all $s\in \mathcal{S}_j$. Since $g_j\in L_0^2(P_j^0)$, we recall that $g_j(z,s)$ does not depend on $(z_{j+1},\ldots,z_d)$, we continue our earlier convention and write $g_j(\bar{z}_j,s)$ to denote unique value that $g_j(z',s)$ takes for all $z'$ that are such that $\bar{z}_j'=\bar{z}_j$. By the quadratic mean differentiability of $\{P^{(\epsilon)} : \epsilon\in [0,\delta)\}$,
\begin{align*}
    o(\epsilon^2)&= \int \left(\frac{dP^{(\epsilon)}}{dP^0}(z,s)^{1/2}- 1 - \frac{1}{2}\epsilon g_j(\bar{z}_j,s)\right)^2 dP^0(z,s) \\ 
    &= \int \mathbbm{1}_{\mathcal{S}_j}(s) \mathbbm{1}_{\bar{\mathcal{Z}}_{j-1}^\dagger}(\bar{z}_{j-1}) \left(\frac{dP^{(\epsilon)}}{dP^0}(z,s)^{1/2}- 1 - \frac{1}{2}\epsilon g_j(\bar{z}_j,s)\right)^2 dP^0(z,s) \\
    &\quad+ \int \mathbbm{1}\left(\{s\not\in \mathcal{S}_j\}\cup \{\bar{z}_{j-1}\not\in\bar{\mathcal{Z}}_{j-1}^\dagger\}\right)\left(\frac{dP^{(\epsilon)}}{dP^0}(z,s)^{1/2}- 1 - \frac{1}{2}\epsilon g_j(\bar{z}_j,s)\right)^2 dP^0(z,s).
\end{align*}
Since both terms on the right are nonnegative, they must both be $o(\epsilon^2)$. This is true, in particular, for the first term, yielding:
\begin{align*}
    o(\epsilon^2)&= \int \mathbbm{1}_{\mathcal{S}_j}(s) \mathbbm{1}_{\bar{\mathcal{Z}}_{j-1}^\dagger}(\bar{z}_{j-1}) \left(\frac{dP^{(\epsilon)}}{dP^0}(z,s)^{1/2}- 1 - \frac{1}{2}\epsilon g_j(\bar{z}_j,s)\right)^2 dP^0(z,s) \\
    &= \sum_{s : s\in\mathcal{S}_j} P^0(S=s)\int  \mathbbm{1}_{\bar{\mathcal{Z}}_{j-1}^\dagger}(\bar{z}_{j-1}) \left(\frac{dP^{(\epsilon)}}{dP^0}(z,s)^{1/2}- 1 - \frac{1}{2}\epsilon g_j(\bar{z}_j,s)\right)^2 P^0(dz\mid s).
\end{align*}
For each $\epsilon\in [0,\delta)$, let $Q^{(\epsilon)}\in\mathcal{Q}$ be such that $Q_j^{(\epsilon)}(\cdot\mid \bar{z}_{j-1})=P_j^{(\epsilon)}(\cdot\mid \bar{z}_{j-1}, \mathcal{S}_j)$ for $\uQ^0$-almost all $\bar{z}_{j-1}\in\bar{\mathcal{Z}}_{j-1}^\dagger$ and $Q_i^{(\epsilon)}=Q_i^{(\epsilon)}$ for all $i\not= j$. It can then be verified that $\frac{dP^{(\epsilon)}}{dP^0}(z,s)=\frac{dQ^{(\epsilon)}}{d\uQ^0}(z)$ for all $s\in\mathcal{S}_j$ and $z$ that are such that $\bar{z}_{j-1}\in \bar{\mathcal{Z}}_{j-1}^\dagger$. Combining the fact that $\frac{dP^{(\epsilon)}}{dP^0}(z,s)$ does not depend on the particular value of $s\in\mathcal{S}_j$ with the fact that all $k$ of the (nonnegative) terms in the sum above are $o(\epsilon^2)$, it must be the case that there exists some function $f_j : \prod_{i=1}^j \mathcal{Z}_i\rightarrow\mathbb{R}$ such that $f_j(\bar{z}_j)=g_j(\bar{z}_j,s)$ for $P^0$-almost all $(\bar{z}_j,s)$ that are such that $\bar{z}_{j-1}\in\bar{\mathcal{Z}}_{j-1}^\dagger$ and $s\in\mathcal{S}_j$. Plugging these observations into the above yields that
\begin{align*}
    o(\epsilon^2)&= \sum_{s : s\in\mathcal{S}_j} P^0(S=s) \int \mathbbm{1}_{\bar{\mathcal{Z}}_{j-1}^\dagger}(\bar{z}_{j-1})  \left(\frac{dQ^{(\epsilon)}}{d\uQ^0}(z)^{1/2}- 1 - \frac{1}{2}\epsilon f_j(\bar{z}_j)\right)^2 P^0(dz\mid s) \\
    &=P^0(S\in\mathcal{S}_j) \int \mathbbm{1}_{\bar{\mathcal{Z}}_{j-1}^\dagger}(\bar{z}_{j-1})   \left(\frac{dQ^{(\epsilon)}}{d\uQ^0}(z)^{1/2}- 1 - \frac{1}{2}\epsilon f_j(\bar{z}_j)\right)^2 P^0(dz\mid \mathcal{S}_j) \\
    &=P^0(S\in\mathcal{S}_j) \int_{\bar{\mathcal{Z}}_{j-1}^\dagger} \int_{\mathcal{Z}_j}   \left(\frac{dQ^{(\epsilon)}}{d\uQ^0}(z)^{1/2}- 1 - \frac{1}{2}\epsilon f_j(\bar{z}_j)\right)^2 P_j^0(dz_j\mid \bar{z}_{j-1},\mathcal{S}_j) \bar{P}_{j-1}^0(d\bar{z}_{j-1}\mid \mathcal{S}_j) \\
    &= P^0(S\in\mathcal{S}_j) \int_{\bar{\mathcal{Z}}_{j-1}^\dagger} \int_{\mathcal{Z}_j}   \left(\frac{dQ^{(\epsilon)}}{d\uQ^0}(z)^{1/2}- 1 - \frac{1}{2}\epsilon f_j(\bar{z}_j)\right)^2 \uQ_j^0(dz_j\mid \bar{z}_{j-1}) \bar{P}_{j-1}^0(d\bar{z}_{j-1}\mid \mathcal{S}_j) \\
    &= P^0(S\in\mathcal{S}_j) \int_{\bar{\mathcal{Z}}_{j-1}^\dagger} \int_{\mathcal{Z}_j}   \left(\frac{dQ^{(\epsilon)}}{d\uQ^0}(z)^{1/2}- 1 - \frac{1}{2}\epsilon f_j(\bar{z}_j)\right)^2 \\
    &\hspace{11em}\cdot \uQ_j^0(dz_j\mid \bar{z}_{j-1}) \lambda_{j-1}(\bar{z}_{j-1})^{-1} \bar{\uQ}_{j-1}^0(d\bar{z}_{j-1}) \\
    &\ge P^0(S\in\mathcal{S}_j) c_{j-1}^{-1} \int_{\bar{\mathcal{Z}}_{j-1}^\dagger} \int_{\mathcal{Z}_j}   \left(\frac{dQ^{(\epsilon)}}{d\uQ^0}(z)^{1/2}- 1 - \frac{1}{2}\epsilon f_j(\bar{z}_j)\right)^2 \uQ_j^0(dz_j\mid \bar{z}_{j-1})  \bar{\uQ}_{j-1}^0(d\bar{z}_{j-1}) \\
    &= P^0(S\in\mathcal{S}_j) c_{j-1}^{-1} \int   \left(\frac{dQ^{(\epsilon)}}{d\uQ^0}(z)^{1/2}- 1 - \frac{1}{2}\epsilon f_j(\bar{z}_j)\right)^2 d\uQ^0(z) \\
\end{align*}
where we used Condition~\ref{cond:identifiability} to replace $P_j^0(\cdot\mid \bar{z}_{j-1},s)$ by $\uQ_j^0(\cdot \mid \bar{z}_{j-1})$ and Condition~\ref{cond:posPart} to lower bound $\lambda_j(\bar{z}_{j-1})^{-1}$ by $c_{j-1}^{-1}$. As $P^0(S\in\mathcal{S}_j) c_{j-1}^{-1}>0$, the above is only possible if the integral above is $o(\epsilon^2)$. This implies that $\{Q^{(\epsilon)} : \epsilon\in [0,\delta)\}$ is a submodel of $\mathcal{Q}$ that is quadratic mean differentiable at $\epsilon=0$ with score $f_j$ and $Q_i^{(\epsilon)}=\uQ_i^0$ for all $i\not=j$. It is also readily verified that $Q^{(\epsilon)}=\uQ^0$ when $\epsilon=0$. Hence, it must be the case that $f_j\in \mathcal{T}(\uQ^0,\mathcal{Q}_j)$. Finally, noting that $f_j(\bar{z}_j)=g_j(\bar{z}_j,s)$ for $P^0$-almost all $(\bar{z}_j,s)$ that are such that $\bar{z}_{j-1}\in\bar{\mathcal{Z}}_{j-1}^\dagger$ and $s\in\mathcal{S}_j$, we see that $g_j(z,s)=h_j(\bar{z}_j,s) + \mathbbm{1}_{\mathcal{S}_j}(s) \mathbbm{1}_{\bar{\mathcal{Z}}_{j-1}^\dagger}(\bar{z}_{j-1}) [f_j(\bar{z}_j)-h_j(\bar{z}_j,s)]$ $P^0$-almost everywhere. Hence, $g_j\in\mathcal{A}_j$. As $g_j$ was an arbitrary element of $\mathcal{T}(P^0,\mathcal{P}_j)$, we see that $\mathcal{T}(P^0,\mathcal{P}_j)\subseteq\mathcal{A}_j$.
\end{proof}

\subsection{Projecting onto the tangent space of $\mathcal{P}$}\label{app:projTanSpace}

For a subspace $\mathcal{A}$ of $L_0^2(\uQ^0)$, we let $\Pi_{\uQ^0}(\,\cdot\mid \mathcal{A})$ denote the projection operator onto $\mathcal{A}$. We define $\Pi_{P^0}(\,\cdot\mid \mathcal{A})$ similarly for subspaces $\mathcal{A}$ of $L_0^2(P^0)$. We begin with a lemma that will prove useful later when we establish the form of $\Pi_{P^0}\{\,\cdot\mid \mathcal{T}(P^0,\mathcal{P})\}$.

\begin{lemma}
\label{lem:lambda_Q}
Let $f\in L_0^2(P^0)$ and $j\in\mathcal{J}$. If Conditions~\ref{cond:identifiability} and \ref{cond:posPart} hold, then the following function is contained in $L_0^2(\uQ_j^0)$:
\begin{align*}
    \Gamma_j(f) : \bar{z}_j&\mapsto \mathbbm{1}_{\bar{\mathcal{Z}}_{j-1}^\dagger}(\bar{z}_{j-1}) E_{P^0}\left[E_{P^0}\left\{f(Z,S)\mid \bar{Z}_j,S\right\} - E_{P^0}\left\{f(Z,S)\mid \bar{Z}_{j-1},S\right\}\mid \bar{Z}_j=\bar{z}_j,S\in\mathcal{S}_j\right].
\end{align*}
\end{lemma}
\begin{proof}[Proof of Lemma~\ref{lem:lambda_Q}]
To ease notation, we let $f_j\equiv \Gamma_j(f)$. 
Condition~\ref{cond:identifiability} ensures that $f_j$ is uniquely defined up to $\uQ^0$-null sets. Now, again using Condition~\ref{cond:identifiability} and also applying Condition~\ref{cond:posPart}, we see that the following holds $\uQ^0$-almost surely:
\begin{align*}
    &E_{\uQ^0}[f_j(\bar{Z}_j)\mid \bar{Z}_{j-1}] \\
    &= E_{\uQ^0}\left[E_{P^0}\left[E_{P^0}\left\{f(Z,S)\mid \bar{Z}_j,S\right\} - E_{P^0}\left\{f(Z,S)\mid \bar{Z}_{j-1},S\right\}\mid \bar{Z}_j,S\in\mathcal{S}_j\right]\mid \bar{Z}_{j-1} \right] \\
    &= E_{P^0}\left(E_{P^0}\left[E_{P^0}\left\{f(Z,S)\mid \bar{Z}_j,S\right\} - E_{P^0}\left\{f(Z,S)\mid \bar{Z}_{j-1},S\right\}\mid \bar{Z}_j,S\in\mathcal{S}_j\right]\mid \bar{Z}_{j-1},S\in\mathcal{S}_j \right) \\
    &= E_{P^0}\left\{f(Z,S)\mid \bar{Z}_{j-1},S\in\mathcal{S}_j\right\} - E_{P^0}\left\{f(Z,S)\mid \bar{Z}_{j-1},S\in\mathcal{S}_j\right\} \\
    &= 0.
\end{align*}
We now show that $E_{\uQ^0}[f_j(\bar{Z}_j)^2]<\infty$. This can be seen to hold since
\begin{align*}
    E_{\uQ^0}[f_j(\bar{Z}_j)^2]&= \int \int f_j(\bar{z}_j)^2 \uQ_j^0(dz_j\mid \bar{z}_{j-1}) d\bar{\uQ}_{j-1}^0(\bar{z}_{j-1}) \\
    &= \int_{\bar{\mathcal{Z}}_{j-1}^\dagger} \int_{\mathcal{Z}_j} f_j(\bar{z}_j)^2 \uQ_j^0(dz_j\mid \bar{z}_{j-1}) d\bar{\uQ}_{j-1}^0(\bar{z}_{j-1}) \\
    &= \int_{\bar{\mathcal{Z}}_{j-1}^\dagger} \int_{\mathcal{Z}_j} f_j(\bar{z}_j)^2 P_j^0(dz_j\mid \bar{z}_{j-1},\mathcal{S}_j) d\bar{\uQ}_{j-1}^0(\bar{z}_{j-1}) \\
    &= \int_{\bar{\mathcal{Z}}_{j-1}^\dagger} \int_{\mathcal{Z}_j} f_j(\bar{z}_j)^2 P_j^0(dz_j\mid \bar{z}_{j-1},\mathcal{S}_j) \lambda_{j-1}(\bar{z}_{j-1}) \bar{P}_{j-1}^0(d\bar{z}_{j-1}\mid \mathcal{S}_j) \\
    &\le c_{j-1} \int_{\bar{\mathcal{Z}}_{j-1}^\dagger} \int_{\mathcal{Z}_j} f_j(\bar{z}_j)^2 P_j^0(dz_j\mid \bar{z}_{j-1},\mathcal{S}_j)  \bar{P}_{j-1}^0(d\bar{z}_{j-1}\mid \mathcal{S}_j) \\
    &= c_{j-1}E_{P^0}\left[f_j(\bar{Z}_j)^2\mid \mathcal{S}_j\right] \\
    &\le \frac{c_{j-1}}{P^0(S\in\mathcal{S}_j)} E_{P^0}\left[f_j(\bar{Z}_j)^2\right].
\end{align*}
It is also readily verified that $E_{P^0}\left[f_j(\bar{Z}_j)^2\right]$ is upper bounded by a $P^0$-dependent constant times $E_{P^0}\left[f(\bar{Z}_j)^2\right]$, which is finite since $f\in L_0^2(P^0)$.
\end{proof}

The above ensures that, for any $f\in L_0^2(P^0)$ and $j\in\mathcal{J}$, the $L_0^2(\uQ^0)$-projection of $\Gamma_j(f)$ onto $\mathcal{T}(\uQ^0,\mathcal{Q})$ is well-defined. This proves useful in a result to follow (Lemma~\ref{lem:Tproj}), which characterizes the $L_0^2(P^0)$-projection operator onto $\mathcal{T}(P^0,\mathcal{P})$. Before presenting that result, we provide another characterization of $\Gamma_j$. In the proof of the next result, we let $\bar{P}_{j}^0(\cdot \mid \mathcal{S}_j)$ denote the conditional distribution of $\bar{Z}_j$ given $S\in\mathcal{S}_j$ under sampling from $P^0$.
\begin{lemma}\label{lem:LambdajIdent}
Suppose that Condition~\ref{cond:identifiability} holds. 
For any $f\in L_0^2(P^0)$ and $j\in\mathcal{J}$, the following holds for $\bar{P}_{j}^0(\cdot\mid \mathcal{S}_j)$-almost all $\bar{z}_j$:
\begin{align*}
    \Gamma_j(f)(\bar{z}_j)=\mathbbm{1}_{\bar{\mathcal{Z}}_{j-1}^\dagger}(\bar{z}_{j-1}) \left(E_{P^0}\left[f(Z,S)\mid \bar{Z}_j=\bar{z}_j,S\in\mathcal{S}_j\right] - E_{P^0}\left[f(Z,S)\mid \bar{Z}_{j-1}=\bar{z}_{j-1},S\in\mathcal{S}_j\right]\right).
\end{align*}
\end{lemma}
\begin{proof}[Proof of Lemma~\ref{lem:LambdajIdent}]
By applying the law of total expectation to the first term in the definition of $\Gamma_j(f)$ from Lemma~\ref{lem:lambda_Q}, it suffices to show that, for $\bar{P}_{j}^0(\cdot\mid \mathcal{S}_j)$-almost all $\bar{z}_j$,
\begin{align*}
    &\mathbbm{1}_{\bar{\mathcal{Z}}_{j-1}^\dagger}(\bar{z}_{j-1})E_{P^0}\left[f(Z,S)\mid \bar{Z}_{j-1}=\bar{z}_{j-1},S\in\mathcal{S}_j\right] \\
    &= \mathbbm{1}_{\bar{\mathcal{Z}}_{j-1}^\dagger}(\bar{z}_{j-1})E_{P^0}\left[E_{P^0}\left\{f(Z,S)\mid \bar{Z}_{j-1},S\right\}\mid \bar{Z}_j=\bar{z}_j,S\in\mathcal{S}_j\right].
\end{align*}
This can be seen to hold since, for $\bar{P}_{j}^0(\cdot\mid \mathcal{S}_j)$-almost all $\bar{z}_j$,
\begin{align*}
    &\mathbbm{1}_{\bar{\mathcal{Z}}_{j-1}^\dagger}(\bar{z}_{j-1})E_{P^0}\left[E_{P^0}\left\{f(Z,S)\mid \bar{Z}_{j-1},S\right\}\mid \bar{Z}_j=\bar{z}_j,S\in\mathcal{S}_j\right] \\
    &= \mathbbm{1}_{\bar{\mathcal{Z}}_{j-1}^\dagger}(\bar{z}_{j-1})E_{P^0}\left[E_{P^0}\left\{E_{P^0}[f(Z,S)\mid \bar{Z}_j,S]\mid \bar{Z}_{j-1},S\right\}\mid \bar{Z}_j=\bar{z}_j,S\in\mathcal{S}_j\right] \\
    &= \mathbbm{1}_{\bar{\mathcal{Z}}_{j-1}^\dagger}(\bar{z}_{j-1})E_{P^0}\left[E_{P^0}\left\{E_{P^0}[f(Z,S)\mid \bar{Z}_j,S]\mid \bar{Z}_{j-1},S\in\mathcal{S}_j\right\}\mid \bar{Z}_j=\bar{z}_j,S\in\mathcal{S}_j\right] \\
    &= \mathbbm{1}_{\bar{\mathcal{Z}}_{j-1}^\dagger}(\bar{z}_{j-1})E_{P^0}\left[E_{P^0}\left\{f(Z,S)\mid \bar{Z}_{j-1},S\in\mathcal{S}_j\right\}\mid \bar{Z}_j=\bar{z}_j,S\in\mathcal{S}_j\right] \\
    &= \mathbbm{1}_{\bar{\mathcal{Z}}_{j-1}^\dagger}(\bar{z}_{j-1})E_{P^0}\left\{f(Z,S)\mid \bar{Z}_{j-1}=\bar{z}_{j-1},S\in\mathcal{S}_j\right\},
\end{align*}
where the first and third equalities above hold by the law of total expectation, the second holds by Condition~\ref{cond:identifiability}, and the fourth holds by properties of conditional expectations and the fact that $\bar{Z}_j=(\bar{Z}_{j-1},Z_j)$.
\end{proof}

\begin{lemma}\label{lem:Tproj}
Suppose that Conditions~\ref{cond:identifiability}, \ref{cond:var_indep}, and \ref{cond:posPart} hold. For any $f\in L_0^2(P^0)$ satisfying $\Pi_{\uQ^0}\{\Gamma_j(f)\mid \mathcal{T}(\uQ^0,\mathcal{Q})\}(\bar{z}_j) = \lambda_{j-1}(\bar{z}_{j-1})\Pi_{\uQ^0}\{\Gamma_j(f)/\lambda_{j-1}\mid \mathcal{T}(\uQ^0,\mathcal{Q})\}(\bar{z}_j)$ $P^0$-a.s. for all $j\in\mathcal{J}$, it holds that
\begin{align}
    \Pi_{P^0}&\{f \mid \mathcal{T}(P^0,\mathcal{P})\}(z,s) \nonumber \\
    &= \Pi_{P^0}\{f\mid L_0^2(P_0^0)\}(s) + \sum_{j=1}^d \Pi_{P^0}\{f\mid L_0^2(P_j^0)\}(\bar{z}_j,s) \nonumber \\
    &\quad+ \sum_{j\in\mathcal{J}} \mathbbm{1}_{\mathcal{S}_j}(s) \mathbbm{1}_{\bar{\mathcal{Z}}_{j-1}^\dagger}(\bar{z}_{j-1})\left[\Pi_{\uQ^0}\left\{\Gamma_j(f)\mid \mathcal{T}(\uQ^0,\mathcal{Q})\right\}(\bar{z}_j)-\Pi_{P^0}\{f\mid L_0^2(P_j^0)\}(\bar{z}_j,s)\right]. \label{eq:PifTP}
\end{align}
\end{lemma}
It is worth noting that Lemma~\ref{lem:Tproj} provides a means to compute the canonical gradient of arbitrary pathwise differentiable functional $\eta : \mathcal{P}\rightarrow\mathbb{R}$ within the data fusion model $\mathcal{P}$. Hence, this lemma may be of independent interest beyond the special setting that we consider in this paper, namely functionals of the form $\psi\circ \theta$ whose evaluation at $P^0$ corresponds to a summary of the target distribution $\uQ^0$. 
To see why Lemma~\ref{lem:Tproj} makes consideration of such functionals possible, note that this lemma provides a means to project an arbitrary function $f\in L_0^2(P^0)$ onto the tangent space of $\mathcal{P}$ at $P^0$. 
Therefore, given an arbitrary initial gradient of $\eta$, the canonical gradient of $\eta$ relative to $\mathcal{P}$ can be computed by projecting that gradient onto the tangent space of $\mathcal{P}$. 
A simple example of a parameter for which Lemma~\ref{lem:Tproj} can be a useful tool for computing the canonical gradient is the functional $P\mapsto E_P[Z_d]$. This functional does not generally take the form $\psi\circ \theta$ unless $\mathcal{J}=[d]$ and $\mathcal{S}_j=[k]$ for all $j$. Nevertheless, $z\mapsto z_d-E_P[Z_d]$ is an initial gradient, and so Lemma~\ref{lem:Tproj} provides a means to compute the canonical gradient of this functional in the data fusion model $\mathcal{P}$.
\begin{proof}[Proof of Lemma~\ref{lem:Tproj}]
Let $g : \mathcal{Z}\times [k]\rightarrow\mathbb{R}$ denote the function defined pointwise so that $g(z,s)$ is equal to the right-hand side of \eqref{eq:PifTP}. We will show that $g\in \mathcal{T}(P^0,\mathcal{P})$ and also that $\langle f-g, h\rangle_{P^0}=0$ for all $h\in \mathcal{T}(P^0,\mathcal{P})$, where $\langle\cdot,\cdot\rangle_{P^0}$ is the inner product in $L_0^2(P^0)$. We first show that $g\in \mathcal{T}(P^0,\mathcal{P})$. For all $j\in \{0\}\cup [d]$, it holds that $\Pi_{P^0}\{f \mid L_0^2(P_j^0)\}\in L_0^2(P_j^0)$. Similarly, for each $j\in\mathcal{J}$, $\Pi_{\uQ^0}\left\{\Gamma_j(f)\mid \mathcal{T}(\uQ^0,\mathcal{Q})\right\}\in \mathcal{T}(\uQ^0,\mathcal{Q})$. Recalling that $\mathcal{T}(P^0,\mathcal{P})=\bigoplus_{j=0}^d \mathcal{T}(P^0,\mathcal{P}_j)$, we see that $g\in \mathcal{T}(P^0,\mathcal{P})$.

The remainder of this proof shows that, for any $h\in \mathcal{T}(P^0,\mathcal{P})$, $\langle f-g, h\rangle_{P^0}=0$. Fix $h\in \mathcal{T}(P^0,\mathcal{P})$. As $L_0^2(P_j^0)$, $j=0,1,\ldots,d$, are orthogonal subspaces of $L_0^2(P^0)$ that are such that $L_0^2(P^0)=\bigoplus_{j=0}^d L_0^2(P_j^0)$, it holds that
\begin{align*}
    \langle f-g, h\rangle_{P^0}&= \int [f(z,s)-g(z,s)]h(z,s) dP^0(z,s) \\
    &= \sum_{j=0}^d \int \Pi_{P^0}\{f-g\mid L_0^2(P_j^0)\}(\bar{z}_j,s) \Pi_{P^0}\{h\mid L_0^2(P_j^0)\}(\bar{z}_j,s) dP^0(z,s).
\end{align*}
We show that each of the terms in the sum above is zero. If $j\in \{0\}\cup \mathcal{I}$, then this follows immediately from the fact that $\Pi_{P^0}\{g\mid L_0^2(P_j^0)\}=\Pi_{P^0}\{f\mid L_0^2(P_j^0)\}$, and so the corresponding term in the above sum is zero. Now suppose that $j\in \mathcal{J}$. We have that
\begin{align*}
    &\int \Pi_{P^0}\{f-g\mid L_0^2(P_j^0)\}(\bar{z}_j,s) \Pi_{P^0}\{h\mid L_0^2(P_j^0)\}(\bar{z}_j,s) dP^0(z,s) \\
    &= \int [1-\mathbbm{1}_{\mathcal{S}_j}(s) \mathbbm{1}_{\bar{\mathcal{Z}}_{j-1}^\dagger}(\bar{z}_{j-1})]\Pi_{P^0}\{f-g\mid L_0^2(P_j^0)\}(\bar{z}_j,s) \Pi_{P^0}\{h\mid L_0^2(P_j^0)\}(\bar{z}_j,s) dP^0(z,s) \\
    &\quad + \int \mathbbm{1}_{\mathcal{S}_j}(s) \mathbbm{1}_{\bar{\mathcal{Z}}_{j-1}^\dagger}(\bar{z}_{j-1})\Pi_{P^0}\{f-g\mid L_0^2(P_j^0)\}(\bar{z}_j,s) \Pi_{P^0}\{h\mid L_0^2(P_j^0)\}(\bar{z}_j,s) dP^0(z,s).
\end{align*}
When $(s,z)\not\in\mathcal{S}_j\times\bar{\mathcal{Z}}_{j-1}^\dagger$, it is straightforward to show that $\Pi_{P^0}\{g \mid L_0^2(P_j^0)\}(z,s)=\Pi_{P^0}\{f \mid L_0^2(P_j^0)\}(z,s)$. Hence, the first term on the right-hand side above is zero. We now study the second term. We begin by noting that $\mathcal{T}(P^0,\mathcal{P}_j)$ is a subspace of $L_0^2(P_j^0)$ and, for all $i\not=j$, $\mathcal{T}(P^0,\mathcal{P}_i)\cap L_0^2(P_j^0)=\{0\}$.
Hence, $\Pi_{P^0}\{h\mid L_0^2(P_j^0)\}\in \mathcal{T}(P^0,\mathcal{P}_j)$. Consequently, by \eqref{eq:TPj}, there exists an $r_j\in\mathcal{T}(\uQ^0,\mathcal{Q}_j)$ such that, whenever $(s,\bar{z}_{j-1})\in\mathcal{S}_j\times\bar{\mathcal{Z}}_{j-1}^\dagger$, $\Pi_{P^0}\{h\mid L_0^2(P_j^0)\}(\bar{z}_j,s)=r_j(\bar{z}_j)$. Thus, the second term above, which we refer to as (II), rewrites as
\begin{align*}
    \text{(II)}&= \int \mathbbm{1}_{\mathcal{S}_j}(s) \mathbbm{1}_{\bar{\mathcal{Z}}_{j-1}^\dagger}(\bar{z}_{j-1})\Pi_{P^0}\{f-g\mid L_0^2(P_j^0)\}(\bar{z}_j,s) r_j(\bar{z}_j) dP^0(z,s).
\end{align*}
Noting that $\Pi_{P^0}\{f-g\mid L_0^2(P_j^0)\}=\Pi_{P^0}\{f\mid L_0^2(P_j^0)\}-\Pi_{P^0}\{g\mid L_0^2(P_j^0)\}$ and using that $\Pi_{P^0}\{g\mid L_0^2(P_j^0)\}(\bar{z}_j,s)=\Pi_{\uQ^0}\left\{\Gamma_j(f)\mid \mathcal{T}(\uQ^0,\mathcal{Q})\right\}(\bar{z}_j)$ whenever $(s,\bar{z}_{j-1})\in\mathcal{S}_j\times\bar{\mathcal{Z}}_{j-1}^\dagger$, we see that
\begin{align*}
    \text{(II)}&= \int \mathbbm{1}_{\mathcal{S}_j}(s) \mathbbm{1}_{\bar{\mathcal{Z}}_{j-1}^\dagger}(\bar{z}_{j-1})\left[\Pi_{P^0}\{f\mid L_0^2(P_j^0)\}(\bar{z}_j,s) - \Pi_{\uQ^0}\left\{\Gamma_j(f)\mid \mathcal{T}(\uQ^0,\mathcal{Q})\right\}(\bar{z}_j)\right] \nonumber \\ 
    &\hspace{2.5em}\cdot r_j(\bar{z}_j) dP^0(z,s) \nonumber \\
    &= \int \mathbbm{1}_{\mathcal{S}_j}(s) \mathbbm{1}_{\bar{\mathcal{Z}}_{j-1}^\dagger}(\bar{z}_{j-1})\Big\{E_{P^0}\left[f(Z,S)\mid \bar{Z}_j=\bar{z}_j,S=s\right]-E_{P^0}\left[f(Z,S)\mid \bar{Z}_{j-1}=\bar{z}_{j-1},S=s\right]\nonumber \\
    &\hspace{11em}- \Pi_{\uQ^0}\left\{\Gamma_j(f)\mid \mathcal{T}(\uQ^0,\mathcal{Q})\right\}(\bar{z}_j)\Big\} r_j(\bar{z}_j) dP^0(z,s) \\
    &= \sum_{s\in\mathcal{S}_j} \int  \mathbbm{1}_{\bar{\mathcal{Z}}_{j-1}^\dagger}(\bar{z}_{j-1})\Big\{E_{P^0}\left[f(Z,S)\mid \bar{Z}_j=\bar{z}_j,S=s\right]-E_{P^0}\left[f(Z,S)\mid \bar{Z}_{j-1}=\bar{z}_{j-1},S=s\right] \\ 
    &\hspace{10em}- \Pi_{\uQ^0}\left\{\Gamma_j(f)\mid \mathcal{T}(\uQ^0,\mathcal{Q})\right\}(\bar{z}_j)\Big\} r_j(\bar{z}_j) P^0(S=s\mid \bar{Z}_j=\bar{z}_j) P_{\bar{Z}_j}^0(d\bar{z}_j),
\end{align*}
where $P_{\bar{Z}_j}^0$ denotes the marginal distribution of $\bar{Z}_j$ under sampling from $P^0$. Noting that, for all $s\in\mathcal{S}_j$,
\begin{align*}
    E_{P^0}\left[f(Z,S)\mid \bar{Z}_{j-1}=\bar{z}_{j-1},S=s\right]&=E_{P^0}\left[E_{P^0}\left\{f(Z,S)\mid \bar{Z}_{j},S\right\}\mid \bar{Z}_{j-1}=\bar{z}_{j-1},S\in\mathcal{S}_j\right] \\
    &= E_{P^0}\left[f(Z,S)\mid \bar{Z}_{j-1}=\bar{z}_{j-1},S\in\mathcal{S}_j\right]
\end{align*}
by the law of total expectation and Condition~\ref{cond:identifiability}, and also that
\begin{align*}
    &\sum_{s\in\mathcal{S}_j}  E_{P^0}\left[f(Z,S)\mid  \bar{Z}_j=\bar{z}_j,S=s\right] P^0(S=s\mid \bar{Z}_j=\bar{z}_j) \\
    &= E_{P^0}\left[f(Z,S)\mid  \bar{Z}_j=\bar{z}_j,S\in\mathcal{S}_j\right]P^0(S\in\mathcal{S}_j\mid \bar{Z}_j=\bar{z}_j),
\end{align*}
and $\sum_{s\in\mathcal{S}_j} P^0(S=s\mid \bar{Z}_j=\bar{z}_j) = P^0(S\in\mathcal{S}_j\mid \bar{Z}_j=\bar{z}_j)$, we see that the most recent expression for $\text{(II)}$ rewrites as
\begin{align*}
    \text{(II)}&= \int  \mathbbm{1}_{\bar{\mathcal{Z}}_{j-1}^\dagger}(\bar{z}_{j-1})\Big\{ E_{P^0}\left[f(Z,S)\mid \bar{Z}_j=\bar{z}_j,S\in\mathcal{S}_j\right] -E_{P^0}\left[f(Z,S)\mid \bar{Z}_{j-1}=\bar{z}_{j-1},S\in\mathcal{S}_j\right] \\ 
    &\hspace{9em}- \Pi_{\uQ^0}\left\{\Gamma_j(f)\mid \mathcal{T}(\uQ^0,\mathcal{Q})\right\}(\bar{z}_j)\Big\} r_j(\bar{z}_j) P^0(S\in\mathcal{S}_j\mid \bar{Z}_j=\bar{z}_j) P_{\bar{Z}_j}^0(d\bar{z}_j).
\end{align*}
Let $p_j^0\equiv P^0(S\in\mathcal{S}_j)$ and $\bar{P}_{j}^0(\cdot \mid \mathcal{S}_j)$ denotes the conditional distribution of $\bar{Z}_j$ given $S\in\mathcal{S}_j$ under sampling from $P^0$. Employing Bayes rule and Lemma~\ref{lem:LambdajIdent}, we see that
\begin{align*}
    \text{(II)}&= p_j^0 \int  \mathbbm{1}_{\bar{\mathcal{Z}}_{j-1}^\dagger}(\bar{z}_{j-1})\Big\{ E_{P^0}\left[f(Z,S)\mid \bar{Z}_j=\bar{z}_j,S\in\mathcal{S}_j\right] -E_{P^0}\left[f(Z,S)\mid \bar{Z}_{j-1}=\bar{z}_{j-1},S\in\mathcal{S}_j\right] \\ 
    &\hspace{9em}- \Pi_{\uQ^0}\left\{\Gamma_j(f)\mid \mathcal{T}(\uQ^0,\mathcal{Q})\right\}(\bar{z}_j)\Big\} r_j(\bar{z}_j)  \bar{P}_{j}^0(d\bar{z}_{j}\mid \mathcal{S}_j) \\
    &= p_j^0 \int  \mathbbm{1}_{\bar{\mathcal{Z}}_{j-1}^\dagger}(\bar{z}_{j-1})\Big\{\Gamma_j(f)(\bar{z}_j) - \Pi_{\uQ^0}\left\{\Gamma_j(f)\mid \mathcal{T}(\uQ^0,\mathcal{Q})\right\}(\bar{z}_j)\Big\} r_j(\bar{z}_j)  \bar{P}_{j}^0(d\bar{z}_{j}\mid \mathcal{S}_j) \\
    &= p_j^0 \int_{\bar{\mathcal{Z}}_{j-1}^\dagger} \int_{\mathcal{Z}_j}   \Big\{\Gamma_j(f)(\bar{z}_j) - \Pi_{\uQ^0}\left\{\Gamma_j(f)\mid \mathcal{T}(\uQ^0,\mathcal{Q})\right\}(\bar{z}_j)\Big\} \\
    &\hspace{7em}\times r_j(\bar{z}_j)  P_j^0(dz_j\mid \bar{z}_{j-1},\mathcal{S}_j)\bar{P}_{j-1}^0(d\bar{z}_{j-1}\mid \mathcal{S}_j),
\end{align*}
where we recall that $\bar{P}_{j-1}^0(\cdot\mid \mathcal{S}_j)$ is the conditional distribution of $\bar{Z}_{j-1}$ under $P^0$ given that $S\in \mathcal{S}_j$. Applying Conditions~\ref{cond:identifiability} and \ref{cond:posPart}, we see that
\begin{align*}
    \text{(II)}&= p_j^0 \int_{\bar{\mathcal{Z}}_{j-1}^\dagger} \int_{\mathcal{Z}_j}   \Big\{\Gamma_j(f)(\bar{z}_j) - \Pi_{\uQ^0}\left\{\Gamma_j(f)\mid \mathcal{T}(\uQ^0,\mathcal{Q})\right\}(\bar{z}_j)\Big\} r_j(\bar{z}_j)  \\
    &\hspace{7em}\cdot \uQ_j^0(dz_j\mid \bar{z}_{j-1})\bar{P}_{j-1}^0(d\bar{z}_{j-1}\mid \mathcal{S}_j) \\
    &= p_j^0 \int_{\bar{\mathcal{Z}}_{j-1}^\dagger} \int_{\mathcal{Z}_j}   \Big\{\Gamma_j(f)(\bar{z}_j) - \Pi_{\uQ^0}\left\{\Gamma_j(f)\mid \mathcal{T}(\uQ^0,\mathcal{Q})\right\}(\bar{z}_j)\Big\} \\
    &\hspace{7em} \times r_j(\bar{z}_j)  \uQ_j^0(dz_j\mid \bar{z}_{j-1})\lambda_{j-1}(\bar{z}_{j-1})^{-1}\bar{\uQ}_{j-1}^0(d\bar{z}_{j-1}).
\end{align*}
Using that $\Pi_{\uQ^0}\{\Gamma_j(f)\mid \mathcal{T}(\uQ^0,\mathcal{Q})\}(\bar{z}_j) = \lambda_{j-1}(\bar{z}_{j-1})\Pi_{\uQ^0}\{\Gamma_j(f)/\lambda_{j-1}\mid \mathcal{T}(\uQ^0,\mathcal{Q})\}(\bar{z}_j)$ and properties of orthogonal projections, we have
\begin{align*}
    \text{(II)}
    & = p_j^0 \int_{\bar{\mathcal{Z}}_j} \Big\{\lambda_{j-1}(\bar{z}_{j-1})^{-1} \Gamma_j(f)(\bar{z}_j) - \lambda_{j-1}(\bar{z}_{j-1})^{-1} \Pi_{\uQ^0}\left\{\Gamma_j(f)\mid \mathcal{T}(\uQ^0,\mathcal{Q})\right\}(\bar{z}_j)\Big\}r_j(\bar{z}_j) \uQ^0(d\bar{z}_{j})\\
      & = p_j^0 \int_{\bar{\mathcal{Z}}_j} \Big\{\lambda_{j-1}(\bar{z}_{j-1})^{-1} \Gamma_j(f)(\bar{z}_j) -\Pi_{\uQ^0}\left\{ \lambda_{j-1}(\bar{z}_{j-1})^{-1} \Gamma_j(f)\mid \mathcal{T}(\uQ^0,\mathcal{Q})\right\}(\bar{z}_j)\Big\}r_j(\bar{z}_j) \uQ^0(d\bar{z}_{j}) \\
      & =0.
\end{align*}
\end{proof}

\subsection{Equivalence of pathwise differentiability of $\psi$ and $\phi$}\label{sec:equivDiff}

We begin with a lemma regarding the nuisance tangent spaces $\mathcal{T}^\ddagger(\uQ^0,\mathcal{Q})$ and $\mathcal{T}^\ddagger(P^0,\mathcal{P})$ of $\mathcal{T}(\uQ^0,\mathcal{Q})$ and $\mathcal{T}(P^0,\mathcal{P})$ relative to $\psi$ and $\phi$, respectively. 
The nuisance tangent space $\mathcal{T}^\ddagger(\uQ^0,\mathcal{Q})$ is defined as the set of scores in  $\mathcal{T}(\uQ^0,\mathcal{Q})$ for which the target estimand remains constant, in first order, along a quadratic differentiable mean submodel with that score --- more formally, $\mathcal{T}^\ddagger(\uQ^0,\mathcal{Q})$ consists of all $h \in \mathcal{T}(\uQ^0,\mathcal{Q})$ for which there exists a univariate submodel $\{Q^{(\epsilon)} : \epsilon\in [0,\delta)\}$ with $Q^{(\epsilon)}=\uQ^0$ when $\epsilon=0$, score $h$ at $\epsilon=0$, and $\frac{\partial}{\partial \epsilon} \psi(Q^{(\epsilon)}) \mid_{\epsilon=0}=0$. Similarly, the nuisance tangent space $\mathcal{T}^\ddagger(P^0,\mathcal{P})$ consists of all $f\in \mathcal{T}(P^0,\mathcal{P})$ for which there exists a univariate submodel $\{P^{(\epsilon)} : \epsilon\in [0,\delta)\}$ with $P^{(\epsilon)}=P^0$ when $\epsilon=0$, score $f$ at $\epsilon=0$, and $\frac{\partial}{\partial\epsilon} \phi(P^{(\epsilon)})\mid_{\epsilon=0}=0$. Finally, for $j\in\mathcal{J}$, we let $\mathcal{T}(P^0,\mathcal{P}_j)$ be consist of all $f\in \mathcal{T}(P^0,\mathcal{P}_j)$ for which the restriction of $f$ to $\mathcal{Z}_j\times \bar{\mathcal{Z}}_{j-1}^\dagger\times \mathcal{S}_j$ is equal to zero.

In the upcoming results, we will use the fact that Condition~\ref{cond:var_indep} and Lemma~1.6 of \cite{van2003unified} imply that $\mathcal{T}(\uQ^0,\mathcal{Q})=\bigoplus_{i=1}^d \mathcal{T}(\uQ^0,\mathcal{Q}_i)$. We also define $\mathcal{U}(P^0,\mathcal{P}_j)$ to be the set of all $h_j\in L_0^2(P_j^0)$ that are such that $h(\bar{z}_j,s)=0$ for $P^0$-almost all $(\bar{z}_j,s)$ that are such that $(\bar{z}_{j-1},s)\in\bar{\mathcal{Z}}_{j-1}^\dagger\times\mathcal{S}_j$.

\begin{lemma}\label{lem:nTanSpace}
Suppose that Conditions~\ref{cond:identifiability} and \ref{cond:var_indep} hold. 
All of the following hold:
\begin{enumerate}[label=(\roman*)]
    \item\label{it:nTanSpaceQI} if $i\in\mathcal{I}$, then $\mathcal{T}(\uQ^0,\mathcal{Q}_i)\subseteq \mathcal{T}^\ddagger(\uQ^0,\mathcal{Q})$,
    \item\label{it:nTanSpacePI} if $i\in\{0\}\cup \mathcal{I}$, then  $\mathcal{T}(P^0,\mathcal{P}_i)\subseteq\mathcal{T}^\ddagger(P^0,\mathcal{P})$, and
    \item\label{it:nTanSpacePJ} if $j\in\mathcal{J}$, then $\mathcal{U}(P^0,\mathcal{P}_j)\subseteq \mathcal{T}^\ddagger(P^0,\mathcal{P})$.
\end{enumerate}
\end{lemma}
\begin{proof}[Proof of Lemma~\ref{lem:nTanSpace}]
If $\mathcal{I}$ is empty then \ref{it:nTanSpaceQI} and \ref{it:nTanSpacePI} are obvious. Hence, when proving those results, we suppose that $\mathcal{I}$ is nonempty.

We first prove \ref{it:nTanSpaceQI}. Fix $i\in \mathcal{I}$ and $f\in \mathcal{T}(\uQ^0,\mathcal{Q}_i)$. As $\mathcal{T}(\uQ^0,\mathcal{Q})=\bigoplus_{i=1}^d \mathcal{T}(\uQ^0,\mathcal{Q}_i)$ and $\mathcal{T}(\uQ^0,\mathcal{Q})$ was assumed to be a closed space, there exists a $\delta>0$ and a univariate submodel $\{\tilde{Q}^{(\epsilon)} : \epsilon\in [0,\delta)\}$ such that $\tilde{Q}^{(\epsilon)}=\tilde{\uQ}^0$ when $\epsilon=0$ and such that the model has score $f$ at $\epsilon=0$. By Condition~\ref{cond:var_indep}, we can further define $\{Q^{(\epsilon)} : \epsilon\in [0,\delta)\}\subseteq\mathcal{Q}$ to be such that, for each $\epsilon$, $Q_i^{(\epsilon)}=\tilde{Q}_i^{(\epsilon)}$ and, for all $j\not=i$, $Q_j^{(\epsilon)}=\uQ_j^0$. It can be readily verified that $Q^{(\epsilon)}=\uQ^0$ when $\epsilon=0$ and $\{Q^{(\epsilon)} : \epsilon\in [0,\delta)\}$ has score $f$ at $\epsilon=0$. Since $Q_j^{(\epsilon)}=\uQ_j^0$ for all $j\in\mathcal{J}\subseteq [d]\backslash \{i\}$ , the definition of $\mathcal{J}$ shows that $\psi(Q^{(\epsilon)})$ is constant over $\epsilon\in [0,\delta)$, and so $\frac{\partial}{\partial \epsilon} \psi(Q^{(\epsilon)})=0$. Hence, $f\in \mathcal{T}^\ddagger(\uQ^0,\mathcal{Q})$. As $f\in\mathcal{T}(\uQ^0,\mathcal{Q}_i)$ was arbitrary, $\mathcal{T}(\uQ^0,\mathcal{Q}_i)\subseteq \mathcal{T}^\ddagger(\uQ^0,\mathcal{Q})$.

We now prove \ref{it:nTanSpacePI}. Fix $i\in\{0\}\cup\mathcal{I}$ and $h\in\mathcal{T}(P^0,\mathcal{P}_i)$. By similar arguments to those used to prove \ref{it:nTanSpaceQI}, there exists $\{P^{(\epsilon)} : \epsilon\in [0,\delta)\}\subseteq \mathcal{P}$ with score $h$ at $\epsilon=0$ that is such that $P_j^{(\epsilon)}=P_j^0$ for all $\epsilon$ and $j\not=i$ and $P^{(\epsilon)}=P^0$ when $\epsilon=0$. Combining this with Condition~\ref{cond:identifiability} shows that $P_j^{(\epsilon)}(\cdot\mid \bar{z}_{j-1},S\in\mathcal{S}_j)=P_j^0(\cdot\mid \bar{z}_{j-1},S\in\mathcal{S}_j)=\uQ_j^0$ for all $j\in \mathcal{J}\subseteq [d]\backslash \{i\}$. Hence, for all $j\in\mathcal{J}$, the distribution of $\bar{Z}_j\mid \bar{Z}_{j-1}$ under $\theta(P^{(\epsilon)})$ is equal to $\uQ_j^0$. The definition of $\mathcal{J}$ then shows that $\phi(P^{(\epsilon)})=\psi\circ \theta(P^{(\epsilon)})$ is constant in $\epsilon$, and so $\frac{\partial}{\partial\epsilon} \phi(P^{(\epsilon)})\mid_{\epsilon=0}=0$. As $h\in\mathcal{T}(P^0,\mathcal{P}_i)$ was arbitrary, $\mathcal{T}(P^0,\mathcal{P}_i)\subseteq \mathcal{T}^\ddagger(P^0,\mathcal{P})$.

We now prove \ref{it:nTanSpacePJ}. Fix $j\in\mathcal{J}$ and $h\in \mathcal{U}(P^0,\mathcal{P}_j)$. First, note that $\mathcal{U}(P^0,\mathcal{P}_j)\subseteq\mathcal{T}(P^0,\mathcal{P}_j)$. Next, note that it is possible to construct a submodel $\{P^{(\epsilon)} : \epsilon\in [0,\delta)\}$ of $\mathcal{P}$ with score $h$ at $\epsilon=0$ that is such that $P^{(\epsilon)}=0$ when $\epsilon=0$ and $P_i^{(\epsilon)}=P_i^0$ for all $i\not=j$ --- in fact, the first part of the proof of Lemma~\ref{lem:TPj} provides such a construction (this can be seen by taking $f_j=0$ in the first part of that proof). Since $h_j\in\mathcal{U}(P^0,\mathcal{P}_j)$, $P_j^{(\epsilon)}(\cdot\mid \bar{z}_{j-1},s)$ to $P_j^0(\cdot\mid \bar{z}_{j-1},s)$ for $P^0$-almost all $(\bar{z}_{j-1},s)\in\bar{\mathcal{Z}}_{j-1}^\dagger\times \mathcal{S}_j$. Now, since $\bar{\mathcal{Z}}_{j-1}^\dagger$ denotes the support of $\bar{Z}_{j-1}$ under sampling from any $Q\in\mathcal{Q}$, it then must hold that the distribution of $Z_j\mid \bar{Z}_{j-1},S$ under $\theta(P^{(\epsilon)})$ is the same for all $\epsilon\in [0,\delta)$; also, for all $i\in \mathcal{J}\backslash \{j\}$, the distribution of $Z_i\mid \bar{Z}_{i-1},S$ under $\theta(P^{(\epsilon)})$ is the same for all $\epsilon\in [0,\delta)$ since $P_i^{(\epsilon)}=P_i^0$. Hence, by the definition of $\mathcal{J}$, $\phi(P^{(\epsilon)})=\psi\circ \theta(P^{(\epsilon)})$ is constant in $\epsilon$, and so $\frac{\partial}{\partial\epsilon} \phi(P^{(\epsilon)})\mid_{\epsilon=0}=0$. As $h\in\mathcal{U}(P^0,\mathcal{P}_j)$ was arbitrary, $\mathcal{U}(P^0,\mathcal{P}_j)\subseteq \mathcal{T}^\ddagger(P^0,\mathcal{P})$.
\end{proof}

Because the proofs are related, we prove Lemma~\ref{lem:pd} and Theorem~\ref{thm:grad} together.
\begin{proof}[Proofs of Lemma~\ref{lem:pd} and Theorem~\ref{thm:grad}]

We begin with a sketch of our proof. We will first suppose that $\psi$ is pathwise differentiable at $\uQ^0$ relative to $\mathcal{Q}$ and fix a gradient $D_{\uQ^0}$ of $\psi$. We will show that, for any submodel $\{P^{(\epsilon)} : \epsilon\in [0,\delta)\}$ with score $h\in \mathcal{T}(P^0,\mathcal{P})$ and with $P^{(\epsilon)}=P^0$ when $\epsilon=0$, it holds that $\frac{\partial}{\partial \epsilon}\phi(P^{(\epsilon)}) \mid_{\epsilon=0} = E_{P^0}\{D_{P^0}(Z,S)h(Z,S)\}$, where $D_{P^0}$ takes the form given in \eqref{eq:phiGrad}. This will show that $D_{P^0}$ is a gradient of $\phi$, which will complete the proof of Theorem~\ref{thm:grad} and the forward direction of Lemma~\ref{lem:pd}. It will then remain to prove the reverse direction of Lemma~\ref{lem:pd}, which we will provide in the latter half of this proof.

Suppose that $\psi$ is pathwise differentiable at $\uQ^0$ relative to $\mathcal{Q}$ and fix a gradient $D_{\uQ^0}$ of $\psi$ at $\uQ^0$ relative to $\mathcal{Q}$. Since $D_{\uQ^0}$ is a gradient, for any submodel $\{Q^{(\epsilon)} : \epsilon\in [0,\delta)\}$ with score $f\in \mathcal{T}(\uQ^0,\mathcal{Q})$ and with $Q^{(\epsilon)}=\uQ^0$ when $\epsilon=0$, it holds that $\frac{\partial}{\partial \epsilon}\psi(Q^{(\epsilon)}) \mid_{\epsilon=0} = E_{\uQ^0}\{D^*_{\uQ^0}(Z)f(Z)\}$. 
As $L_0^2(\uQ^0)=\bigoplus_{i=1}^d L_0^2(\uQ_j^0)$, there exist $D_{\uQ^0,j}\in L_0^2(\uQ_j^0)$, $j\in [d]$, such that $D_{\uQ^0}=\sum_{j=1}^d D_{\uQ^0,j}$ --- in particular, $D_{\uQ^0,j}(\bar{z}_j)=E_{\uQ^0}[D_{\uQ^0}(Z)\mid \bar{Z}_j=\bar{z}_j]-E_{\uQ^0}[D_{\uQ^0}(Z)\mid \bar{Z}_{j-1}=\bar{z}_{j-1}]$. Moreover, since gradients for $\psi$ reside in the orthogonal complement of the nuisance tangent space $\mathcal{T}^\ddagger(\uQ^0,\mathcal{Q})$, Lemma~\ref{lem:nTanSpace} shows that $D_{\uQ^0,i}=0$ for all $i\in\mathcal{I}$. 

Fix a function $h\in \mathcal{T}(P^0,\mathcal{P})$ and submodel $\{P^{(\epsilon)} : \epsilon\in [0,\delta)\}$. Since $\mathcal{T}(P^0,\mathcal{P})=\bigoplus_{j=0}^d \mathcal{T}(P^0,\mathcal{P}_j)$, there exist $h_j\in \mathcal{T}(P^0,\mathcal{P}_j)$, $j\in \{0\}\cup [d]$ , such that $h=\sum_{j=0}^d h_j$. Moreover, for each $j\in \mathcal{J}$, Lemma~\ref{lem:TPj} shows that there exists an $f_j\in \mathcal{T}(\uQ^0,\mathcal{Q}_j)$   such that $h_j(\bar{z}_j,s) =f_j(\bar{z}_j)$ for $(s,\bar{z}_{j-1}) \in \mathcal{S}_j \times \bar{Z}^{\dagger}_{j-1}$. 
For each $\epsilon\in [0,\delta)$, let $Q^{(\epsilon)}\in\mathcal{Q}$ be such that $Q_i^{(\epsilon)}=\uQ_i^0$ for all $i\in\mathcal{I}$ and, for all $j\in\mathcal{J}$, $Q_j^{(\epsilon)}(\cdot\mid \bar{z}_{j-1})=P_j^{(\epsilon)}(\cdot\mid \bar{z}_{j-1}, \mathcal{S}_j)$ for $\uQ^0$-almost all $\bar{z}_{j-1}\in\bar{\mathcal{Z}}_{j-1}^\dagger$. Clearly $Q^{(\epsilon)}=\uQ^0$ when $\epsilon=0$. Moreover, by analogous arguments to those given in the second part of the proof of Lemma~\ref{lem:TPj}, $\{Q^{(\epsilon)} : \epsilon\in [0,\delta)\}$ has score $\sum_{j\in\mathcal{J}} f_j$ at $\epsilon=0$. As $\psi$ is pathwise differentiable at $\uQ^0$ relative to $\mathcal{Q}$,
\begin{align*}
    \frac{\partial}{\partial\epsilon} \psi(Q^{(\epsilon)})\mid_{\epsilon=0}=E_{\uQ^0}\left\{D_{\uQ^0}(Z)\sum_{j\in\mathcal{J}} f_j(\bar{Z}_j)\right\}=E_{\uQ^0}\left\{\sum_{j\in\mathcal{J}} D_{\uQ^0,j}(\bar{Z}_j)f_j(\bar{Z}_j)\right\},
\end{align*}
where the latter equality used the orthogonality of the subspaces $L_0^2(\uQ_j^0)$ and $L_0^2(\uQ_i^0)$ when $i\not=j$. 
By the law of total expectation and Condition~\ref{cond:identifiability}, this shows that
\begin{align*}
     \frac{\partial}{\partial\epsilon} \psi(Q^{(\epsilon)})\mid_{\epsilon=0}&= E_{\uQ^0}\left[\sum_{j\in\mathcal{J}} E_{P^0}\left\{D_{\uQ^0,j}(\bar{Z}_j)f_j(\bar{Z}_j)\mid \bar{Z}_{j-1},S\in\mathcal{S}_j\right\}\right] \\
     &= E_{P^0}\left[\sum_{j\in\mathcal{J}} \lambda_{j-1}(\bar{Z}_{j-1}) E_{P^0}\left\{D_{\uQ^0,j}(\bar{Z}_j)f_j(\bar{Z}_j)\mid \bar{Z}_{j-1},S\in\mathcal{S}_j\right\}\mid S\in\mathcal{S}_j\right] \\
     &= E_{P^0}\left[\sum_{j\in\mathcal{J}} \lambda_{j-1}(\bar{Z}_{j-1}) D_{\uQ^0,j}(\bar{Z}_j)f_j(\bar{Z}_j)\mid S\in\mathcal{S}_j\right] \\
     &= E_{P^0}\left[\sum_{j\in\mathcal{J}} \frac{\mathbbm{1}(S\in\mathcal{S}_j)}{P(S\in\mathcal{S}_j)} \lambda_{j-1}(\bar{Z}_{j-1}) D_{\uQ^0,j}(\bar{Z}_j)f_j(\bar{Z}_j)\right].
\end{align*}
Now, by the construction of $Q^{(\epsilon)}$, it can be verified that, for all $j\in\mathcal{J}$, the distribution of $\bar{Z}_j\mid \bar{Z}_{j-1}$ under $\theta(P^{(\epsilon)})$ is equal to $Q_j^{(\epsilon)}$. Hence, for all $\epsilon\in [0,\delta)$, $\psi(Q^{(\epsilon)})=\phi(P^{(\epsilon)})$. Combining this with the fact that  the above shows that $\mathbbm{1}(s\in\mathcal{S}_j)\lambda_{j-1}(\bar{z}_{j-1})=0$ for $P^0$-almost all $(s,\bar{z}_{j-1})\not\in\mathcal{S}_j\times \bar{Z}_{j-1}^\dagger$, we see that
\begin{align*}
    \frac{\partial}{\partial\epsilon} \phi(P^{(\epsilon)})\mid_{\epsilon=0}&= E_{P^0}\left[\sum_{j\in\mathcal{J}} \frac{\mathbbm{1}(S\in\mathcal{S}_j)}{P(S\in\mathcal{S}_j)} \lambda_{j-1}(\bar{Z}_{j-1}) D_{\uQ^0,j}(\bar{Z}_j)f_j(\bar{Z}_j)\right] \\
    &= E_{P^0}\left[\sum_{j\in\mathcal{J}} \frac{\mathbbm{1}(S\in\mathcal{S}_j)}{P(S\in\mathcal{S}_j)} \lambda_{j-1}(\bar{Z}_{j-1}) D_{\uQ^0,j}(\bar{Z}_j)h_j(\bar{Z}_j,S)\right].
\end{align*}
 Using that $L_0^2(P_j^0)$ and $L_0^2(P_i^0)$ are orthogonal spaces for $i\not=j$ and also that $(z,s)\mapsto \frac{\mathbbm{1}(s\in\mathcal{S}_j)}{P(S\in\mathcal{S}_j)} \lambda_{j-1}(\bar{z}_{j-1}) D_{\uQ^0,j}(\bar{z}_j)\in L_0^2(P_j^0)$, where here we used Conditions~\ref{cond:identifiability} and \ref{cond:posPart} to ensure that this function has finite second moment, we see that $ \frac{\partial}{\partial\epsilon} \phi(P^{(\epsilon)})\mid_{\epsilon=0} = E_{P^0}\{D_{P^0}(Z,S)h(Z,S)\}$, where $D_{P^0}$ takes the form in \eqref{eq:phiGrad}. As $h\in \mathcal{T}(P^0,\mathcal{P})$ was arbitrary, $\phi$ is pathwise differentiable at $P^0$ relative to $\mathcal{P}$ with gradient $D_{P^0}$. 
This proves the forward direction of Lemma~\ref{lem:pd} and also proves Theorem~\ref{thm:grad}.

We now prove the other direction of Lemma~\ref{lem:pd}. Suppose that $\phi$ is pathwise differentiable at $P^0$ relative to $\mathcal{P}$ and let $D_{P^0}^*$ denote the canonical gradient of $\phi$. Fix a univariate submodel $\{Q^{(\epsilon)} : \epsilon\in [0,\delta)\}$ of $\mathcal{Q}$ that has score $f\in \mathcal{T}(\uQ^0,\mathcal{Q})$ and is such that $Q^{(\epsilon)}=\uQ^0$ when $\epsilon=0$. Since $f=\bigoplus_{j=1}^d \mathcal{T}(\uQ^0,\mathcal{Q}_j)$, it holds that $f=\sum_{j=1}^d f_j$, where $f_j$ is the projection of $f$ onto $\mathcal{T}(\uQ^0,\mathcal{Q}_j)$ in $L_0^2(\uQ^0)$. By Lemma~\ref{lem:TPj}, the fact that the tangent set of $\mathcal{P}$ at $P^0$ is a closed linear space, and the variation independence condition, there exists a submodel $\{P^{(\epsilon)} : \epsilon\in [0,\delta)\}$ with score $(z,s)\mapsto \sum_{j\in\mathcal{J}}\mathbbm{1}_{\mathcal{S}_j}(s) \mathbbm{1}_{\bar{\mathcal{Z}}_{j-1}^\dagger}(\bar{z}_{j-1}) f_j(\bar{z}_j)$ and $P^{(\epsilon)}=P^0$ when $\epsilon=0$. Hence, by the pathwise differentiability of $\phi$,
\begin{align*}
     \frac{\partial}{\partial \epsilon}\phi(P^{(\epsilon)}) \mid_{\epsilon=0}&= E_{P^0}\left[D_{P^0}^*(Z,S)\sum_{j\in\mathcal{J}}\mathbbm{1}_{\mathcal{S}_j}(S) \mathbbm{1}_{\bar{\mathcal{Z}}_{j-1}^\dagger}(\bar{Z}_{j-1}) f_j(\bar{Z}_j)\right].
\end{align*}
As $\mathcal{T}(P^0,\mathcal{P})=\bigoplus_{j=0}^d \mathcal{T}(P^0,\mathcal{P}_j)$ and the canonical gradient falls in both the tangent space $\mathcal{T}(P^0,\mathcal{P})$ and the orthogonal complement of the nuisance tangent space $\mathcal{T}^\ddagger(P^0,\mathcal{P})$, Lemmas~\ref{lem:TPj} and \ref{lem:nTanSpace} together show that there exist $D_{\uQ^0,j}\in\mathcal{T}(\uQ^0,\mathcal{Q}_j)$, $j\in\mathcal{J}$, such that $D_{P^0}^*$ takes the form $(z,s)\mapsto \sum_{j\in\mathcal{J}} \mathbbm{1}_{\mathcal{S}_j}(s)\mathbbm{1}_{\bar{\mathcal{Z}}_{j-1}^\dagger}(\bar{z}_{j-1}) D_{\uQ^0,j}(\bar{z}_j)$. Combining this with the above, the fact that $\frac{\partial}{\partial \epsilon}\psi(Q^{(\epsilon)}) \mid_{\epsilon=0}=\frac{\partial}{\partial \epsilon}\phi(P^{(\epsilon)}) \mid_{\epsilon=0}$ under Condition~\ref{cond:identifiability},  and the law of total expectation, we see that
\begin{align*}
     &\frac{\partial}{\partial \epsilon}\psi(Q^{(\epsilon)}) \mid_{\epsilon=0} \\
     &= E_{P^0}\left[\sum_{j'\in\mathcal{J}} \mathbbm{1}_{\mathcal{S}_{j'}}(S)\mathbbm{1}_{\bar{\mathcal{Z}}_{j'-1}^\dagger}(\bar{Z}_{j'-1}) D_{\uQ^0,j'}(\bar{Z}_{j'})\sum_{j\in\mathcal{J}}\mathbbm{1}_{\mathcal{S}_j}(S) \mathbbm{1}_{\bar{\mathcal{Z}}_{j-1}^\dagger}(\bar{Z}_{j-1}) f_j(\bar{Z}_j)\right] \\
     &= E_{P^0}\left[\sum_{j\in\mathcal{J}} \mathbbm{1}_{\bar{\mathcal{Z}}_{j-1}^\dagger}(\bar{Z}_{j-1}) E_{P^0}\{\mathbbm{1}_{\mathcal{S}_{j}}(S)D_{\uQ^0,j}(\bar{Z}_{j}) f_j(\bar{Z}_j)\mid \bar{Z}_{j-1},S\}\right] \\
     &\quad+ E_{P^0}\left[\sum_{j,j'\in\mathcal{J} : j<j'} \mathbbm{1}_{\bar{\mathcal{Z}}_{j'-1}^\dagger}(\bar{Z}_{j'-1}) E_{P^0}\{\mathbbm{1}_{\mathcal{S}_{j'}}(S)D_{\uQ^0,j'}(\bar{Z}_{j'})\mid \bar{Z}_{j'-1},S\}\mathbbm{1}_{\mathcal{S}_j}(S) \mathbbm{1}_{\bar{\mathcal{Z}}_{j-1}^\dagger}(\bar{Z}_{j-1}) f_j(\bar{Z}_j)\right] \\
     &\quad+ E_{P^0}\left[\sum_{j,j'\in\mathcal{J} : j>j'} \mathbbm{1}_{\mathcal{S}_{j'}}(S)\mathbbm{1}_{\bar{\mathcal{Z}}_{j'-1}^\dagger}(\bar{Z}_{j'-1}) D_{\uQ^0,j'}(\bar{Z}_{j'}) \mathbbm{1}_{\bar{\mathcal{Z}}_{j-1}^\dagger}(\bar{Z}_{j-1}) E_{P^0}\{\mathbbm{1}_{\mathcal{S}_j}(S)f_j(\bar{Z}_j)\mid \bar{Z}_{j-1},S\}\right].
\end{align*}
The expectations conditional on $(\bar{Z}_{j-1},S)$ in the latter two terms above are zero by Conditions~\ref{cond:identifiability} and the fact that functions in $L_0^2(\uQ_j^0)$ are $\uQ^0$-mean-zero for any $j$. Hence, the above display continues as
\begin{align*}
    \frac{\partial}{\partial \epsilon}\psi(Q^{(\epsilon)}) \mid_{\epsilon=0}&= E_{P^0}\left[\sum_{j\in\mathcal{J}} \mathbbm{1}_{\bar{\mathcal{Z}}_{j-1}^\dagger}(\bar{Z}_{j-1}) E_{P^0}\{\mathbbm{1}_{\mathcal{S}_{j}}(S)D_{\uQ^0,j}(\bar{Z}_{j}) f_j(\bar{Z}_j)\mid \bar{Z}_{j-1},S\}\right] \\
    &= E_{P^0}\left[\sum_{j\in\mathcal{J}} P^0(S\in\mathcal{S}_j) \mathbbm{1}_{\bar{\mathcal{Z}}_{j-1}^\dagger}(\bar{Z}_{j-1}) E_{P^0}\{D_{\uQ^0,j}(\bar{Z}_{j}) f_j(\bar{Z}_j)\mid \bar{Z}_{j-1},S\in\mathcal{S}_j\}\right] \\
    &= E_{P^0}\left[\sum_{j\in\mathcal{J}} P^0(S\in\mathcal{S}_j) \mathbbm{1}_{\bar{\mathcal{Z}}_{j-1}^\dagger}(\bar{Z}_{j-1}) E_{\uQ^0}\{D_{\uQ^0,j}(\bar{Z}_{j}) f_j(\bar{Z}_j)\mid \bar{Z}_{j-1}\}\right],
\end{align*}
where the final equality used Condition~\ref{cond:identifiability}. Applying Condition~\ref{cond:posPart} and the law of total expectation and using that $\bar{Z}_{j-1}$ has support $\bar{\mathcal{Z}}_{j-1}^\dagger$ under sampling from $\uQ^0$, we see that
\begin{align*}
      \frac{\partial}{\partial \epsilon}\psi(Q^{(\epsilon)}) \mid_{\epsilon=0}&= E_{\uQ^0}\left[\sum_{j\in\mathcal{J}} P^0(S\in\mathcal{S}_j) \mathbbm{1}_{\bar{\mathcal{Z}}_{j-1}^\dagger}(\bar{Z}_{j-1}) E_{\uQ^0}\{D_{\uQ^0,j}(\bar{Z}_{j}) f_j(\bar{Z}_j)\mid \bar{Z}_{j-1}\}\lambda_{j-1}(\bar{Z}_{j-1})^{-1}\right] \\
      &= E_{\uQ^0}\left[\sum_{j\in\mathcal{J}} P^0(S\in\mathcal{S}_j) D_{\uQ^0,j}(\bar{Z}_{j}) f_j(\bar{Z}_j)\lambda_{j-1}(\bar{Z}_{j-1})^{-1}\right]
\end{align*}
Since $z\mapsto D_{\uQ^0,j}(\bar{z}_j)\lambda_{j-1}(\bar{z}_{j-1})^{-1}$ and $f_j$ are both in $L_0^2(\uQ_j^0)$, $j\in\mathcal{J}$, and since $L_0^2(\uQ_j^0)$ and $L_0^2(\uQ_i^0)$ are orthogonal subspaces of $L_0^2(\uQ^0)$ when $i\not=j$, we see that
\begin{align*}
    \frac{\partial}{\partial \epsilon}\psi(Q^{(\epsilon)}) \mid_{\epsilon=0}&= E_{\uQ^0}\left[\left\{\sum_{j\in\mathcal{J}} P^0(S\in\mathcal{S}_j) D_{\uQ^0,j}(\bar{Z}_{j})\lambda_{j-1}(\bar{Z}_{j-1})^{-1}\right\} \sum_{j\in\mathcal{J}} f_j(\bar{Z}_j)\right] \\
    &= E_{\uQ^0}\left[\left\{\sum_{j\in\mathcal{J}} P^0(S\in\mathcal{S}_j) D_{\uQ^0,j}(\bar{Z}_{j})\lambda_{j-1}(\bar{Z}_{j-1})^{-1}\right\} \sum_{j=1}^d f_j(\bar{Z}_j)\right] \\
    &= E_{\uQ^0}\left[\left\{\sum_{j\in\mathcal{J}} P^0(S\in\mathcal{S}_j) D_{\uQ^0,j}(\bar{Z}_{j})\lambda_{j-1}(\bar{Z}_{j-1})^{-1}\right\} f(Z)\right].
\end{align*}
By Condition~\ref{cond:posPart}, $\lambda_{j-1}(\bar{Z}_{j-1})^{-1}$ is bounded, and so $z\mapsto \sum_{j\in\mathcal{J}} P^0(S\in\mathcal{S}_j) D_{\uQ^0,j}(\bar{z}_{j})\lambda_{j-1}(\bar{z}_{j-1})^{-1}$ belongs to $L_0^2(\uQ^0)$. As this function also does not depend on the arbitrarily chosen score $f\in\mathcal{T}(\uQ^0,\mathcal{Q})$, $\psi$ is pathwise differentiable at $\uQ^0$ relative to $\mathcal{Q}$.
\end{proof}

\begin{proof}[Proof of Corollary~\ref{cor:canGrad}]
Fix a gradient $D_{\uQ^0}$ of $\psi$. 
Recall the definition of $D_{P^0}$ from \eqref{eq:phiGrad}. We will show that the $L_0^2(P^0)$-projection of $D_{P^0}$ onto $\mathcal{T}(P^0,\mathcal{P})$ takes the form in \eqref{eq:phiCanGrad}, which establishes the desired result since projecting any gradient onto the tangent space yields the canonical gradient.

First, note that, by Lemma~\ref{lem:LambdajIdent}, we have that, for any $j\in\mathcal{J}$,
\begin{align*}
    \Gamma_j(D_{P^0})(\bar{z}_j)&= E_{P^0}\left[ \frac{\mathbbm{1}(S\in\mathcal{S}_j)}{P(S\in\mathcal{S}_j)} \lambda_{j-1}(\bar{Z}_{j-1}) D_{\uQ^0,j}(\bar{Z}_j)\mid \bar{Z}_j=\bar{z}_j,S\in\mathcal{S}_j\right] \\
    &\quad- E_{P^0}\left[ \frac{\mathbbm{1}(S\in\mathcal{S}_j)}{P(S\in\mathcal{S}_j)} \lambda_{j-1}(\bar{Z}_{j-1}) D_{\uQ^0,j}(\bar{Z}_j)\mid \bar{Z}_{j-1}=\bar{z}_{j-1},S\in\mathcal{S}_j\right] \\
    &=  \frac{\lambda_{j-1}(\bar{z}_{j-1})}{P(S\in\mathcal{S}_j)} D_{\uQ^0,j}(\bar{z}_j) \\
    &\quad-  \frac{1}{P(S\in\mathcal{S}_j)} \lambda_{j-1}(\bar{z}_{j-1}) E_{P^0}\left[D_{\uQ^0,j}(\bar{Z}_j)\mid \bar{Z}_{j-1}=\bar{z}_{j-1},S\in\mathcal{S}_j\right].
\end{align*}
The latter term above is zero since $D_{\uQ^0,j}\in L_0^2(\uQ_j^0)$ and, under Condition~\ref{cond:identifiability}, $E_{P^0}\left[D_{\uQ^0,j}(\bar{Z}_j)\mid \bar{Z}_{j-1}=\bar{z}_{j-1},S\in\mathcal{S}_j\right] = E_{\uQ^0}\left[D_{\uQ^0,j}(\bar{Z}_j)\mid \bar{Z}_{j-1}=\bar{z}_{j-1}\right]$. Hence,
\begin{align*}
    \Gamma_j(D_{P^0})(\bar{z}_j)&=  \frac{\lambda_{j-1}(\bar{z}_{j-1})}{P(S\in\mathcal{S}_j)} D_{\uQ^0,j}(\bar{z}_j).
\end{align*}
Moreover, as $D_{P^0}$ is a gradient of $\phi$ by Theorem~\ref{thm:grad}, and as gradients are orthogonal to the nuisance tangent space, Lemma~\ref{lem:nTanSpace} shows that $\Pi_{P^0}\{D_{P^0}\mid L_0^2(P_0^0)\}=0$ and, for $i\in\mathcal{I}$, $\Pi_{P^0}\{D_{P^0}\mid L_0^2(P_i^0)\}=0$. Combining this with the above shows that
\begin{align*}
    \Pi_{P^0}\{D_{P^0} \mid \mathcal{T}(P^0,\mathcal{P})\}(z,s)&= \sum_{j\in\mathcal{J}} \mathbbm{1}_{\mathcal{S}_j}(s) \mathbbm{1}_{\bar{\mathcal{Z}}_{j-1}^\dagger}(\bar{z}_{j-1})\Pi_{\uQ^0}\left\{\Gamma_j(D_{P^0})\mid \mathcal{T}(\uQ^0,\mathcal{Q})\right\}(\bar{z}_j) \\
    &= \sum_{j\in\mathcal{J}} \mathbbm{1}_{\bar{\mathcal{Z}}_{j-1}^\dagger}(\bar{z}_{j-1})\frac{\mathbbm{1}_{\mathcal{S}_j}(s)}{P^0(S\in\mathcal{S}_j)} \Pi_{\uQ^0}\left\{r_j\mid \mathcal{T}(\uQ^0,\mathcal{Q})\right\}(\bar{z}_j),
\end{align*}
where $r_j$ is as defined in the statement of the corollary.
\end{proof}

\section{Examples of projections onto the tangent spaces of semiparametric models}
\label{sec: projections_examples}
 
  The following examples provide the form of the projection in Corollary~\ref{cor:canGrad} in three interesting cases. Two further examples, namely where a conditional density is known to be symmetric or conditional mean function is known to be linear, can be found in the longitudinal treatment example (see Appendix~\ref{example:longitudinal treatment effect appendix}). We also refer the reader to \cite{bickel1993efficient}, \cite{van2003unified}, and \cite{tsiatis2006semiparametric} for further examples.
  
  \vspace{.75em}
  
  \noindent \textbf{Example 1. Conditional moment restriction.} Suppose $Q^0$ lies in the class $\mathcal{Q}$ containing each univariate distribution $Q$ satisfying the moment condition $E_{Q}[g_0(\bar{Z}_j)|\bar{Z}_{j-1}] = 0$, where, for some $j\in [d]$, $g_0 : \prod_{i=1}^j \mathcal{Z}_j\rightarrow\mathbb{R}$ is a known function. The projection of an arbitrary element $h \in L_0^2(Q_j^0)$ onto the tangent space of $\mathcal{Q}$ at $Q^0$ is given by
  $$\bar{z}_j \mapsto h(\bar{z}_j) - \frac{E_{Q^0}[g_0(\bar{Z}_j)h(\bar{Z}_j)\mid \bar{Z}_{j-1}=\bar{z}_{j-1}]}{E_{Q^0}[g_0(\bar{Z}_j)^2\mid \bar{Z}_{j-1}=\bar{z}_{j-1}]}g_0(\bar{z}_j).$$
  The calculations needed to establish the form of this projection are nearly identical to those used in Example 3.2.3 of \cite{bickel1993efficient} and so are omitted.\vspace{.75em}
  
  \noindent \textbf{Example 2. Repeated measures.}  Suppose that, for some $j\in [d]$, $Z_j=(Z_{j,1},\ldots,Z_{j,r})$ is a tuple of $r$ repeated measurements of a real-valued quantity. Further suppose that, for almost all realizations $\bar{z}_{j-1}$ of $\bar{Z}_{j-1}$, $Z_j\mid \bar{Z}_{j-1}=\bar{z}_{j-1}$ is an exchangeable random variable. The model $\mathcal{Q}$ denotes the collection of all distributions $Q$ for which $Q_j((z_{j,1},\ldots,z_{j,r})\mid \bar{z}_{j-1})=Q_j((z_{j,\pi(1)},z_{j,\pi(2)},\ldots,z_{j,\pi(r)})\mid \bar{z}_{j-1})$ for $Q$-almost all $\bar{z}_j$ and all $\pi$ belonging to the collection $\Pi$ of all permutations of $(1,2,\ldots,r)$. The projection of an arbitrary element $h \in L_0^2(Q_j^0)$ onto the tangent space of $\mathcal{Q}$ at $Q^0$ is given by
  $$\bar{z}_j \mapsto \frac{1}{r!}\sum_{\pi\in \Pi}h(\bar{z}_{j-1},(z_{j,\pi(1)},z_{j,\pi(2)},\ldots,z_{j,\pi(r)})).$$
  The calculations needed to establish the form of this projection follow the arguments used in Example 6.3.3 of \cite{bickel1993efficient}.\vspace{.5em}
  
  \noindent \textbf{Example 3. Conditional independencies as encoded by a directed acyclic graph.} Suppose it is known that the distribution of $Z=(Z_1,Z_2,\ldots,Z_d)$ respects the conditional independencies encoded by a directed acyclic graph with topological ordering $(Z_1,Z_2,\ldots,Z_d)$. More concretely, suppose that $\mathcal{Q}$ is the collection of all distributions $Q$ such that, for each $j\in [d]$, there is a known set $\mathcal{L}_j\subseteq [j-1]$ such that $Q(\cdot\mid \bar{Z}_{j-1})=Q(\cdot\mid {\rm pa}(Z_j))$ with $Q$-probability one, where ${\rm pa}(Z_j)= (Z_\ell)_{\ell\in \mathcal{L}_j}$. For any $j\in [d]$, the projection of an arbitrary $h \in L_0^2(Q_j^0)$ onto the tangent space of $\mathcal{Q}$ at $Q^0$ is given by
  \begin{align*}
      \bar{z}_j\mapsto E_{Q^0}\left[h(\bar{Z}_j)\mid Z_j=z_j,{\rm pa}(Z_j)=(z_\ell)_{\ell\in\mathcal{L}_j}\right] - E_{Q^0}\left[h(\bar{Z}_j)\mid {\rm pa}(Z_j)=(z_\ell)_{\ell\in\mathcal{L}_j}\right],
  \end{align*}
  See Lemma~9 and the proof of Theorem~6 in \cite{rotnitzky2019efficient} for details.

\section{Longitudinal treatment effect: detailed derivation}
\label{example:longitudinal treatment effect appendix}

\noindent \textbf{Derivation of influence functions:} 
\label{sec:derive_long}
\Citet{van2012targeted} gave the form of the canonical gradient under a locally nonparametric model $\mathcal{Q}$, which is $D_{\uQ^0}(z)  = D^{1}_{\uQ^0}(z) - D^{0}_{\uQ^0}(z) $ ,  where $D^{a'}_{\uQ^0}(z) = \sum_{t=1}^T D^{a'}_{\uQ^0,2t-1}(\bar{h}_t)$, and
\begin{align*}
    D^{a'}_{\uQ^0,2t-1}(\bar{h}_t) \equiv \left\{\prod_{m=1}^{t-1} \frac{\mathbbm{1}(a_m=a')}{\uQ^0(A_m=a' \mid \bar{U}_{m} = \bar{u}_m,\bar{A}_{m-1}=a')}\right\}\left\{L^{a'}_t(\bar{h}_t) - L^{a'}_{t-1}(\bar{h}_{t-1})\right\} .
\end{align*}
Following the results in Corollary~\ref{cor:canGrad}, the canonical gradient of $\phi$ under a locally nonparametric model $\mathcal{P}$ is, 
\begin{align*}
    D_{P^0}(x) & = \sum_{t=1}^T \mathbbm{1}(\bar{u}_{t-1} \in \bar{\mathcal{U}}^{\dagger}_{t-1})\frac{\mathbbm{1}(s \in \mathcal{S}_{2t-1})}{P^0(S \in \mathcal{S}_{2t-1})} \Pi_{\uQ^0}\{\lambda_{2t-2}D_{\uQ^0,2t-1} \mid \mathcal{T}(\uQ^0,\mathcal{Q})\}.
\end{align*}
Under a locally nonparametric model, $\mathcal{T}(Q_j^0,\mathcal{Q})\} = L_0^2 (\uQ^0_j)$. Hence, for each $t \in \{1,\ldots,T\}$, $\Pi_{\uQ^0}\{\lambda_{2t-2}D_{\uQ^0,2t-1} \mid \mathcal{T}(\uQ^0,\mathcal{Q})\}=\lambda_{2t-2}(\bar{h}_{t-1},a_{t-1})D_{\uQ^0,2t-1}(\bar{h}_{t})$. For each $t \in \{1,\ldots,T\}$, we substitute $\lambda_{2t-2}(\bar{h}_{t-1},a_{t-1}) = d\uQ^0(\bar{h}_{t-1},a_{t-1})/dP^0(\bar{h}_{t-1},a_{t-1} \mid S \in \mathcal{S}_{2t-1})$ and $D_{\uQ^0,2t-1}(\bar{h}_t)$ back into the expression for $  D_{P^0}(x) $. Abbreviating the event that $S\in\mathcal{S}_{r}$ by $\mathcal{S}_{r}$ and considering fixed $t\in \{1,\ldots,T\}$, $\bar{u}_{t-1} \in \bar{\mathcal{U}}^{\dagger}_{t-1}$, and $s \in \mathcal{S}_{2t-1}$, we note that
\begin{align*}
    \lambda_{2t-2}(\bar{h}_{t-1},a_{t-1})D_{\uQ^0,2t-1}(\bar{h}_{t})&=  f(\bar{h}_{t-1},\bar{a}_{t-1})\left\{L^{a'}_t(\bar{h}_t) - L^{a'}_{t-1}(\bar{h}_{t-1})\right\} \\
    &= f(\bar{h}_{t-1},\bar{a}_{t-1})\left\{\Tilde{L}^{a'}_t(\bar{h}_{t},s) -\Tilde{L}^{a'}_{t-1}(\bar{h}_{t-1},s) \right\},
\end{align*}
where the data fusion condition (Condition \ref{cond:identifiability}) shows that
\begin{align*}
    &f(\bar{h}_{t-1},\bar{a}_{t-1}) \equiv \frac{d\uQ^0(\bar{h}_{t-1},a_{t-1})}{dP^0(\bar{h}_{t-1},a_{t-1} \mid S \in \mathcal{S}_{2t-1})} \prod_{m=1}^{t-1} \frac{\mathbbm{1}(a_m=a')}{\uQ^0(A_m=a' \mid \bar{U}_{m} = \bar{u}_m,\bar{A}_{m-1}=a')} \\
    & =  \prod_{m=1}^{t-1} \frac{\mathbbm{1}(a_m = a')}{P^0(A_m=a' \mid \bar{u}_m,\bar{A}_{m-1}=a', \mathcal{S}_{2t-1})}\frac{ d\uQ^0(u_m \mid  \bar{U}_{m-1} = \bar{u}_{m-1},\bar{A}_{m-1} = a') }{dP^0(u_m \mid  \bar{U}_{m-1} = \bar{u}_{m-1},\bar{A}_{m-1} = a', \mathcal{S}_{2t-1})}  \\
    & = \prod_{m=1}^{t-1} \frac{\mathbbm{1}(a_m = a')}{P^0(A_m=a' \mid \bar{u}_m,\bar{A}_{m-1}=a', \mathcal{S}_{2t-1})}\frac{ dP^0( u_m \mid  \bar{U}_{m-1} = \bar{u}_{m-1},\bar{A}_{m-1} = a',  \mathcal{S}_{2m-1}) }{dP^0(u_m \mid  \bar{U}_{m-1} = \bar{u}_{m-1},\bar{A}_{m-1} = a', \mathcal{S}_{2t-1})} 
\end{align*}
Combining the above observations gives the form of the canonical gradient provided in the main text.

Under the semiparametric model where the conditional distribution $U_T \mid \bar{H}_{T-1},A_{T-1}$ is symmetric about $g(\bar{H}_{T-1},A_{T-1})$ for some unknown function $g(\cdot)$ (see Section~\ref{sec:examples} for more details), we show that $D^*_{P^0} = D^{*1}_{P^0} - D^{*0}_{P^0}$ where $D^{*a'}_{P^0}$ is specified in \eqref{eq:EIF_long_semi}, is indeed the canonical gradient of $\phi$ via the following three steps. First, we present the form of the tangent space of $\mathcal{Q}$ at $\uQ^0$ and providing the corresponding form of the projection. Second, we use the initial gradient under a locally nonparametric model and project it onto the tangent space $\mathcal{T}(\uQ^0, \mathcal{Q})$. Third, we use Corollary~\ref{cor:canGrad} to derive the canonical gradient of $\phi$.
 
To begin with, we let $\uq^0_{2T-1}(\cdot \mid \bar{h}_{T-1},a_{T-1})$ denote the conditional density of $U_T$ given that $\bar{H}_{T-1}=\bar{h}_{T-1}$ and $A_{T-1}=a_{T-1}$. We also let $\dot{\uq}^0_{2T-1}(u_T\mid \bar{h}_{T-1},a_{T-1})=\frac{\partial}{\partial u_T}\uq^0_{2T-1}(u_T \mid \bar{h}_{T-1},a_{T-1})$, $\ell(z) \equiv \dot{\uq}^0_{2T-1}(u_T \mid \bar{h}_{T-1},a_{T-1})/\uq^0_{2T-1}(u_T \mid \bar{h}_{T-1},a_{T-1})$, and $I_{2T-1}(\bar{h}_{T-1},a_{T-1}) \equiv \int \dot{\uq}^0_{2T-1}(u_T \mid \bar{h}_{T-1},a_{T-1})^2 / \uq^0_{2T-1} (u_T \mid \bar{h}_{T-1},a_{T-1}) du_T$. 
Similarly, we let $p^0_{2T-1}(\,\cdot \mid \bar{h}_{T-1},a_{T-1}, s \in \mathcal{S}_{2T-1})$ denotes the conditional density of $Y$ given that $\bar{H}_{T-1}=\bar{h}_{T-1}$, $A_{T-1}=a_{T-1}$ and $S \in \mathcal{S}_{2T-1}$, $\dot{p}^0_{2T-1}(y\mid \bar{h}_{T-1},a_{T-1}, s \in \mathcal{S}_{2T-1})=\frac{\partial}{\partial y}p^0_{2T-1}(y \mid \bar{h}_{T-1},a_{T-1}, s \in \mathcal{S}_{2T-1})$,  $\tilde{\ell}(x) \equiv \dot{p}^0_{2T-1}(y \mid \bar{h}_{T-1},a_{T-1}, s \in \mathcal{S}_{2T-1})/p^0_{2T-1}(y \mid \bar{h}_{T-1},a_{T-1}, s \in \mathcal{S}_{2T-1})$, and $\tilde{I}_{2T-1} \equiv  \int \dot{p}^0_{2T-1}(y \mid \bar{h}_{T-1},a_{T-1},s \in \mathcal{S}_{2T-1})^2 / p^0_{2T-1} (y \mid \bar{h}_{T-1},a_{T-1},s \in \mathcal{S}_{2T-1}) dy$.By Condition~\ref{cond:var_indep}, the tangent space writes as $\mathcal{T}(\uQ^0, \mathcal{Q}) = \bigoplus_{j=1}^{2T-2} L_0^2(\uQ^0_j) + \mathcal{T}(\uQ^0_{2T-1},\mathcal{Q}_{2T-1})$, where $\mathcal{T}(\uQ^0_{2T-1},\mathcal{Q}_{2T-1}) = \mathcal{T}_1(\uQ^0_{2T-1},\mathcal{Q}_{2T-1}) \bigoplus \mathcal{T}_2(\uQ^0_{2T-1},\mathcal{Q}_{2T-1})$ with $\mathcal{T}_1(\uQ^0_{2T-1},\mathcal{Q}_{2T-1})$ being equal to the $L_0^2(\uQ^0)$-closure of $\{z \mapsto c(\bar{h}_{T-1}, a_{T-1}) \ell(z) \textnormal{ for any bounded function } c(\cdot) \}$ and, letting $\tilde{u}_T=2g(\bar{h}_{T-1},a_{T-1})-u_T$,
\begin{align*}
    \mathcal{T}_2(\uQ^0_{2T-1},\mathcal{Q}_{2T-1}) = \Bigg\{z \mapsto l(u_T,\bar{h}_{T-1},a_{T-1}): &z\mapsto l(u_T) \in L_0^2(\uQ_{2T-1}^0), \textnormal{ where } \\
    &l(u_T,\bar{h}_{T-1},a_{T-1})=l(\tilde{u}_T,\bar{h}_{T-1},a_{T-1})\Bigg\}.
\end{align*}
The proof of this representation for the tangent space, which is omitted, follows similar arguments to those for the univariate symmetric case \cite[Example~3.2.4][]{bickel1993efficient}.

The projections of any $f \in L_0^2(\uQ^0_{2T-1})$ onto $\mathcal{T}_1(\uQ^0_{2T-1},\mathcal{Q}_{2T-1})$ and $\mathcal{T}_2(\uQ^0_{2T-1},\mathcal{Q}_{2T-1})$ have the following pointwise evaluations:
 \begin{align}
      \Pi_{\uQ^0}\{f \mid \mathcal{T}_1(\uQ^0_{2T-1},\mathcal{Q}_{2T-1})\}(z) &= \frac{E_{\uQ^0}\left[f(Z) \ell (Z) \mid \bar{H}_{T-1} = \bar{h}_{T-1}, A_{T-1}=a_{T-1}\right]\ell (z) }{I_{2T-1}(\bar{h}_{T-1},a_{T-1})}, \label{eq:center_proj} \\
     \Pi_{\uQ^0}\{f \mid \mathcal{T}_2(\uQ^0_{2T-1},\mathcal{Q}_{2T-1})\} (z)&= \frac{f(z)  + f(\tilde{z})}{2}. \label{eq:sym_proj} 
 \end{align}
 where $\tilde{z}\equiv (\bar{h}_{T-1},a_{T-1},\tilde{u}_t)$. In the special case of a univariate symmetric location model, the forms of these projections can be found in Example~3.3.1 and 3.2.4 of \cite{bickel1993efficient}.

Now we are at the last step and will use Corollary~\ref{cor:canGrad} to derive the canonical gradient of $\phi$. We let $ \Pi_{\uQ^0}\{\lambda_{2T-2}D_{\uQ^0,2T-1} \mid \mathcal{T}(\uQ^0_{2T-1},\mathcal{Q}_{2T-1})\}(z)   = \textnormal{(I) + (II)}$, where 
\begin{align}
    \textnormal{(I)}& = \lambda_{2T-2}(\bar{h}_{T-1},a_{T-1})\frac{E_{\uQ^0}\left[D_{\uQ^0,2T-1}(Z) \ell (Z) \mid \bar{H}_{T-1}=\bar{h}_{T-1},A_{T-1} =a_{T-1} \right]\ell (z) }{ I_{2T-1}(\bar{h}_{T-1},a_{T-1})},\\
    \textnormal{(II)} & = \lambda_{2T-2}(\bar{h}_{T-1},a_{T-1}) \frac{\{D_{\uQ^0,2T-1}(z) + D_{\uQ^0,2T-1}(\tilde{z}) \}}{2}.
\end{align}

Following Corollary~\ref{cor:canGrad} and denoting the canonical gradient of $\phi$ under such semiparametric model as $D^*_{P^0}(x) $, we have
\begin{align*}
    D^{*}_{P^0}(x) & = \sum_{t=1}^T \mathbbm{1}(\bar{u}_{t-1} \in \bar{\mathcal{U}}^{\dagger}_{t-1})\frac{\mathbbm{1}(s \in \mathcal{S}_{2t-1})}{\textnormal{pr}(S \in \mathcal{S}_{2t-1})} \Pi_{\uQ^0}\{\lambda_{2t-2}D_{\uQ^0,2t-1} \mid \mathcal{T}(\uQ^0,\mathcal{Q})\}(z)\\
    & = D_{P^0}(x) - D_{P_{2T-1}^0}(x) +\mathbbm{1}(\bar{u}_{T-1} \in \bar{\mathcal{U}}^{\dagger}_{T-1})\frac{\mathbbm{1}(s \in \mathcal{S}_{2T-1})}{P^0(S \in \mathcal{S}_{2T-1})} \lambda_{2T-2}(\bar{h}_{T-1},a_{T-1})  \\
    &\quad \cdot\Bigg\{ \frac{D_{\uQ^0,2T-1}(z) + D_{\uQ^0,2T-1}(\tilde{z}) }{2} \\
    &\hspace{3em}+ 
  \frac{E_{\uQ^0}\left[D_{\uQ^0,2T-1}(Z) \ell (Z) \mid \bar{H}_{T-1} = \bar{h}_{T-1}, A_{T-1} =a_{T-1}\right]\ell (z) }{ I_{2T-1}(\bar{h}_{T-1},a_{T-1})}\Bigg\}.
\end{align*}
We can use Condition~\ref{cond:identifiability} to replace features of $\uQ^0$ in the expression above by the corresponding features of $P^0$. By the definition of $D_{P^0_{2t-1}}(x)$, the above can then be simplified to 
\begin{align}
    D^{*}_{P^0}(x) & = 
    D^{}_{P^0}(x)  - D^{}_{P_{2T-1}^0}(x)
    + \frac{E_{P^0}\left\{D^{}_{P_{2T-1}^0}(X) \tilde{\ell} (X) \mid \bar{h}_{T-1}, a_{T-1}, S \in \mathcal{S}_{2T-1}\right\} \tilde{\ell} (x)}{ \tilde{I}_{2T-1}(\bar{h}_{T-1},a_{T-1})}. \label{eq:EIF_long_semi} 
\end{align}

We conclude by deriving the form of the canonical gradient in the semiparametric model where $U_T = \beta^\top \kappa(\bar{H}_{T-1},A_{T-1}) + \epsilon$ under nested fusion sets--- see Section~\ref{sec:examples} for more details. The tangent space within this semiparametric model takes the form $\mathcal{T}(\uQ^0, \mathcal{Q}) = \bigoplus_{j=1}^{2T-2} L_0^2(\uQ^0_j) + \mathcal{T}(\uQ^0_{2T-1},\mathcal{Q}_{2T-1})$, where, $\ell$ is defined in Section~\ref{sec:estimation} and $\mathcal{T}(\uQ^0_{2T-1},\mathcal{Q}_{2T-1})=\left\{z\mapsto m^{\top} \ell(z) : m\in\mathbb{R}^{c}\right\}$. It can be verified that 
\begin{align}
    \Pi_{\uQ^0}\{f \mid \mathcal{T}(\uQ^0_{2T-1},\mathcal{Q}_{2T-1})\} (z)&=  \{E_{\uQ^0}[\ell(Z)\ell(Z)^\top]^{-1} E_{\uQ^0}[\ell(Z) f(Z)]\}^\top \ell(z). \label{eq:semi_gaussian_proj}
\end{align}
Following Corollary~\ref{cor:canGrad} and denoting the canonical gradient of $\phi$ under the semiparametric model under consideration as $D^{\dagger}_{P^0}(x) $, we have
\begin{align*}
    D^{\dagger}_{P^0}(x)  & = \sum_{t=1}^T \mathbbm{1}(\bar{u}_{t-1} \in \bar{\mathcal{U}}^{\dagger}_{t-1})\frac{\mathbbm{1}(s \in \mathcal{S}_{2t-1})}{\textnormal{pr}(S \in \mathcal{S}_{2t-1})} \Pi_{\uQ^0}\{\lambda_{2t-2}D_{\uQ^0,2t-1} \mid \mathcal{T}(\uQ^0,\mathcal{Q})\}(z)\\
    & = D_{P^0}(x) - D_{P_{2T-1}^0}(x) +\mathbbm{1}(\bar{u}_{T-1} \in \bar{\mathcal{U}}^{\dagger}_{T-1})\frac{\mathbbm{1}(s \in \mathcal{S}_{2T-1})}{P^0(S \in \mathcal{S}_{2T-1})}  \\
    &\quad \cdot  \{E_{\uQ^0}[\ell(Z)\ell(Z)^\top]^{-1} E_{\uQ^0}[\ell(Z)\lambda_{2T-2}(\bar{H}_{T-1},A_{T-1})  D_{\uQ^0,2T-1}(Z)]\}^\top \ell(z).
\end{align*}
We can use Condition~\ref{cond:identifiability} to replace features of $\uQ^0$ in the expression above by corresponding features of $P^0$.

\section{Additional Examples}
\label{additional examples}
\subsection{Intent-to-treat average treatment effect}
\label{example:ITT}
Primary analyses in randomized clinical trials often concern the intent-to-treat average treatment effect. This estimand corresponds to the difference between mean outcome observed of individuals randomized to treatment versus control, regardless of what intervention they actually receive. Let $Z_1$ denote some baseline characteristic variable, $Z_2$ be the binary randomized treatment assignment, $Z_3$ be an indicator of actually receiving treatment, and $Z_4$ be the real-valued outcome of interest. The model $\mathcal{Q}$ for the unknown target distribution $Q$ consists of all distributions with some common support that are such that treatment assignment is randomized, that is, $Z_2$ is independent of $Z_1$. The intent-to-treat average treatment effect of a distribution $Q\in\mathcal{Q}$ is defined as $\psi(Q)\equiv E_{Q}(Z_4|Z_2=1)  - E_{Q}(Z_4|Z_2=0)$. By leveraging the randomization of treatment assignment and the law of total expectation, it can be seen that $\psi(Q)= \sum_{a=0}^1 (2a-1) E_{Q_1} [E_{Q_3}\{E_{Q_4}(Z_4\mid Z_3,Z_2=a,Z_1)\mid Z_2=a,Z_1\}]$, where here and throughout we write $E_{Q_j}$, rather than $E_Q$, when we want to emphasize that a conditional expectation only depends on the conditional distribution $Q_j$, rather than on the whole distribution $Q$. Because $\psi(Q)$ can be written as a function of $Q_1$, $Q_3$, and $Q_4$ only, it is evident that we can take $\mathcal{I}=\{2\}$ in this example.

Suppose we observe data from $k$ sources of three types. 
The first type of data source only contains covariate information $Z_1$, while randomization, treatment, and outcome information are missing. To indicate such missingness, we let $Z_2=Z_3=Z_4=\star$ for data from sources of this type. Despite this systematic missingness, it is still possible that such data sources belong to $\mathcal{S}_1$ since $\mathcal{S}_1$ only pertains to the marginal distribution of $Z_1$. The second type of data source also comes from a clinical trial setting but does not have relevant outcome information measured, so that $(Z_1,Z_2,Z_3)$ is measured and $Z_4=\star$. The third type of data source comes from a clinical trial setting and has all relevant variables, including outcomes, measured. Data from the first type of source may inform about the covariate distribution in the target population, data from the second and third may inform about the propensity to adhere to a treatment assignment, and data from the third may also inform about the probability of experiencing a particular outcome given treatment and covariate information.

Under Condition~\ref{cond:identifiability}, Theorem~\ref{thm:identifiability} shows that the intent-to-treat average treatment effect on the target population $Q^0$ can be identified from the observed data distribution --- in particular, that $\psi(Q^0)=\phi(P^0)$. In this example, $\phi(P^0)$ takes the following form:
\begin{align*}
    \phi(P^0) & = \sum_{a=0}^1 (2a-1)E_{P^0}\left[ E_{P^0}\{E_{P^0}(Z_4\,|\, Z_3,Z_2=a,Z_1, S \in \mathcal{S}_4)  \,|\,  Z_2=a,Z_1, S \in \mathcal{S}_3\}\,|\,  S \in \mathcal{S}_1\right].
\end{align*}
\cite{rudolph2017robust} considered this problem in the case where $k=2$ data sources are available.
Our work makes it possible to incorporate data from more than two sources.
\vspace{.5em}\noindent \textbf{Comparison of identifiability conditions:}
Assumption 1 in Section 3 of \cite{rudolph2017robust} is a weaker version of our Condition~\ref{cond:identifiability}\ref{suff_alignment} in that it only imposes exchangeability on the conditional means across data sources. Assumption 3 in the same section of that work corresponds to our Condition~\ref{cond:identifiability}\ref{suff_overlap}. The fact that \cite{rudolph2017robust} were able to establish the identifiability of $\psi(Q^0)$ with $\phi(P^0)$ under a slight weakening of our Condition~\ref{cond:identifiability}\ref{suff_alignment} suggests that, in some cases, studying a particular functional of interest $\psi$ can make it possible to state a weaker identifiability condition than the sufficient one that we stated in Section~\ref{sec:notations and setup}. Nevertheless, it is worth noting that the estimation problem resulting from this weakened identifiability condition, namely estimation of $\phi(P^0)$, is the same as the one that we consider in this work. Hence, the general framework that we provide still allows for the construction of an efficient estimator in this setting, even if the sufficient conditions for its validity as an estimator of the target estimand $\psi(Q^0)$ are slightly stronger than is, in fact, necessary.

\vspace{.5em}\noindent \textbf{Derivation of influence functions:}
This result is a natural extension of those of \cite{rudolph2017robust} to more than two data sources and hence we omit the derivation and present the results only.
The canonical gradient of $\phi$ at $P^0$ relative to the model that makes at most assumptions about propensity score and positivity is given by
\begin{align*}
    D_{P^0}(x) & = D^1_{P^0}(x) - D^0_{P^0}(x),
\end{align*}
where
\begin{align*}
    D^{z_2}_{P^0}(x) & = D^{z_2}_{P^0,1}(x) +D^{z_2}_{P^0,3}(x) +D^{z_2}_{P^0,4}(x),
\end{align*}
with
\begin{align*}
    D^{z_2'}_{P^0,1}(x) & = \frac{1(s \in \mathcal{S}_1)}{P^0(S \in \mathcal{S}_1)}\bigg\{E_{P^0}\left\{E_{P^0}[Z_4\mid Z_3,Z_2=z_2',Z_1,S \in \mathcal{S}_4] \mid Z_2=z_2',Z_1=z_1, S\in \mathcal{S}_3\right\}\\
    & \hspace{2em}- E_{P^0}\left\{E_{P^0}\left\{E_{P^0}[Z_4\mid Z_3,Z_2=z_2',Z_1,S \in \mathcal{S}_4] \mid Z_2=z'_2,Z_1, S\in \mathcal{S}_3\right\} \mid S \in \mathcal{S}_1\right\}\bigg\}. \\
    D^{z_2'}_{P^0,3}(x) & = \frac{1(s \in \mathcal{S}_3, z_2=z_2')}{P^0(S \in \mathcal{S}_3)P^0(Z_2=z_2'\mid Z_1=z_1,S\in \mathcal{S}_3)}\frac{ dP^0(z_1 \mid S \in \mathcal{S}_1 )}{dP^0(z_1 \mid S\in \mathcal{S}_3 )} \\
    & \bigg\{E_{P^0}[Z_4\mid Z_3=z_3,Z_2=z_2',Z_1=z_1,S \in \mathcal{S}_4] -\\
    & \hspace{6em} E_{P^0}\left\{E_{P^0}[Z_4\mid Z_3,Z_2=z_2',Z_1,S \in \mathcal{S}_4] \mid Z_2=z_2',Z_1=z_1, S\in \mathcal{S}_3\right\}\bigg\}. \\
    D^{z_2'}_{P^0,4}(x) & = \frac{1(s \in \mathcal{S}_4, z_2=z_2')}{P^0(S \in \mathcal{S}_4)P^0(Z_2=z_2'\mid Z_1=z_1,S\in \mathcal{S}_4)}\frac{dP^0(z_3 \mid Z_2=z_2',Z_1=z_1,S \in \mathcal{S}_3) }{dP^0(z_3 \mid Z_2=z_2',Z_1=z_1,S \in \mathcal{S}_4) } \\
    & \quad \frac{dP^0(z_1 \mid S \in \mathcal{S}_1 )}{dP^0(z_1 \mid S \in \mathcal{S}_4)} \left\{z_4 - E_{P^0}[Z_4\mid Z_3=z_3,Z_2=z_2',Z_1=z_1,S \in \mathcal{S}_4]\right\} .
\end{align*}

\subsection{Complier average treatment effect}
\label{example:CATE}
\noindent \textbf{Parameter of interest:} Suppose that, in the setting of the example in  Appendix~\ref{example:ITT}, we want to measure the impact of an intervention only in the population that complies with its assigned treatment. This quantity of interest is known as the complier average treatment effect \citep{angrist1996identification}, which is defined as, 
\begin{align*}
    \psi(Q) &  = \frac{E_Q(Z_4\mid Z_2=1) - E_Q(Z_4\mid Z_2=0 )}{E_Q(Z_3\mid Z_2=1)-E_Q(Z_3\mid Z_2=0)} \\
    &= \frac{\sum_{a=0}^1 (2a-1) E_{Q_1} [E_{Q_3}\{E_{Q_4}(Z_4\mid Z_3,Z_2=a,Z_1)\mid Z_2=a,Z_1\}]}{\sum_{a=0}^1 (2a-1) E_{Q_1} \{E_{Q_3}(Z_3 \mid Z_2=a,Z_1)\}} \\
    &\equiv \frac{\psi_{ITT}(Q)}{\psi_{c}(Q)},
\end{align*}
where, $\psi_{ITT}$ is the intent-to-treat average treatment effect in  Appendix~\ref{example:ITT} and $\psi_{c}$ measures the proportion of compliers. As in  Appendix~\ref{example:ITT}, the model $\mathcal{Q}$ for the unknown target distribution $Q$ consists of all distributions with some common support that are such that treatment assignment is randomized. Because $\psi$ can be written as a function of $Q_1$, $Q_3$, and $Q_4$ only, we see that we can take $\mathcal{I}=\{2\}$ in this example. Suppose we observe data from $k$ sources and consider the same setting as the one in  Appendix~\ref{example:ITT}. Then under Condition~\ref{cond:identifiability}, the complier average treatment effect on the target population $Q^0$ can be identified as $\phi(P^0)\equiv \phi_{ITT}(P^0) / \phi_{c}(P^0)$ where $\phi_{ITT}(P^0)$ is equal to $\phi(P^0)$ in  Appendix~\ref{example:ITT} and $\phi_{c}(P^0)$ is equal to
\begin{align*}
    \phi_{c}(P^0) & = \sum_{a=0}^1 (2a-1) E_{P^0}\{E_{P^0}(Z_3\,|\, Z_2=a,Z_1, S \in \mathcal{S}_3)  \,|\ Z_1, S \in \mathcal{S}_1\}.
\end{align*}
\cite{rudolph2017robust} considered this problem in the case where $k=2$ data sources are available. Our work makes it possible to incorporate data from more than two sources.

\vspace{.5em}\noindent \textbf{Comparison of identifiability conditions:}
Assumption 1 in Section 5 of \cite{rudolph2017robust} is a weaker version of our Condition~\ref{cond:identifiability}\ref{suff_alignment} in that it only imposes exchangeability on the conditional means across data sources. Assumption 5 in the same section of that work corresponds to our Condition~\ref{cond:identifiability}\ref{suff_overlap}. See our discussion of the identifiability conditions for the intent-to-treat average treatment effect for more discussion on the weaker form of Condition~\ref{cond:identifiability}\ref{suff_alignment} presented in \cite{rudolph2017robust}.

\vspace{.5em}\noindent \textbf{Derivation of influence functions:} Chapter~5.2.3 of \cite{van2011targeted} shows that the canonical gradient of $\psi_c$ is $\tilde{C}_{Q^0}(z) =\tilde{C}^1_{Q^0}(z)-\tilde{C}^0_{Q^0}(z)$, where
\begin{align*}
    \tilde{C}^{z_2'}_{Q^0}(z) & = \frac{1(z_2=z_2')}{Q^0(Z_2=z_2'\mid Z_1 =z_1)}\{z_3 - E_{Q^0}(Z_3 \mid Z_2 =z_2',Z_1 =z_1)\} \nonumber\\
    & \quad + E_{Q^0}(Z_3 \mid Z_2 = z_2',Z_1 =z_1) - E_{Q^0}\left[E_{Q^0}(Z_3 \mid Z_2 = z_2',Z_1)\right].
\end{align*}
Plugging this into \eqref{eq:phiCanGrad} shows that the canonical gradient of $\phi_{c}$ is $C_{P^0}(x) =C^1_{P^0}(x)-C^0_{P^0}(x)$, where
\begin{align}
    C^{z_2'}_{P^0}(x) & = \frac{1(s \in \mathcal{S}_3, z_2=z_2')}{P^0(S \in \mathcal{S}_3)P^0(Z_2=z_2'\mid Z_1 =z_1,S\in \mathcal{S}_3)} \frac{dP^0( z_1 \mid S \in \mathcal{S}_1 )}{dP^0(z_1 \mid S \in \mathcal{S}_3)} \nonumber\\
    & \hspace{14em}\cdot \{z_3 - E_{P^0}(Z_3 \mid Z_2 =z_2',Z_1 =z_1,S \in \mathcal{S}_3)\} \nonumber\\
    & \quad + \frac{1(s \in \mathcal{S}_1)}{P^0(S \in \mathcal{S}_1)} \{E_{P^0}(Z_3 \mid Z_2 = z_2',Z_1 =z_1,S \in \mathcal{S}_3) \nonumber \\
    &\quad\hspace{7em}- E_{P_0}\left[E_{P^0}(Z_3 \mid Z_2 = z_2',Z_1 ,S \in \mathcal{S}_3)\mid S\in \mathcal{S}_1\right]\}.\nonumber
\end{align}
Then by delta method, the canonical gradient of $\phi$ is, 
\begin{align}
    D_{P^0}(x) & = \frac{1}{\phi_{c}(P^0)}T_{P^0}(x) - \frac{\phi_{\textnormal{ITT}}(P^0)}{\phi^2_{c}(P^0)} C_{P^0}(x),\nonumber
\end{align}
where we use $T_{P^0}(x)$ to denote the canonical gradient of $\phi_{\textnormal{ITT}}$ given in  Appendix~\ref{example:ITT}.

\subsection{Off-policy evaluation}
\label{example:OPE}

\noindent \textbf{Parameter of interest:} Researchers are often interested in evaluating the impact of a new, previously unimplemented policy in a population \citep{dudik2014doubly}. 
This is known as off-policy evaluation, which aims to estimate the reward of a given policy using historical data that contains the outcomes under different, currently-implemented policies. Let $Z_1$ denote some baseline characteristic variable, $Z_2$ denote a discrete or continuous action variable, and $Z_3$ denote a real-valued outcome of interest. The evaluation policy $\Pi_e$ corresponds to a conditional distribution of $Z_2$ given $Z_1$. The target estimand is the average reward under the evaluation policy and writes as $\psi(Q) = E_{Q_1} \left[E_{\Pi_e}\left\{E_{Q_3}(Z_3\mid Z_2, Z_1) \mid Z_1 \right\} \right]$. The model $\mathcal{Q}$ consists of all distributions $Q$ such that $Q_2(\cdot\mid z_1)=\Pi_e(\cdot\mid z_1)$ $Q_1$-almost everywhere.  

Because $\psi(Q)$ can be written as a function of $Q_1$ and $Q_3$ only, we can take $\mathcal{I}=\{2\}$ in this example. Then, under Condition~\ref{cond:identifiability}, Theorem~\ref{thm:identifiability} shows that $\psi(Q^0)$ can be identified as
\begin{align*}
    \phi(P^0) &=  E_{P^0} \Big[ E_{\Pi_e}\big\{E_{P^0}(Z_3\mid Z_2, Z_1, S\in \mathcal{S}_3)\mid Z_1\big\} \mid  S \in \mathcal{S}_1\Big], 
\end{align*}
\cite{kallus2020optimal} considered a similar off-policy learning problem to the one we have presented here. Our setup differs from that of this earlier work in two respects. First, the sample size from each dataset is random in our setting as opposed to fixed in Kallus' setting. As we discuss in  Appendix~\ref{sec:miscellanea}, we do not believe that this distinction is particularly important, and in that appendix we outline how our results can be extended to handle the case where fixed sample sizes are observed from each data source. The second, and more important, way in which our results differ from those in \cite{kallus2020optimal} is that 
we do not require that the collection of data sources informing about the covariate distribution, $\mathcal{S}_1$, be the same as the collection of data sources informing about the reward distribution, $\mathcal{S}_3$. This makes it easy to consider the case of covariate shift for some data sources that inform on the reward distribution (by making them belong to $\mathcal{S}_3$ but not $\mathcal{S}_1$), and also to incorporate additional data sources that have information about the target distribution of covariates but not the reward distribution (by placing those data sources in $\mathcal{S}_1$ but not $\mathcal{S}_3$). Notably, our results also handle the case considered in \cite{kallus2020optimal} where $\mathcal{S}_1=\mathcal{S}_3=[k]$.

\vspace{.5em}\noindent \textbf{Comparison of identifiability conditions:} As noted above, \cite{kallus2020optimal} focused on the special case where $\mathcal{S}_1=\mathcal{S}_3=[k]$. In these cases, our Condition~\ref{cond:identifiability}\ref{suff_alignment} encodes the same requirement as in Section~2.1 of \cite{kallus2020optimal}, namely that, conditional on $S=s$, the distribution of the triplet $(Z_1,Z_2,Z_3)$ follows the product distribution $q_1(z_1)p(z_2\mid z_1,s)q_3(z_3\mid z_2,z_1)$. Our Condition~\ref{cond:identifiability}\ref{suff_overlap} corresponds to Assumption 1 in \cite{kallus2020optimal}.

\vspace{.5em}\noindent \textbf{Derivation of influence functions:} 
We let $p_2^0$ denote the conditional density of $Z_2$ given $(Z_1, S)$ under $P^0$ and let $\pi_e$ to denote the conditional density of the evaluation policy $\Pi_e$.
By Lemma~2 of \cite{kennedy2019nonparametric}, the canonical gradient of $\psi$ is given by
\begin{align*}
    D_{Q^0}(z) &= \frac{\pi_e(z_2\mid z_1)}{\pi^0(z_2\mid z_1)}\left\{z_3-E_{Q^0}(Z_3\mid Z_2 = z_2,Z_1 = z_1)\right\} \nonumber \\
      & \quad+ \sum_{z_2' \in \mathcal{Z}_2}E_{Q^0}(Z_3\mid Z_2 = z_2', Z_1=z_1) \pi_e(z_2'\mid z_1) -\psi(Q^0).
\end{align*}
Plugging this into \eqref{eq:phiCanGrad} shows that the canonical gradient of $\phi$ is given by
\begin{align}
      D_{P^0}(x) &= \frac{1(s  \in \mathcal{S}_3)}{P^0(S \in \mathcal{S}_3)} \frac{dP^0( z_1 \mid S \in \mathcal{S}_1 )}{dP^0(z_1 \mid S \in \mathcal{S}_3 )}\frac{\pi_e(z_2\mid z_1)}{\pi^0(z_2\mid z_1)}\left\{z_3-E_{P^0}(Z_3\mid Z_2 = z_2,Z_1 = z_1,S\in \mathcal{S}_3)\right\} \nonumber \\
      & \quad+ \frac{1(s \in \mathcal{S}_1)}{P^0(S \in \mathcal{S}_1)}\left\{\sum_{z_2' \in \mathcal{Z}_2}E_{P^0}(Z_3\mid Z_2 = z_2', Z_1=z_1, S\in \mathcal{S}_3) \pi_e(z_2'\mid z_1) -\phi(P^0)\right\},\nonumber
\end{align}
where $\pi^0(z_2\mid z_1)\equiv \sum_{s=1}^{k} p_2^0(z_2\mid z_1,s)P^0(S=s)$ and the sum over $z_2'$ above should be replaced by a Lebesgue integral.

\subsection{Z-estimation}
\label{example:Z-estimation appendix}
We now consider a more general example that can be applied for a wide variety of estimands. Specifically, we consider a $b$-dimensional Z-estimation problem where 
$\{m_\gamma : \gamma\in \mathbb{R}^b\}$ denotes a collection of $\mathcal{Z}\rightarrow\mathbb{R}^b$ functions (\citealp{hansen1982large}). We are interested in inferring the unknown parameter $\psi(Q^0)$, where, for all $Q$ in a specified model $\mathcal{Q}$, $\psi(Q)$ is defined implicitly as the solution in $\gamma$ to the estimating equation $M(Q)(\gamma)\equiv  E_{Q}\{m_{\gamma}(Z)\} = 0$. It is assumed that this solution is unique for each $Q$. As mentioned in Section~\ref{sec:notations and setup}, it is always possible to take $\mathcal{I}=\emptyset$. For certain classes of functions $\{m_\gamma : \gamma\in \mathbb{R}\}$, it will also be possible to take $\mathcal{I}$ to be a larger, non-empty set --- this is the case, for example, if the conditional distribution $Q_j^0$ is known for some $j$ or if $m_\gamma(z)$ only depends on $(z_1,\ldots,z_{d'})$, $d'<d$.

Under Condition~\ref{cond:identifiability}, $\psi(Q^0)$ can be identified with $\phi(P^0)$, which is the solution in $\gamma$ to 
\begin{align}
    0&= E_{P^0}\left[ \dots E_{P^0}\{E_{P^0}(m_\gamma(Z)\mid \bar{Z}_{d-1}, S\in \mathcal{S}_d) \mid  \bar{Z}_{d-2}, S \in \mathcal{S}_{d-1}\}\mid \ldots ,S \in \mathcal{S}_1\right]. \label{eq:Zestimation}
\end{align}
\cite{chakrabortty2016robust} treat the special case that arises in semi-supervised learning problems, namely the case where $d=2$, $k=2$, $\mathcal{S}_1=\{1,2\}$, and $\mathcal{S}_2=\{1\}$. Our example generalizes this previous work by allowing for additional data sources ($k>2$) and fusion sets ($d>2$). A great variety of problems can be studied using the generality of our framework. One specific example corresponds to a least-squares projection onto a working linear regression model where $(Z_1,\ldots,Z_{d-1})$ are features of interest, $Z_d$ is an outcome of interest, and the criterion function $m_\gamma(z)=\{z_d-(z_1,\ldots,z_{d-1})^\top \gamma\}^2$ is used. 

Under regularity conditions on $\mathcal{Q}$ and the functions $m_\gamma$, $\gamma\in\mathbb{R}^b$, an initial gradient $D_{Q^0}$ to plug into Theorem~\ref{thm:grad} can be found in Theorem~5.21 of \cite{van2000asymptotic}. Following results from Corollary~\ref{cor:canGrad}, the canonical gradient of $\phi$ takes the form $-V_{P^0}^{-1} F_{P^0}(x)$, where $V_{P^0}$ is the derivative matrix at $\phi(P^0)$ of the function of $\gamma$ defined pointwise to be equal to the right-hand side of \eqref{eq:Zestimation} and, for recursively defined $G_j^0(\bar{z}_{j},s) = E_{P^0}\{G_{j+1}^0(\bar{Z}_{j+1})\mid \bar{z}_{j},S\in \mathcal{S}_{j+1}\}$ with $G_{d}^0(\bar{z}_d,s) = m_{\phi(P^0)}(Z)$,

\begin{align}
       F_{P^0}(x) &  = \sum_{j=1}^{d} \mathbbm{1}(\bar{z}_{j-1} \in \bar{\mathcal{Z}}^{\dagger}_{j-1})\frac{\mathbbm{1}(s \in \mathcal{S}_j)}{P^0(S \in \mathcal{S}_j) } \left\{\prod_{m=1}^{j-1} \frac{dP^0(z_m \mid \bar{z}_{m-1},S \in \mathcal{S}_{m})}{dP^0(z_m \mid \bar{z}_{m-1},S \in \mathcal{S}_{j})}\right\} \nonumber\\
       & \hspace{23em} \cdot \left\{G_{j}^0(\bar{z}_{j},s) - G_{j-1}^0(\bar{z}_{j-1},s)\right\}. \label{eq:IF_z}
\end{align}
In fact, it can be verified via Theorem~\ref{thm:grad} that $F_{P^0}$ is the canonical gradient of $P^0\mapsto M\{\theta(P^0)\}(\phi(P^0))$ relative to $\mathcal{P}$, where we recall that $M(Q)(\gamma)\equiv  E_{Q}\{m_{\gamma}(Z)\}$. \\

\noindent \textbf{Derivation of influence functions:} 
Theorem~5.21 of \cite{van2000asymptotic} gives the canonical gradient of $\psi$ at $\uQ^0$ relative to a locally nonparametric $\mathcal{Q}$,  namely
\begin{align*}
    D_{\uQ^0}(z) &\equiv - V^{-1}_{\uQ^0}m_\gamma(z) = - V^{-1}_{\uQ^0}\sum_{j=1}^d\left\{\tilde{G}_j^0(\bar{z}_{j}) - \tilde{G}_{j-1}^0(\bar{z}_{j-1})\right\},
\end{align*}
where $V_{\uQ^0}$ is the derivative matrix at $\psi(\uQ^0)$ of the function of $\gamma$ defined equal to $M(\uQ^0)(\gamma)$ and we recursively define $\tilde{G}^0_j : (\bar{z}_{j}) \mapsto E_{\uQ^0}\{\tilde{G}^0_{j+1}(\bar{Z}_{j+1})\mid \bar{z}_{j}\}$ with $\tilde{G}^0_{d}: (\bar{z}_d) \mapsto m_{\psi(\uQ^0)}(Z)$ and $\tilde{G}_0^0=0$. 
It can be verified that $D_{\uQ_j^0}=-V^{-1}_{\uQ^0}\{\tilde{G}_j^{0} - \tilde{G}_{j-1}^{0}\}$. 

We take $\mathcal{I} = \emptyset$. Using the results in Corollary~\ref{cor:canGrad}, the canonical gradient of $\phi$ under a locally nonparametric model $\mathcal{P}$ is given by
\begin{align*}
    D_{P^0}(x) & = \sum_{j=1}^d \mathbbm{1}(\bar{z}_{j-1} \in \bar{\mathcal{Z}}^{\dagger}_{j-1})\frac{\mathbbm{1}(s \in \mathcal{S}_{j})}{P^0(S \in \mathcal{S}_{j})} \Pi_{\uQ^0}\{\lambda_{j-1}D_{\uQ_{j}^0} \mid \mathcal{T}(\uQ^0,\mathcal{Q})\}.
\end{align*}
Since the model $\mathcal{Q}$ is locally nonparametric, for $j \in \{1,\ldots,d\}$, we have $\mathcal{T}(\uQ^0_j,\mathcal{Q}) = L_0^2(\uQ_j^0)$. As a result,  $\Pi_{\uQ^0}\{\lambda_{j-1}D_{\uQ^0_{j}} \mid \mathcal{T}(\uQ^0,\mathcal{Q})\} = - \lambda_{j-1}(\bar{z}_{j-1})V^{-1}_{\uQ^0} \{ \tilde{G}^{0}_j (\bar{z}_j) - \tilde{G}^{0}_{j-1}(\bar{z}_{j-1})\}$. Substituting this expression back into $D_{P^0}$, we obtain, 
\begin{align*}
        D_{P^0}(x) & = - \sum_{j=1}^d \mathbbm{1}(\bar{z}_{j-1} \in \bar{\mathcal{Z}}^{\dagger}_{j-1})\frac{\mathbbm{1}(s \in \mathcal{S}_{j})}{P^0(S \in \mathcal{S}_{j})} \frac{d\uQ^0(\bar{z}_{j-1})}{dP^0(\bar{z}_{j-1} \mid S \in \mathcal{S}_j)}  V^{-1}_{\uQ^0} \left[\tilde{G}^{0}_j (\bar{z}_j) - \tilde{G}^{0}_{j-1}(\bar{z}_{j-1})\right]. 
\end{align*}
Replacing $\uQ^0(\bar{z}_{j-1})$, $V^{-1}_{\uQ^0}$, and $\tilde{G}^{0}_j$ with the corresponding features of observed data distribution $P^0$ yields the above canonical gradient.

\subsection{Quantile treatment effect}

Average treatment effects are commonly used to quantify the impact of a treatment on an outcome \citep{hernan2020causal}. Quantile treatment effects, which represent the difference of the $\tau$-quantile of the outcome on treatment versus control, provide a complementary approach \citep{firpo2007efficient}. When $\tau$ is near zero or one, quantile treatment effects make it possible to pick up effects that occur in the tails of the outcome distribution. When $\tau$ is far from zero and one, they instead represent a robust treatment effect estimand that is insensitive to outlying values of the outcome.

We now describe how quantile treatment effects fit within our framework. Let $Z_1$ be a real-valued baseline variable, $Z_2$ be a binary treatment variable that is assigned at random (as in a randomized trial), and $Z_3$ be an outcome. For a fixed $\tau\in (0,1)$, the target estimand is $\psi(Q^0)\equiv  u^0_1 -u^0_0$, where $u^{Q}_{z_2}\equiv \inf\{u : Q(Z_3 \leq u \mid Z_2=z_2) \geq \tau\}$ for $z_2\in \{0,1\}$ and we let $u_{z_2}^0 \equiv u_{z_2}^{Q^0}$. The model $\mathcal{Q}$ consists of all distributions $Q$ with support on $\mathbb{R}\times \{0,1\}\times\mathbb{R}$ that are such that $Z_2$ is independent of $Z_1$, the marginal distribution of $Z_2$ takes some known value, and $u\mapsto Q(Z_3\le u\mid Z_2=z_2)$ is everywhere differentiable for each $z_2\in\{0,1\}$. Because $\psi(Q)$ can be written as a function of $Q_1$ and $Q_3$, we see that we can take $\mathcal{I} = \{2\}$ in this example. Under Condition~\ref{cond:identifiability}, $\psi(Q^0)$ can be identified as,
\begin{align*}
    \phi(P^0) = \sum_{a=0}^1 (2a-1)\inf\{u:P^0\{P^0(Z_3 \leq u \mid Z_2=a, Z_1,S\in\mathcal{S}_3 ) \mid S \in\mathcal{S}_1\} \geq \tau \}.
\end{align*}
To our knowledge, quantile treatment effects have not previously been studied in a data fusion setting.

Following results from Corollary~\ref{cor:canGrad},  we can show the canonical gradient under a locally nonparametric model is $D_{P^0}(x) = D^1_{P^0}(x)-D^0_{P^0}(x)$ where, letting $\uq^0(\,\cdot\mid z_2)$ denote the conditional density of $Z_3$ given that $Z_2=z_2$ and, for $z_2'\in \{0,1\}$, $\rho^{z_2'}_{\tau}(z_3)\equiv  \{\tau-\mathbbm{1}(z_3 \leq u^{0}_{z_2'})\}/\int p^0(Z_3 = u^0_{z_2'} \mid Z_2=z_2',z_1, S\in \mathcal{S}_3)  p^0(z_1 \mid S \in \mathcal{S}_3) dz_1$, $D^{z_2'}_{P^0}(x)$ takes the form
\begin{align}
    D^{z_2'}_{P^0}(x) & =  \mathbbm{1}(z_1 \in \mathcal{Z}^{\dagger}_1) \frac{\mathbbm{1}(s \in \mathcal{S}_3)}{P^0(S \in \mathcal{S}_3)}\frac{\mathbbm{1}(z_2=z_2')}{P^0(Z_2=z_2'\mid z_1,S\in \mathcal{S}_3)}\frac{dP^0(z_1 \mid S\in \mathcal{S}_1)}{dP^0(z_1 \mid S\in \mathcal{S}_3 )} \nonumber \\
    & \hspace{18em}\cdot \left[\rho^{z_2'}_{\tau}(z_3) - E\left\{\rho^{z_2'}_{\tau}(Z_3)\mid Z_2=z_2',z_1,S \in \mathcal{S}_3\right\}\right] \nonumber\\
    &\quad+ \frac{\mathbbm{1}(s \in \mathcal{S}_1)}{P^0(S \in \mathcal{S}_1)} E\left\{\rho^{z_2'}_{\tau}(Z_3)\mid Z_2=z_2',z_1,S\in \mathcal{S}_3\right\}. \label{eq:IF_qte}
\end{align}

In below, the above form of the canonical gradient is derived directly via Corollary~\ref{cor:canGrad}. An alternative approach to deriving this result involves noting that the quantile treatment effect can be written as the difference of two implicitly defined functionals. Consequently, the results for Z-estimation in  Appendix~\ref{example:Z-estimation appendix} could be used to derive the canonical gradients of these two functionals, and then the delta method would provide the canonical gradient for the quantile treatment effect. More concretely, this approach involves noting that
$(u^0_1,u^0_0)$ is the solution in $\gamma\in\mathbb{R}^2$ to the estimating equation $M(Q^0)(\gamma) = E_{Q^0}[m_\gamma(Z)]$, where $m_\gamma(z)=(\mathbbm{1}(Z_2=z_2)\{ \mathbbm{1}(Z_3 \leq \gamma_{z_2}) - \tau\})_{z_2=0}^1$.\\

\noindent \textbf{Derivation of influence functions:}
\label{sec:derive_qte}
\cite{firpo2007efficient} gave the canonical gradient of $\psi$ under $\mathcal{Q}$, which is $D_{\uQ^0} \equiv D^{1}_{\uQ^0} -D^{0}_{\uQ^0}$, where 
\begin{align*}
    D^{z'_2}_{\uQ^0}(z) & \equiv \frac{\mathbbm{1}(z_2 = z'_2)}{\uQ^0(Z_2=z'_2 \mid Z_1 = z_1) }\left[ \rho_{\tau}^{z'_2}(z_3)  - E_{\uQ^0} \{\rho_{\tau}^{z'_2}(Z_3) \mid Z_2 = z'_2, Z_1 =z_1\}\right] \\
    & \quad  + E_{\uQ^0} \{\rho_{\tau}^{z'_2}(Z_3) \mid Z_2 = z'_2, Z_1 =z_1\}.
\end{align*}
We take $\mathcal{I} = \{2\}$. Since the only restriction on the model $\mathcal{Q}$ is that $Z_2$ is independent of $Z_1$,  $\Pi_{\uQ^0} \{\lambda_{j-1}D^{z'_2}_{\uQ^0,j} \mid \mathcal{T}(\uQ^0, \mathcal{Q})\} = \lambda_{j-1}D^{z'_2}_{\uQ^0,j}$ when $j\in \{1,3\}$. Following the results in Corollary~\ref{cor:canGrad}, the canonical gradient of $\phi$ relative to a locally nonparametric model $\mathcal{P}$ is
\begin{align*}
    D^{z'_2}_{P^0}(x) & = \mathbbm{1}(\bar{z}_{2} \in \bar{\mathcal{Z}}^{\dagger}_{2})\frac{\mathbbm{1}(s \in \mathcal{S}_{3})}{P^0(S \in \mathcal{S}_{3})} \Pi_{\uQ^0}\{\lambda_{2}D_{\uQ^0,3} \mid \mathcal{T}(\uQ^0,\mathcal{Q})\}  \\
    &\quad+ \frac{\mathbbm{1}(s \in \mathcal{S}_{1})}{P^0(S \in \mathcal{S}_{1})} \Pi_{\uQ^0}\{D_{\uQ^0,1} \mid \mathcal{T}(\uQ^0,\mathcal{Q})\} \\
    & =  \mathbbm{1}(\bar{z}_{2} \in \bar{\mathcal{Z}}^{\dagger}_{2})\frac{\mathbbm{1}(s \in \mathcal{S}_{3})}{P^0(S \in \mathcal{S}_{3})} \frac{d\uQ^0(\bar{z}_2)}{dP^0(\bar{z}_2 \mid S \in \mathcal{S}_3)}D_{\uQ^0,3}(\bar{z}_{3})  + \frac{\mathbbm{1}(s \in \mathcal{S}_{1})}{P^0(S \in \mathcal{S}_{1})} D_{\uQ^0,1}(z_{1}) \\
    & =  \mathbbm{1}(\bar{z}_{2} \in \bar{\mathcal{Z}}^{\dagger}_{2})\frac{\mathbbm{1}(s \in \mathcal{S}_{3})}{P^0(S \in \mathcal{S}_{3})} \frac{P^0(Z_2=z'_2 \mid Z_1=z_1, S \in \mathcal{S}_2)dP^0(z_1 \mid S \in \mathcal{S}_1)}{P^0(Z_2 = z'_2 \mid Z_1 = z_1, S \in \mathcal{S}_3) dP^0(z_1 \mid S \in \mathcal{S}_3)}D_{\uQ^0,3}(\bar{z}_{3})  \\
    & \quad  + \frac{\mathbbm{1}(s \in \mathcal{S}_{1})}{P^0(S \in \mathcal{S}_{1})}D_{\uQ^0,1}(z_{1}). 
\end{align*}
Substituting the expressions for $D_{\uQ^0,3}(z_{3})$ and $D_{\uQ^0,1}(z_{1})$ into the above, we obtain \eqref{eq:IF_qte}.

\section{Review of semiparametric theory}
\label{sec:efficiency theory_appendix}

We review some important aspects of nonparametric and semiparametric theory in this subsection. Further details can be found in \cite{bickel1993efficient}. We begin by discussing the case that $\phi$ is univariate ($b=1$), and then we discuss the case where $\phi$ is multivariate ($b\ge 2$). An estimator $\hat{\phi}$ of $\phi(P)$ is called asymptotically linear with influence function $D_P$ if it can be written as $\hat{\phi} - \phi(P) = n^{-1}\sum_{i=1}^n  D_P(X_i) + o_p(n^{-1/2})$, where $E_P\{D_P(X_i)\}=0$ and $\sigma_P^2\equiv E_P\{D_P(X_i)^2\}<\infty$. 
One reason such estimators are attractive is that they are consistent and asymptotically normal, in the sense that $\sqrt{n} \{ \hat{\phi} - \phi(P)\} \xrightarrow{d} N(0,\sigma_P^2)$ under sampling $n$ independent draws from $P$. This facilitates the construction of confidence intervals and hypothesis tests. 
It is also often desirable for $\hat{\phi}$ to not depend on the particular data-generating distribution too heavily, in the sense that $\sqrt{n} \{\hat{\phi} - \phi(P^{(n^{-1/2})})\}$ converges to the same distribution under sampling from any sequence of distributions $(P^{(n^{-1/2})})_{n=1}^\infty$ that converges to $P$ in an appropriate sense. In particular, $(P^{(n^{-1/2})})_{n=1}^\infty$ should arise from a submodel in the collection $\mathscr{P}(P,\mathcal{P})$ of submodels $\{P^{(\epsilon)} : \epsilon\in [0,\delta)\}$ of $\mathcal{P}$ with $P^{(0)}=P$ and with score $h$ at $\epsilon=0$, where the score is defined in a quadratic mean differentiability sense \citep{van2000asymptotic}. 
If $\hat{\phi}$ is regular and asymptotically linear at $P$ \citep{bickel1993efficient}, then $\phi$ is pathwise differentiable and the influence function $D_P$ is a gradient of $\phi$, in the sense that, for all submodels $\{P^{(\epsilon)} : \epsilon\in [0,\delta)\}\in\mathscr{P}(P,\mathcal{P})$, $\frac{\partial}{\partial \epsilon} \phi(P_{\epsilon}) \mid_{\epsilon=0} = E_P\{D_P(X)h(X)\}$. The representation $E_P\{D_P(X)h(X)\}$ can be viewed as an inner product between $D_P$ and $h$ in the Hilbert space $L_0^2(P)$ of $P$-mean-zero functions, finite variance functions. The tangent set $\mathcal{T}(P,\mathcal{P})$ of $\mathcal{P}$ at $P$ is defined as the set of all scores of submodels in $\mathscr{P}(P,\mathcal{P})$. Since scores are mean-zero, finite variance functions, $\mathcal{T}(P,\mathcal{P})\subseteq L_0^2(P)$. The canonical gradient $D_P^*$ corresponds to the $L_0^2(P)$-projection of any gradient $D_P$ onto the closure of the linear span of scores in $\mathcal{T}(P,\mathcal{P})$. Since $L_0^2(P)$ projections reduce variance and the influence function of any regular and asymptotically linear estimator is a gradient, any regular and asymptotically linear estimator that has the canonical gradient as its influence function achieves the minimal possible asymptotic variance among all such estimators. Thus, $D_P^*$ is also referred to as the efficient influence function.

\section{Construction of estimators}
\label{app:contruct_estimators}
 One way to construct a regular asymptotically linear estimator with influence function $D_{P^0}$ is through one-step estimation \citep{ibragimov1981statistical,bickel1982adaptive}. 
Given an estimate $\widehat{P}$ of $P^0$, 
the one-step estimator is given by $\hat{\phi} \equiv \phi(\widehat{P}) + \sum_{i=1}^n D_{\widehat{P}}(X_i)/n$. 
This estimator will be asymptotically linear with influence function $D_{P^0}$ if
\begin{enumerate}
    \item the remainder term $R(\widehat{P},P^0)\equiv \phi(\widehat{P}) - \phi(P^0) + E_{P^0}\{D_{\widehat{P}}(X)\}$ is $o_p(n^{-1/2})$, and
    \item the empirical mean of $D_{\widehat{P}}(X)-D_{P^0}(X)$ is within $o_p(n^{-1/2})$ of the mean of this term when $X\sim P^0$.
\end{enumerate}
When the above two conditions hold and $D_{P^0}$ is the canonical gradient of $\phi$ relative to the statistical model $\mathcal{P}$, the one-step estimator $\hat{\phi}$ will be efficient in the sense that it will achieve the minimal possible asymptotic variance among all regular estimators \citep{bickel1993efficient}. 
Typically, the remainder $R(\widehat{P},P^0)$ appearing in the first of these requirements will involve second-order products of errors. To make meeting this first requirement as likely as possible, we recommend using nonparametric approaches such as kernel density estimation for estimating density ratios and data-adaptive methods for estimating other nuisance parameters. For example, in our data illustration, we use the ensemble learning approach known as super learning to estimate the needed regression functions and conditional probabilities. 
The second of these requirements will hold under an appropriate empirical process and consistency condition, namely that $D_{\widehat{P}}$ belongs to a fixed $P^0$-Donsker class with probability tending to one and the $L^2(P^0)$-norm of $x\mapsto D_{\widehat{P}}(x)-D_{P^0}(x)$ converges to zero in probability \citep[Lemma 19.24 of][]{van2000asymptotic}. 

This empirical process condition can be avoided if cross-fitting is used when developing the initial estimator $\widehat{P}$ of $P^0$ \citep[e.g.,][]{zheng2011cross,chernozhukov2018double}. When $D_{\widehat{P}}$ is a gradient of $\phi$ at $\widehat{P}$ relative to the data fusion model $\mathcal{P}$, but not necessarily with respect to a locally nonparametric model, it will be important to ensure that the initial estimate $\widehat{P}$ of $P^0$ belongs to the $\mathcal{P}$, so that, for all $j\in\mathcal{J}$, there exists some $Q\in\mathcal{Q}$ such that $\widehat{P}_j(\,\cdot\mid \bar{z}_{j-1},s)=Q_j(\,\cdot\mid \bar{z}_{j-1})$ for all $s\in\mathcal{S}_j$; otherwise, it will not generally be plausible that $R(\widehat{P},P^0)\equiv \phi(\widehat{P}) - \phi(P^0) + E_{P^0}\{D_{\widehat{P}}(X)\}$ is $o_p(n^{-1/2})$. 
Alternative approaches for constructing asymptotically linear estimators include targeted minimum loss-based estimation \citep{van2006targeted} and estimating equations \citep{van2003unified,tsiatis2006semiparametric}.

\section{Simulation setup and results}\label{app:sim}
\label{app:sim_long}
\begin{table}[htb]
\centering
\caption{Specification of each simulated data source. The distribution of the target population is specified in Table~\ref{tab:Long_setup_distributions}. For others, $U_1 \mid S\notin \mathcal{S}_1 \sim \textnormal{Normal}(3,2)$, $U_2 \mid A_1,U_1, S\notin \mathcal{S}_3 \sim \textnormal{Normal}(\mathbbm{1}(A_1=0)(3+1.8U_1) +\mathbbm{1}(A_1=1)(5+0.8U_1), 5 )$, $U_3 \mid (\bar{H}_2, A_2,S\notin \mathcal{S}_5) \sim \textnormal{Normal}(\mathbbm{1}(A_1=A_2=0)(20+0.44U_1+0.07U_2) +\mathbbm{1}(A_2=1,A_1=0)(-10+0.71U_1+0.12U_2)+\mathbbm{1}(A_2=0,A_1=1)(10+0.44U_1+0.07U_2)+\mathbbm{1}(A_2=A_1=1)(10+0.41U_1+0.24U_2), \mathbbm{1}(A_2=A_1=0) 0.74 + \mathbbm{1}(A_2=1,A_1=0) 2.34 +\mathbbm{1}(A_2=0,A_1=1) 0.74 +\mathbbm{1}(A_2=A_1=1) 3.56)$.}
\begin{tabular}{lrrr}
\toprule
k &  Observed Variables & Sample Size & Distribution \\
\midrule
1 & $(U_1)$ & $n=2000$ & $\mathcal{S}_1$\\
2 & $(U_1)$ & $n=400$ & \\
3 & $(U_1,A_1,U_2)$ & $n=2000$ & $\mathcal{S}_1$, $\mathcal{S}_2$, $\mathcal{S}_3$\\
4 & $(U_1,A_1,U_2)$ & $n=400$ & $\mathcal{S}_2$\\
5 & $(U_1,A_1,U_2,A_2,U_3)$ & $n=2000$ & $\mathcal{S}_2$, $\mathcal{S}_3$, $\mathcal{S}_4$, $\mathcal{S}_5$ \\
6 & $(U_1,A_1,U_2,A_2,U_3,A_3,Y)$ & $n=4000$ &  $\mathcal{S}_2$, $\mathcal{S}_4$, $\mathcal{S}_5$, $\mathcal{S}_6$, $\mathcal{S}_7$\\
7 & $(U_1,A_1,U_2,A_2,U_3)$ & $n=2000$ &  $\mathcal{S}_2$, $\mathcal{S}_3$, $\mathcal{S}_4$\\
8 & $(U_1,A_1,U_2,A_2,U_3,A_3,Y)$ & $n=4000$ &  $\mathcal{S}_2$, $\mathcal{S}_4$, $\mathcal{S}_5$, $\mathcal{S}_6$, $\mathcal{S}_7$\\
9 & $(U_1,A_1,U_2,A_2,U_3,A_3,Y)$ & $n=2000$ & $\mathcal{S}_1$, $\mathcal{S}_2$, $\mathcal{S}_3$, $\mathcal{S}_4$, $\mathcal{S}_5$, $\mathcal{S}_6$, $\mathcal{S}_7$\\
\bottomrule
\end{tabular}
\label{tab:Long_setup}
\end{table}

\begin{table}[htb]
\centering
\caption{Specification of the distribution of $\bar{U}_3 \mid \bar{A}_2$ of the target population. }
\begin{tabular}{lrr}
\toprule
$(A_1,A_2)$ &  $\mu$ & $\Sigma$  \\
\midrule
(0,0) & (0,1,10) & $\begin{pmatrix} 1& 0.8& 0.5\\0.8&2&0.5\\0.5&0.5&1\end{pmatrix}$\\
(1,0) & (0,1.5,20) & $\begin{pmatrix} 1& 0.8& 0.5\\0.8&2&0.5\\0.5&0.5&1\end{pmatrix}$\\
(0,1) & (0,1,20) & $\begin{pmatrix} 1& 0.8& 0.5\\0.8&2&0.5\\0.5&0.5&3\end{pmatrix}$\\
(1,1) & (0,1.5,40) & $\begin{pmatrix} 1& 0.8& 0.6\\0.8&2&0.8\\0.6&0.8&4\end{pmatrix}$\\
\bottomrule
\end{tabular}
\label{tab:Long_setup_distributions}
\end{table}

The transformation $\kappa$ is taken to be equal to 
\begin{align*}
    \kappa(\bar{H}_3,A_3) &= (1,A_1,A_2,A_3,A_1A_2,A_1A_3,A_2A_3,A_1A_2A_3,\\
    & \quad U_1,U_2,U_3,A_1U_1,A_2U_1,A_3U_1,A_1U_2,A_2U_2,A_3U_2,A_1U_3,A_2U_3,A_3U_3,\\
    & \quad A_1A_2U_1,A_1A_3U_1,A_2A_3U_1,A_1A_2U_2,A_1A_3U_2,A_2A_3U_2,A_1A_2U_3,A_1A_3U_3,A_2A_3U_3,\\
    & \quad A_1A_2A_3U_1,A_1A_2A_3U_2,A_1A_2A_3U_3)
\end{align*}
with $c = 32$, and the values of $\beta$ are such that
\begin{align*}
    E[Y\mid \bar{H}_3, A_3] &= \mathbbm{1}(A_1=A_2=A_3=0)[5+ \mu_{A_1,A_2}^{\top} \Sigma_{A_1,A_2}^{-1} \{(U_1,U_2,U_3)^{\top} - \mu_{A_1,A_2}\} ]\\
    & \quad +\mathbbm{1}(A_1+A_2=1,A_3=0)[8+ \mu_{A_1,A_2}^{\top} \Sigma_{A_1,A_2}^{-1} \{(U_1,U_2,U_3)^{\top} - \mu_{A_1,A_2}\}]\\
    & \quad +\mathbbm{1}(A_1=A_2=0,A_3=1)[9+ \mu_{A_1,A_2}^{\top} \Sigma_{A_1,A_2}^{-1} \{(U_1,U_2,U_3)^{\top} - \mu_{A_1,A_2}\}]\\
    & \quad +\mathbbm{1}(A_1=1,A_2+A_3=1)[10+ \mu_{A_1,A_2}^{\top} \Sigma_{A_1,A_2}^{-1} \{(U_1,U_2,U_3)^{\top} - \mu_{A_1,A_2}\}]\\
    & \quad +\mathbbm{1}(A_1=0,A_2=A_3=1)[12+ \mu_{A_1,A_2}^{\top} \Sigma_{A_1,A_2}^{-1} \{(U_1,U_2,U_3)^{\top} - \mu_{A_1,A_2}\}]\\
    & \quad +\mathbbm{1}(A_1=A_2=A_3=1)[15+ \mu_{A_1,A_2}^{\top} \Sigma_{A_1,A_2}^{-1} \{(U_1,U_2,U_3)^{\top} - \mu_{A_1,A_2}\}].
\end{align*}

When constructing the initial plug-in $\hat{P}$ for the one-step estimator, we want to make sure such $\hat{P}$ resides in the model. For the semiparametric model that assumed a symmetric conditional outcome distribution, we set the density estimate $\hat{p}_{2T-1}(u_T \mid \bar{h}_{T-1},a_{T-1}, S\in\mathcal{S}_{2T-1})$ to be $\{\hat{f}(u_T) + \hat{f}(2g(\bar{h}_{T-1},a_{T-1}) - u_T)\}/2$, where $\hat{f}(u_T \mid \bar{h}_{T-1},a_{T-1})$ is the kernel density estimator for the conditional density of $U_T=u_T$ given $(\bar{H}_{T-1}, A_{T-1}) = (\bar{h}_{T-1},a_{T-1})$. Similarly, we set the estimate for the derivative of density $\hat{p}'_{2T-1}(u_T \mid \bar{h}_{T-1},a_{T-1}, S\in\mathcal{S}_{2T-1})$ to be $\{\hat{f}'(u_T) - \hat{f}'(2g(\bar{h}_{T-1},a_{T-1}) - u_T)\}/2$, where $\hat{f}'(u_T \mid \bar{h}_{T-1},a_{T-1})$ is the kernel density estimator for the derivative of the conditional density of $U_T=u_T$ given $(\bar{H}_{T-1}, A_{T-1}) = (\bar{h}_{T-1},a_{T-1})$. Kernel density estimation and the corresponding derivative density estimation were performed using a normal scale bandwidth \citep{duong2007ks} from the \texttt{ks} R package. To avoid over-fitting, we obtained the estimate of $g(\bar{h}_{T-1},a_{T-1})$ via a 2-fold cross-fitting and took the average as $\hat{g}(\bar{h}_{T-1},a_{T-1})$. We estimated $\tilde{I}$ and the conditional expectation in equation~\ref{eq:EIF_long_semi} using SuperLearner \citep{van2007super} with a library containing generalized linear model, general additive model and elastic.

\section{Estimand of interest, canonical gradients and additional results for the data illustration} \label{sec: real_data_appendix}
We studied the associations between immune responses and baseline covariates in Kullback-Leibler projections onto univariate working logistic regression models. Specifically, for each baseline covariate of interest $Z_{1j}$, we are interested in estimating 
\begin{align*}
    \beta^0 : = \argmin \beta E_{Q^0}[l(Z;\beta)],
\end{align*}
where $l(Z;\beta)$ denotes the weighted quasi-log-likelihood loss function that accounts for two-phase sampling weights $w(Z)$ and is given by
\begin{align*}
        l(z;\beta)& := w(z)\left\{- z_3 \log \mu(z_{1j};\beta) - (1- z_3) \log \left(1 - \mu(z_{1j};\beta)\right)\right\},
\end{align*}
where $\mu(z_{1j};\beta) = 1/\exp ( -\beta_0 - \beta_1 z_{1j})$ denotes the mean function. To derive the canonical gradients of $\beta^0$, we adopt the Z-estimation framework introduced in  Appendix~\ref{example:Z-estimation appendix}. Following previous notations, our problem is equivalent to solving 
\begin{align*}
     0 & = E_{P^0}[E_{P^0}[m_{\beta^0}(Z)\mid Z_1, S \in \mathcal{S}_3] \mid S \in \mathcal{S}_1],
\end{align*}
where $m_{\beta}(z) = \frac{\partial l(z;\beta)}{\partial \beta}\mid_{\beta = \beta^0}$. The canonical gradient of $\beta^0$ is thus given by,
\begin{align*}
    D_{P^0}(z,s)&  = - \left( \frac{\partial  E_{P^0}[E_{P^0}[m_{\beta}(Z)\mid Z_1, S \in \mathcal{S}_3] \mid S \in \mathcal{S}_1]}{\partial \beta} \mid_{\beta = \beta^0}\right)^{-1} F_{P^0}(z,s;\beta^0),
    \intertext{where}
    F_{P^0}(z,s;\beta)&  = \frac{\mathbbm{1}(s \in \mathcal{S}_3)}{P^0(S \in \mathcal{S}_3)} \frac{dP^0(z_1 \mid S\in \mathcal{S}_1)}{dP^0(z_1 \mid S\in \mathcal{S}_3)} \left(m_{\beta} - E_{P^0}[m_{\beta}\mid z_1, S \in \mathcal{S}_3]\right) \\
    & \quad +\frac{\mathbbm{1}(s \in \mathcal{S}_1)}{P^0(S \in \mathcal{S}_1)} \left( E_{P^0}[m_{\beta}\mid z_1, S \in \mathcal{S}_3] - E_{P^0}[E_{P^0}[m_{\beta}\mid z_1, S \in \mathcal{S}_3]\mid S \in \mathcal{S}_1]\right).
\end{align*}
We examined the overlaps in distributions of baseline covariates across the two trials. Table~\ref{tab:HVTN} presents baseline characteristics of participants in STEP and Phambili. Figure~\ref{fig:density} presents estimates of the distributions of baseline covariates of participants who had their immune responses measured. Table~\ref{tab:beta_rest} presents the estimated coefficient in the univariate working logistic models considered.

\begin{table}[htb]
\centering
\caption{Baseline characteristics of participants in STEP and Phambili.}
\begin{tabular}{lrr}
\toprule
 &  \textbf{STEP} & \textbf{Phambili} \\
 &  (N=2979) & (N=801) \\
\midrule
Age (years) & 18-45 & 18-35 \\
Sex \\
\hspace{0.3em} Male & 1844 (61.9\%) & 441 (55.1\%)\\
\hspace{0.3em} Female & 1135 (38.1\%) & 360 (44.9\%)\\
Race \\
\hspace{0.3em} Black & 889 (29.8\%)  & 793 (99.0\%) \\
\hspace{0.3em} Other & 2090 (70.2\%)  & 8 (1.0\%) \\
Adenovirus serotype-5 positivity & 2021 (67.8\%) & 647 (80.8\%) \\
Circumcision (men only) & 1003 (54.3\%) & 129 (29.3\%)\\
\bottomrule
\end{tabular}
\label{tab:HVTN}
\end{table}

\begin{figure}[htb]
\centering
\includegraphics[scale = 0.6]{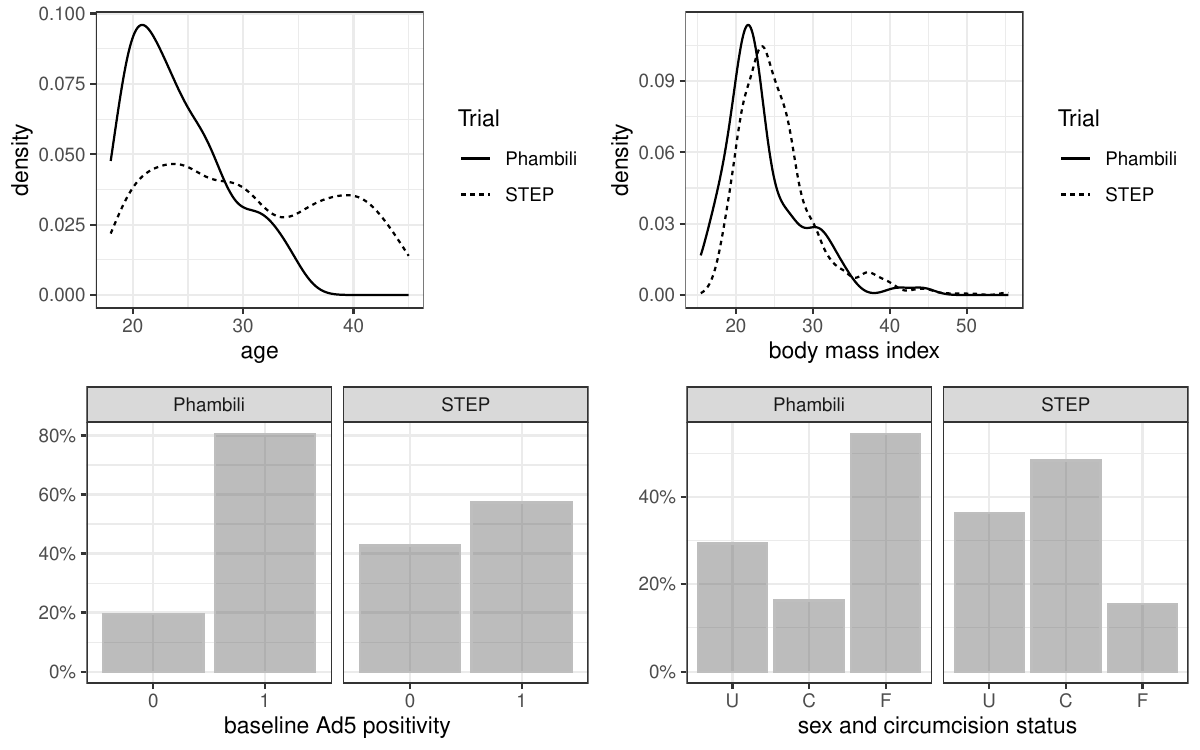}
\caption{Kernel density estimates and histograms describing the distributions of baseline covariates of participants who had their immune responses measured by ELISpot.}
\label{fig:density}
\end{figure}

\begin{table}[htb]
\centering
\caption{Estimated coefficient in univariate working logistic model between baseline covariates and, Nef and Pol using the STEP and Phambili study data. Estimation results are presented as estimates  (standard errors). }
\begin{tabular}{lrrrr}
\toprule
   & \multicolumn{2}{c}{Augmenting STEP} 
& \multicolumn{2}{c}{Augmenting Phambili} \\
\cmidrule(l){2-3} \cmidrule(l){4-5}
 &  STEP only & Both  &  Phambili only & Both \\
 &  (N=979) & (N=1072)  &  (N=93) & (N=1072) \\
\midrule
\textbf{Nef} \\
\hspace{1em}Age  & 0.03 (0.01) & 0.03 (0.01) & 0.07 (0.06) & 0.16 (0.04)\\
\hspace{1em}Circumcised male& 1.32 (0.13) & 1.33 (0.12) & 0.70 (0.55) & 0.42 (0.42)\\
\hspace{1em}Uncircumcised male& -0.40 (0.19) & -0.40 (0.18) & 0.01 (0.68) & 0.07 (0.53)\\
\hspace{1em}Female & 0.32 (0.34) & 0.37 (0.28) & 0.25 (0.63) & 0.43 (0.48)\\
\hspace{1em}BMI & -0.01 (0.02) & -0.01 (0.02) & -0.01 (0.04) & -0.02 (0.03)\\
\hspace{1em}Ad-5 positivity  & -1.07 (0.20) & -1.09 (0.20) & -0.16 (0.58) & 0.33 (0.46)\\
\textbf{Pol} \\
\hspace{1em}Age & 0.02 (0.01) & 0.01 (0.01) & 0.03 (0.05) & 0.12 (0.04)\\
\hspace{1em}Circumcised male& 1.35 (0.13) & 1.34 (0.11) & 0.41 (0.53) & 0.48 (0.33)\\
\hspace{1em}Uncircumcised male& -0.15 (0.20) & -0.12 (0.18) & 0.13 (0.66) & -0.36 (0.47)\\
\hspace{1em}Female& -0.97 (0.27) & -1.06 (0.23) & 0.75 (0.62) & 0.93 (0.43)\\
\hspace{1em}BMI & -0.03 (0.02) & -0.03 (0.02) & 0.02 (0.05) & 0.06 (0.03)\\
\hspace{1em}Ad-5 positivity & -2.04 (0.25) & -2.04 (0.24) & -0.94 (0.68) & -0.78 (0.55)\\
\bottomrule
\end{tabular}
\label{tab:beta_rest}
\end{table}

\section{Implementation and Possible Extensions}
\label{sec:miscellanea}

As mentioned in Section~\ref{sec:efficiency theory}, the initial estimate $\hat{P}$ of $P^0$ used to construct an efficient one-step estimator generally must reside in the model $\mathcal{P}$ for guarantees on such estimators to hold. In practice, $Q^0_j$, $j \in \mathcal{J}$, can be estimated by pooling data from sources in $\mathcal{S}_j$ and setting $\hat{P}_j(\cdot \mid \bar{z}_{j-1},s)$ equal to that estimate of $Q^0_j$ for those data sources $s$. Sometimes it is only necessary to estimate certain components of $P^0$, namely the ones needed to evaluate $\phi$ and the gradient. For example, constructing a one-step estimator of the longitudinal treatment effect in Section~\ref{sec:estimation} does not require estimating the conditional distribution of the outcome $U_T$ given the past; instead, only the conditional mean of $U_T$ given the past must be estimated.

We conjecture that Condition~\ref{cond:identifiability} can often be relaxed. In particular, rather than needing to require exact equality between $Q_j^0(\,\cdot\mid \bar{z}_{j-1})$ and conditional distributions $P_j^0(\,\cdot\mid \bar{z}_{j-1},s)$ under aligning data sources $s\in \mathcal{S}_j$, we believe that it is typically only necessary to have certain features of $P_j^0(\,\cdot\mid \bar{z}_{j-1},s)$ -- such as conditional expectations like $E_{P_0}(Z_j\mid \bar{Z}_{j-1}=\cdot,S=s)$ -- align with $E(Z_j\mid \bar{Z}_{j-1}=\cdot)$. The particular features that need to align in any given problem should be those that are needed to be able to evaluate the functional $\phi$ and its canonical gradient relative to $\mathcal{P}$.
For example, in the longitudinal treatment effect problem that we have studied, we conjecture that data sources $s \in \mathcal{S}_{2T-1}$ would only need to be such that $E_{P^0}(U_{T} \mid \bar{H}_{T-1}=\bar{h}_{T-1}, A_{T-1}=a, S \in \mathcal{S}_{2T-1}) = E_{Q^0}(U_{T} \mid \bar{H}_{T-1}=\bar{h}_{T-1}, A_{T-1}=a)$, rather than the stronger condition that $U_{T} \mid \bar{H}_{T-1}=\bar{h}_{T-1}, A_{T-1}=a, S \in \mathcal{S}_{2T-1}$ under sampling from $P^0$ has the same distribution as $U_{T} \mid \bar{H}_{T-1}=\bar{h}_{T-1}, A_{T-1}=a$ under sampling from $Q^0$.

We have treated the data source $S$ as a random variable in our developments. In fact, similar theoretical guarantees to those that we have established can also be established for the case where there are $k$ datasets with fixed sample sizes $n_1, \ldots, n_k$. The results for this new case are similar, but, in the expressions for the gradients, $P(S\in\mathcal{S}_j)$ is replaced by $\{\sum_{i=1}^k n_i \mathbbm{1}(i\in\mathcal{S}_j)\} / (\sum_{i=1}^k n_i)$ and $P_j^0(\,\cdot\mid \bar{z}_{j-1},s)$ is replaced by the conditional probability of $Z_j$ given $\bar{Z}_{j-1}$ for a random draw from data source $s$. For estimation, a variable $S$ can be introduced into the dataset, which takes the value of the data source from which an observation was drawn, and then the estimation can proceed as described in this work, treating $S$ as though it were random. Under suitable regularity conditions, including that $n_i/\sum_{j=1}^k n_j$ converges to a positive constant for each $i\in [k]$, the resulting estimator can be shown to be semiparametrically efficient in the model where independent samples of deterministic sizes are drawn from each data source. Because the theoretical arguments needed to formalize this statement are nearly identical to those we have already given for the case where $S$ is random, they are omitted. Instead, we compare simulation results from two data generating mechanisms, one with sample sizes fixed and one with $S$ being random. The results are almost identical and are therefore omitted.

\end{appendices}

\bibliography{paper-ref}

\begin{thebibliography}{}

\bibitem[Angrist et~al., 1996]{angrist1996identification}
Angrist, J.~D., Imbens, G.~W., and Rubin, D.~B. (1996).
\newblock Identification of causal effects using instrumental variables.
\newblock {\em J. Am. Stat. Assoc.}, 91(434):444--455.

\bibitem[Athey et~al., 2019]{athey2019surrogate}
Athey, S., Chetty, R., Imbens, G.~W., and Kang, H. (2019).
\newblock The surrogate index: Combining short-term proxies to estimate
  long-term treatment effects more rapidly and precisely.
\newblock Technical report, National Bureau of Economic Research.

\bibitem[Bang and Robins, 2005]{bang2005doubly}
Bang, H. and Robins, J.~M. (2005).
\newblock Doubly robust estimation in missing data and causal inference models.
\newblock {\em Biometrics}, 61(4):962--973.

\bibitem[Bareinboim and Pearl, 2014]{bareinboim2014transportability}
Bareinboim, E. and Pearl, J. (2014).
\newblock Transportability from multiple environments with limited experiments:
  Completeness results.
\newblock {\em Advances in neural information processing systems}, 27:280--288.

\bibitem[Bareinboim and Pearl, 2016]{bareinboim2016causal}
Bareinboim, E. and Pearl, J. (2016).
\newblock Causal inference and the data-fusion problem.
\newblock {\em Proc. Natl. Acad. Sci.}, 113(27):7345--7352.

\bibitem[Bickel, 1982]{bickel1982adaptive}
Bickel, P.~J. (1982).
\newblock On adaptive estimation.
\newblock {\em Ann. Stat.}, pages 647--671.

\bibitem[Bickel et~al., 1993]{bickel1993efficient}
Bickel, P.~J., Klaassen, C.~A., Bickel, P.~J., Ritov, Y., Klaassen, J.,
  Wellner, J.~A., and Ritov, Y. (1993).
\newblock {\em Efficient and adaptive estimation for semiparametric models},
  volume~4.
\newblock Johns Hopkins University Press Baltimore.

\bibitem[Buchbinder et~al., 2008]{buchbinder2008efficacy}
Buchbinder, S.~P., Mehrotra, D.~V., Duerr, A., Fitzgerald, D.~W., Mogg, R., Li,
  D., Gilbert, P.~B., Lama, J.~R., Marmor, M., Del~Rio, C., et~al. (2008).
\newblock Efficacy assessment of a cell-mediated immunity hiv-1 vaccine (the
  step study): a double-blind, randomised, placebo-controlled, test-of-concept
  trial.
\newblock {\em Lancet}, 372(9653):1881--1893.

\bibitem[Chakrabortty, 2016]{chakrabortty2016robust}
Chakrabortty, A. (2016).
\newblock {\em Robust Semi-Parametric Inference in Semi-Supervised Settings}.
\newblock PhD thesis, Havard University.

\bibitem[Chapelle et~al., 2009]{chapelle2009semi}
Chapelle, O., Scholkopf, B., and Zien, A. (2009).
\newblock Semi-supervised learning (chapelle, o. et al., eds.; 2006)[book
  reviews].
\newblock {\em IEEE Trans. Neural Netw. Learn. Syst.}, 20(3):542--542.

\bibitem[Chernozhukov et~al., 2018]{chernozhukov2018double}
Chernozhukov, V., Chetverikov, D., Demirer, M., Duflo, E., Hansen, C., Newey,
  W., and Robins, J. (2018).
\newblock Double/debiased machine learning for treatment and structural
  parameters.

\bibitem[Churchyard et~al., 2011]{churchyard2011phase}
Churchyard, G.~J., Morgan, C., Adams, E., Hural, J., Graham, B.~S., Moodie, Z.,
  Grove, D., Gray, G., Bekker, L.-G., McElrath, M.~J., et~al. (2011).
\newblock A phase iia randomized clinical trial of a multiclade hiv-1 dna prime
  followed by a multiclade rad5 hiv-1 vaccine boost in healthy adults
  (hvtn204).
\newblock {\em PloS one}, 6(8):e21225.

\bibitem[Dahabreh and Hern{\'a}n, 2019]{dahabreh2019extending}
Dahabreh, I.~J. and Hern{\'a}n, M.~A. (2019).
\newblock Extending inferences from a randomized trial to a target population.
\newblock {\em Eur. J. Epidemiol.}, 34(8):719--722.

\bibitem[Dahabreh et~al., 2019]{dahabreh2019efficient}
Dahabreh, I.~J., Robertson, S.~E., Petito, L.~C., Hern{\'a}n, M.~A., and
  Steingrimsson, J.~A. (2019).
\newblock Efficient and robust methods for causally interpretable
  meta-analysis: transporting inferences from multiple randomized trials to a
  target population.
\newblock {\em arXiv preprint arXiv:1908.09230}.

\bibitem[Dong et~al., 2020]{dong2020integrative}
Dong, L., Yang, S., Wang, X., Zeng, D., and Cai, J. (2020).
\newblock Integrative analysis of randomized clinical trials with real world
  evidence studies.
\newblock {\em arXiv preprint arXiv:2003.01242}.

\bibitem[Dud{\'\i}k et~al., 2014]{dudik2014doubly}
Dud{\'\i}k, M., Erhan, D., Langford, J., and Li, L. (2014).
\newblock Doubly robust policy evaluation and optimization.
\newblock {\em Stat. Sci.}, 29(4):485--511.

\bibitem[Duong et~al., 2007]{duong2007ks}
Duong, T. et~al. (2007).
\newblock ks: Kernel density estimation and kernel discriminant analysis for
  multivariate data in r.
\newblock {\em J. Stat. Softw.}, 21(7):1--16.

\bibitem[Evans et~al., 2018]{evans2018doubly}
Evans, K., Sun, B., Robins, J., and Tchetgen, E. J.~T. (2018).
\newblock Doubly robust regression analysis for data fusion.
\newblock {\em arXiv preprint arXiv:1808.07309}.

\bibitem[Firpo, 2007]{firpo2007efficient}
Firpo, S. (2007).
\newblock Efficient semiparametric estimation of quantile treatment effects.
\newblock {\em Econometrica}, 75(1):259--276.

\bibitem[Gray et~al., 2011]{gray2011safety}
Gray, G., Allen, M., Moodie, Z., Churchyard, G., Bekker, L., Nchabeleng, M.,
  Mlisana, K., Metch, B., de~Bruyn, G., Latka, M., et~al. (2011).
\newblock Safety and efficacy assessment of the hvtn 503/phambili study: A
  double-blind randomized placebo-controlled test-of-concept study of a clade
  b-based hiv-1 vaccine in south africa.
\newblock {\em The Lancet infectious diseases}, 11(7):507.

\bibitem[Hansen, 1982]{hansen1982large}
Hansen, L.~P. (1982).
\newblock Large sample properties of generalized method of moments estimators.
\newblock {\em Econometrica}, pages 1029--1054.

\bibitem[Heitjan and Rubin, 1991]{heitjan1991ignorability}
Heitjan, D.~F. and Rubin, D.~B. (1991).
\newblock Ignorability and coarse data.
\newblock {\em Ann. Stat.}, pages 2244--2253.

\bibitem[Hern{\'a}n and Robins, 2020]{hernan2020causal}
Hern{\'a}n, M.~A. and Robins, J.~M. (2020).
\newblock Causal inference: what if.

\bibitem[Hern{\'a}n and VanderWeele, 2011]{hernan2011compound}
Hern{\'a}n, M.~A. and VanderWeele, T.~J. (2011).
\newblock Compound treatments and transportability of causal inference.
\newblock {\em Epidemiology (Cambridge, Mass.)}, 22(3):368.

\bibitem[Huang et~al., 2014]{huang2014immune}
Huang, Y., Duerr, A., Frahm, N., Zhang, L., Moodie, Z., De~Rosa, S., McElrath,
  M.~J., and Gilbert, P.~B. (2014).
\newblock Immune-correlates analysis of an hiv-1 vaccine efficacy trial reveals
  an association of nonspecific interferon-$\gamma$ secretion with increased
  hiv-1 infection risk: a cohort-based modeling study.
\newblock {\em PLoS One}, 9(11):e108631.

\bibitem[Ibragimov and Has'minskii, 1981]{ibragimov1981statistical}
Ibragimov, I.~A. and Has'minskii, R.~Z. (1981).
\newblock {\em Statistical estimation: asymptotic theory}, volume~16.
\newblock Springer Science \& Business Media.

\bibitem[Kallus et~al., 2020]{kallus2020optimal}
Kallus, N., Saito, Y., and Uehara, M. (2020).
\newblock Optimal off-policy evaluation from multiple logging policies.
\newblock {\em arXiv preprint arXiv:2010.11002}.

\bibitem[Kennedy, 2019]{kennedy2019nonparametric}
Kennedy, E.~H. (2019).
\newblock Nonparametric causal effects based on incremental propensity score
  interventions.
\newblock {\em J. Am. Stat. Assoc.}, 114(526):645--656.

\bibitem[Lanckriet et~al., 2004]{lanckriet2004statistical}
Lanckriet, G.~R., De~Bie, T., Cristianini, N., Jordan, M.~I., and Noble, W.~S.
  (2004).
\newblock A statistical framework for genomic data fusion.
\newblock {\em Bioinformatics}, 20(16):2626--2635.

\bibitem[Lu et~al., 2021]{lu2021you}
Lu, B., Ben-Michael, E., Feller, A., and Miratrix, L. (2021).
\newblock Is it who you are or where you are? accounting for compositional
  differences in cross-site treatment variation.
\newblock {\em arXiv preprint arXiv:2103.14765}.

\bibitem[Luedtke et~al., 2019]{luedtke2019omnibus}
Luedtke, A., Carone, M., and van~der Laan, M.~J. (2019).
\newblock An omnibus non-parametric test of equality in distribution for
  unknown functions.
\newblock {\em J. R. Stat. Soc.}, 81(1):75--99.

\bibitem[Mo et~al., 2020]{mo2020learning}
Mo, W., Qi, Z., and Liu, Y. (2020).
\newblock Learning optimal distributionally robust individualized treatment
  rules.
\newblock {\em J. Am. Stat. Assoc.}, pages 1--16.

\bibitem[Pearl and Bareinboim, 2011]{pearl2011transportability}
Pearl, J. and Bareinboim, E. (2011).
\newblock Transportability of causal and statistical relations: A formal
  approach.
\newblock {\em AAAI}, 25(1).

\bibitem[Pfanzagl, 1990]{pfanzagl1990estimation}
Pfanzagl, J. (1990).
\newblock Estimation in semiparametric models.
\newblock In {\em Estimation in Semiparametric Models}, pages 17--22. Springer.

\bibitem[Plotkin and Gilbert, 2012]{plotkin2012nomenclature}
Plotkin, S.~A. and Gilbert, P.~B. (2012).
\newblock Nomenclature for immune correlates of protection after vaccination.
\newblock {\em Clin. Infect. Dis.}, 54(11):1615--1617.

\bibitem[Polley and Van Der~Laan, 2010]{polley2010super}
Polley, E.~C. and Van Der~Laan, M.~J. (2010).
\newblock Super learner in prediction.

\bibitem[Robins et~al., 1994]{robins1994estimation}
Robins, J.~M., Rotnitzky, A., and Zhao, L.~P. (1994).
\newblock Estimation of regression coefficients when some regressors are not
  always observed.
\newblock {\em J. Am. Stat. Assoc.}, 89(427):846--866.

\bibitem[Rotnitzky and Smucler, 2019]{rotnitzky2019efficient}
Rotnitzky, A. and Smucler, E. (2019).
\newblock Efficient adjustment sets for population average treatment effect
  estimation in non-parametric causal graphical models.
\newblock {\em arXiv preprint arXiv:1912.00306}.

\bibitem[Rudolph and van~der Laan, 2017]{rudolph2017robust}
Rudolph, K.~E. and van~der Laan, M.~J. (2017).
\newblock Robust estimation of encouragement-design intervention effects
  transported across sites.
\newblock {\em J. R. Stat. Soc.}, 79(5):1509.

\bibitem[Stuart et~al., 2015]{stuart2015assessing}
Stuart, E.~A., Bradshaw, C.~P., and Leaf, P.~J. (2015).
\newblock Assessing the generalizability of randomized trial results to target
  populations.
\newblock {\em Prevention Science}, 16(3):475--485.

\bibitem[Stuart et~al., 2011]{stuart2011use}
Stuart, E.~A., Cole, S.~R., Bradshaw, C.~P., and Leaf, P.~J. (2011).
\newblock The use of propensity scores to assess the generalizability of
  results from randomized trials.
\newblock {\em J. R. Stat. Soc.}, 174(2):369--386.

\bibitem[Sun and Miao, 2018]{sun2018semiparametric}
Sun, B. and Miao, W. (2018).
\newblock On semiparametric instrumental variable estimation of average
  treatment effects through data fusion.
\newblock {\em arXiv preprint arXiv:1810.03353}.

\bibitem[Tsiatis, 2006]{tsiatis2006semiparametric}
Tsiatis, A.~A. (2006).
\newblock {\em Semiparametric theory and missing data}.
\newblock Springer.

\bibitem[van~der Laan and Gruber, 2012]{van2012targeted}
van~der Laan, M.~J. and Gruber, S. (2012).
\newblock Targeted minimum loss based estimation of causal effects of multiple
  time point interventions.
\newblock {\em Int. J. Biostat.}, 8(1).

\bibitem[Van~der Laan et~al., 2007]{van2007super}
Van~der Laan, M.~J., Polley, E.~C., and Hubbard, A.~E. (2007).
\newblock Super learner.
\newblock {\em Stat. Appl. Genet. Mol. Biol.}, 6(1).

\bibitem[Van~der Laan and Robins, 2003]{van2003unified}
Van~der Laan, M.~J. and Robins, J.~M. (2003).
\newblock {\em Unified methods for censored longitudinal data and causality}.
\newblock Springer Science \& Business Media.

\bibitem[Van~der Laan et~al., 2011]{van2011targeted}
Van~der Laan, M.~J., Rose, S., et~al. (2011).
\newblock {\em Targeted learning: causal inference for observational and
  experimental data}, volume~10.
\newblock Springer.

\bibitem[Van Der~Laan and Rubin, 2006]{van2006targeted}
Van Der~Laan, M.~J. and Rubin, D. (2006).
\newblock Targeted maximum likelihood learning.
\newblock {\em Int. J. Biostat.}, 2(1).

\bibitem[Van~der Vaart, 2000]{van2000asymptotic}
Van~der Vaart, A.~W. (2000).
\newblock {\em Asymptotic statistics}, volume~3.
\newblock Cambridge university press.

\bibitem[Van Der~Vaart et~al., 1996]{van1996weak}
Van Der~Vaart, A.~W., van~der Vaart, A., van~der Vaart, A.~W., and Wellner, J.
  (1996).
\newblock {\em Weak convergence and empirical processes: with applications to
  statistics}.
\newblock Springer Science \& Business Media.

\bibitem[Wedam et~al., 2020]{wedam2020fda}
Wedam, S., Fashoyin-Aje, L., Bloomquist, E., Tang, S., Sridhara, R., Goldberg,
  K.~B., Theoret, M.~R., Amiri-Kordestani, L., Pazdur, R., and Beaver, J.~A.
  (2020).
\newblock Fda approval summary: palbociclib for male patients with metastatic
  breast cancer.
\newblock {\em Clin. Cancer Res.}, 26(6):1208--1212.

\bibitem[Westling, 2021]{westling2021nonparametric}
Westling, T. (2021).
\newblock Nonparametric tests of the causal null with nondiscrete exposures.
\newblock {\em J. Am. Stat. Assoc.}, pages 1--12.

\bibitem[Westreich et~al., 2017]{westreich2017transportability}
Westreich, D., Edwards, J.~K., Lesko, C.~R., Stuart, E., and Cole, S.~R.
  (2017).
\newblock Transportability of trial results using inverse odds of sampling
  weights.
\newblock {\em Am. J. Epidemiol.}, 186(8):1010--1014.

\bibitem[Zheng and van~der Laan, 2011]{zheng2011cross}
Zheng, W. and van~der Laan, M.~J. (2011).
\newblock Cross-validated targeted minimum-loss-based estimation.
\newblock In {\em Targeted Learning}, pages 459--474. Springer.

\end{thebibliography}
\end{document}